\documentclass[journal,twocolumn,10pt]{IEEEtran} %journal

\usepackage{graphics} % for pdf, bitmapped graphics files
\usepackage{epsfig} % for postscript graphics files
\usepackage{mathptmx} %  It supersedes both the original times and the mathptm packages.
\usepackage{amsthm}
\usepackage{amsmath} % assumes amsmath package installed
\usepackage{amssymb}  % assumes amsmath package installed
\usepackage{cuted}
\usepackage{float}
\usepackage{color}
\usepackage{url}
\usepackage{bm}

\usepackage[ruled,vlined,linesnumbered,noresetcount]{algorithm2e}
\usepackage{tabularx}

\usepackage{subcaption,cite}
\usepackage[font=small]{caption}

\newtheorem{lemma}{Lemma}
\newtheorem{proposition}{Proposition}
\newtheorem{definition}{Definition}
\newtheorem{remark}{Remark}
\newtheorem{example}{Example}

\makeatletter
\def\ps@IEEEtitlepagestyle{%
  \def\@oddfoot{\mycopyrightnotice}%
  \def\@evenfoot{}%
}
\def\mycopyrightnotice{%
  \begin{minipage}{\textwidth}
  \centering \scriptsize
  Copyright~\copyright~2023 IEEE. Personal use of this material is permitted. Permission from IEEE must be obtained for all other uses, in any current or future media, including reprinting/republishing this material for advertising or promotional purposes, creating new collective works, for resale or redistribution to servers or lists, or reuse of any copyrighted component of this work in other works. 
  \end{minipage}
}
\makeatother

\begin{document}
\title{
%ZETAR: An Automated Design of Customized Incentive Mechanisms for Compliance Enhancement
%ZETAR: A Modeling and Computational Framework for Adaptive Audit and Compliance Policies to Incentivize Insiders
ZETAR: Modeling and Computational Design of Strategic and Adaptive Compliance Policies
%ZETAR: An Adaptive Computational Framework of Accountability-Driven and Information-Driven Incentive Design to Fill the Compliance Gap of Insiders
%ZETAR: An Adaptive Computational Framework of Audit and Strategic Recommendation Policy to Fill the Compliance Gap of Insiders
%in Cybersecurity  %(Automated/formal/resilient)  %Network Security Compliance
%ZETAR Incentive Program Design: Enhance insiders' Compliance by Zero-Trust Audit with Cognitive Recommendation for Corporate Security 
%achieves zero-trust audit with cognitive recommendation 
%Robust Incentive Design for Compliance in Insider Threats: combating insider negligence and instituting policy compliance}
%Robust Quantification of Tameability of Non-compliant Insiders
}

\author{Linan~Huang %~\IEEEmembership{Student Member,~IEEE,}
        and~Quanyan~Zhu,~\IEEEmembership{Member,~IEEE}% <-this % stops a space
\thanks{This paper has been accepted for publication in IEEE Transactions on Computational Social Systems}
\thanks{
L. Huang is  with the Beijing National Research Center for Information Science and Technology (BNRist), Tsinghua University, Beijing 100084, China. E-mail:huanglinan@mail.tsinghua.edu.cn}
\thanks{
 Q. Zhu is with the Department of Electrical and Computer Engineering, New York University, Brooklyn, NY, 11201, USA. E-mail:qz494@nyu.edu
%L. Huang and Q. Zhu are with the ECE Department, New York University, Brooklyn, NY, 11201, USA. E-mail:\{lh2328,qz494\}@nyu.edu
}
% <-this % stops a space
\thanks{
 This work was supported in part by the National
Science Foundation (NSF) under Grant ECCS-1847056, Grant BCS-2122060; in part by Shuimu Tsinghua Scholar Program 2022SM046; in part by International Postdoctoral Exchange Fellowship Program(Talent-Introduction Program) YJ20220128; and in part by National Natural Science Foundation of China No.62341109, No.62341106, Shanghai Municipal Science and Technology Major Project, and Tsinghua University Initiative Scientific Research Program.}
\thanks{Digital Object Identifier 10.1109/TCSS.2023.3323539}
}

\maketitle
\begin{abstract}
Compliance management plays an important role in mitigating insider threats. 
Incentive design is a proactive and non-invasive approach to achieving compliance by aligning an insider's incentive with the defender's security objective, \textcolor{black}{which motivates (rather than commands) an insider to act in the organization's interests.}  
Controlling insiders' incentives for population-level compliance is challenging because they are neither precisely known nor directly controllable. 
To this end, we develop ZETAR, a zero-trust audit and recommendation framework, to provide a quantitative approach to model insiders' incentives and design customized recommendation policies to improve their compliance. 
We formulate primal and dual convex programs to compute the optimal bespoke recommendation policies. 
We create the theoretical underpinning for understanding trust, compliance, and satisfaction, which leads to scoring mechanisms of how compliant and persuadable an insider is. 
After classifying insiders as malicious, self-interested, or amenable based on their incentive misalignment levels with the defender, we establish bespoke information disclosure principles for these insiders of different incentive categories. 
%Based on three categorizations of incentive misalignment between insiders and the defender, we establish information disclosure principles for malicious, self-interested, and amenable insiders. 
%and it leads to security insights, including fundamental limits of Completely Trustworthy (CT) recommendation,  the principle of compliance equivalency, and strategic information disclosure.   %the formally characterize the Completely Trustworthy (CT) recommendation policy and the Average Compliance Enhancement Level (ACEL). 
%The characterizations lead to CT feasibility, CT separability, fundamental limits, crucial principles (e.g., compliance equivalency under perturbations), security insights (e.g.,
%We discover that information disclosure policies depend on whether the insiders are amenable or malicious, and 
We identify the policy separability principle and the set convexity, which enable finite-step algorithms to efficiently learn the Completely Trustworthy (CT) policy set when insiders' incentives are unknown. 
Finally, we present a case study to corroborate the design. 
Our results show that ZETAR can well adapt to insiders with different risk and compliance attitudes and significantly improve compliance. 
Moreover, trustworthy recommendations can provably promote cyber hygiene and insiders' satisfaction. 
% Some good words and motivation from McKinsey Report on Cyber Risk (Perspectives on transforming cybersecurity). 
% \begin{itemize}
%     \item Start from compliance in cyber space. Check NIST Compliance.  https://reciprocity.com/frameworks/nist-framework-and-nist-compliance/
%     \item Some useful terms SEE https://reciprocity.com/compliance/
% \end{itemize}
\end{abstract}
%
%% Note that keywords are not normally used for peerreview papers.
\begin{IEEEkeywords}
Insider threat, information design, incentive mechanism,  zero-trust, incentive learning, Bayesian persuasion. %feedback %data-driven BP %inverse LP
%binary search
\end{IEEEkeywords}

% \fi
\IEEEpeerreviewmaketitle


%\allowdisplaybreaks

\section{Introduction}
%\IEEEPARstart{C}{yber} deception, which has been widely used by attackers for adversarial 
Insider threats in cyberspace refer to vulnerabilities and risks posed to an organization due to the misbehavior of its trusted but not trustworthy insiders, such as insiders, maintenance personnel, and system administrators. In $2021$, insider threats have caused around $39\%$ of breaches \cite{verizon2021}, which have resulted in significant operational disruptions, data loss, and reputation damage. 
%account for approximately the half of all security incidents can be traced back to nonmalicious behavior of insiders not complying with organizational security policies (Ernst & Young, 2017; PricewaterhouseCoopers, 2017).
%Those non-malicious violations 

\textcolor{black}{
Many organizations design insider threat countermeasures based on the presumption that insider security violations are either malicious or unintentional \cite{xx1}. 
This dichotomous perspective, however, overlooks the sizable middle ground of intentional yet non-malicious violations, which often emanate from self-interested insiders who place personal convenience or advantage above organizational security. 
The task of managing these non-compliance behaviors entails a strategic shift from straightforward deterrence and awareness training to the subtler approach of aligning insiders' incentives with the security objectives of the organization. 
By properly designing insiders' incentives, the organization can elicit proper behaviors in a \textit{proactive} and \textit{non-invasive} way; i.e., the insiders voluntarily reduce non-compliance and misbehavior.}  
%An incentive-based insider threat program can influence insiders' incentives and elicit proper behaviors that align with the organization's security objectives. 

Existing studies \cite{MooreTheCritical2016,theis2019common,Mitigation} have recognized the critical role of incentives in mitigating insider threats 
\textcolor{black}{and emphasized the integrated usage of various incentive methods, including monetary rewards, recognition and penalties, peer comparisons, and cultural cultivation. 
These studies have laid the empirical and experimental foundations for identifying the key incentive factors. 
However, there lacks a unified model to formally define incentives and systematically quantify the impact of those factors on the insiders' incentives and their resulting behaviors. Our work aims to address the above challenges of modeling the abstracted concept of incentives, characterizing the impacts of incentive factors, and ultimately developing a quantitative and automated design framework to guide the changes in  insiders' incentives to enhance compliance and mitigate insider threats.}   

To this end, we develop a modeling and computational framework called ZETAR (ZEro-Trust Audit with strategic Recommendation) for the defender \textcolor{black}{whose goal is to} improve the insiders' compliance and organizational cyber hygiene. 
\begin{figure}[h]
\centering
\includegraphics[width=.5 \textwidth]{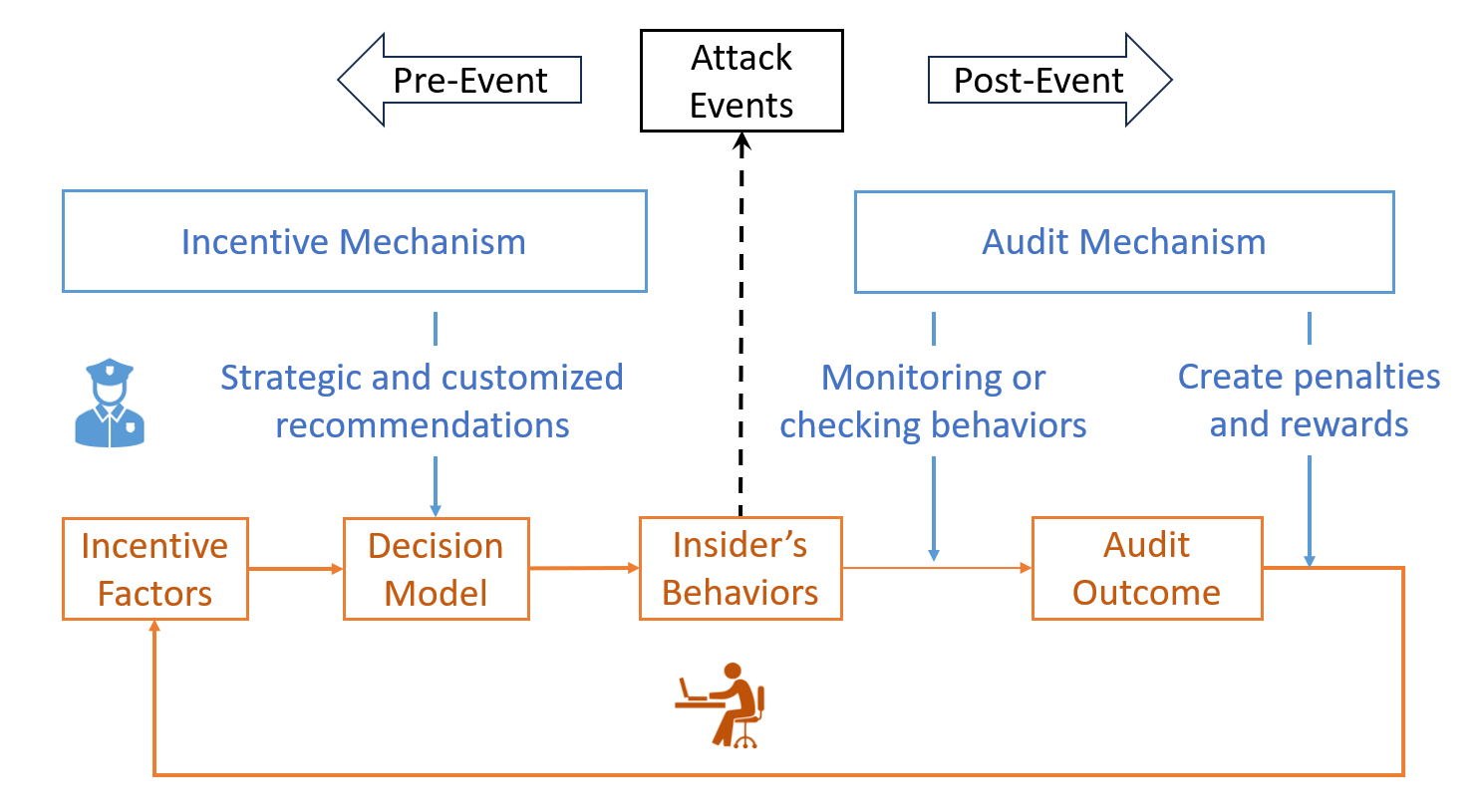}
\caption{ 
\textcolor{black}{ZETAR affects the incentives and behaviors of the insiders (depicted in orange) through a social-technical approach (highlighted in blue) that integrates the incentive and audit mechanisms.   
The technical approach of audits provides post-event remediation after the insiders have taken actions that could potentially lead to attack events. 
In contrast, the social approach of recommendation delivers a pre-event preemptive approach by actively shaping the insiders' motivations to enhance compliance.} 
}
\label{fig:ASrecommendation}
\end{figure}
As illustrated in Fig. \ref{fig:ASrecommendation}, ZETAR \textcolor{black}{consists of two functional modules for the defender:} the zero-trust audit and an incentive mechanism, \textcolor{black}{both targeting the insiders' decision-making loop.} 
The zero-trust audit mechanism assigns \textcolor{black}{no prior trust to the insiders and} inspects each insider's behaviors. 
\textcolor{black}{Based on the audit outcome, the audit mechanism changes the incentive factors (e.g., creating penalties or rewards), which indirectly affect the insider's behaviors through the insider's decision model.}  
%to attribute accountable insiders, penalize non-compliant behaviors, and reward compliant behaviors. 

\textcolor{black}{An audit enables the defender to detect non-compliant behaviors and implement  post-event remediation. 
Yet, given the vast landscape of non-malicious violations, pinpointing malicious infractions might strain the defender's time and budget. 
Therefore, ZETAR further introduces an incentive mechanism that recommends compliance policies for the insiders. 
The defender strategically designs these recommendations to be trustworthy and informative, based on the selected audit mechanism and each insider's incentive. The aim is to influence the decision-making process of the insider and promote compliant behaviors. 
In this way, the defender can preemptively curtail non-malicious breaches from self-motivated insiders, allocating limited defense resources more effectively towards malicious violations. 
As illustrated in Fig. \ref{fig:ScenarioDiag}, ZETAR tailors recommendations for each insider based on their individual incentives, ensuring the same audit mechanism can cater to diverse motivations and facilitate population-level compliance.} 

% \textcolor{black}{Since behavior inspection can be costly, the defender's  design of the audit mechanism should balance effectiveness and cost based on the insider's incentive, e.g., by reducing (resp. increasing) the frequency and granularity of the behavior inspection for compliant (resp. non-compliant) insiders.}  
% However, it is challenging to customize the audit mechanism \textcolor{black}{for each insider with a different incentive due to} time and budget constraints. 
% \textcolor{black}{For example, when having a security guard manually verify the desirable behavior of ``locking insiders' workstations after leaving'', the security guard usually conducts the same audit procedure and scoring system for a large area where insiders with different incentives sit. 
% Therefore, ZETAR further introduces a recommendation mechanism that recommends different compliance policies to each insider, as illustrated in Fig. \ref{fig:ScenarioDiag}.  
% In this way, the same audit mechanism can enhance the compliance of a large population of insiders with varied incentives and achieve population-level compliance.} 

\begin{figure}[h]
\centering
\includegraphics[width=.45 \textwidth]{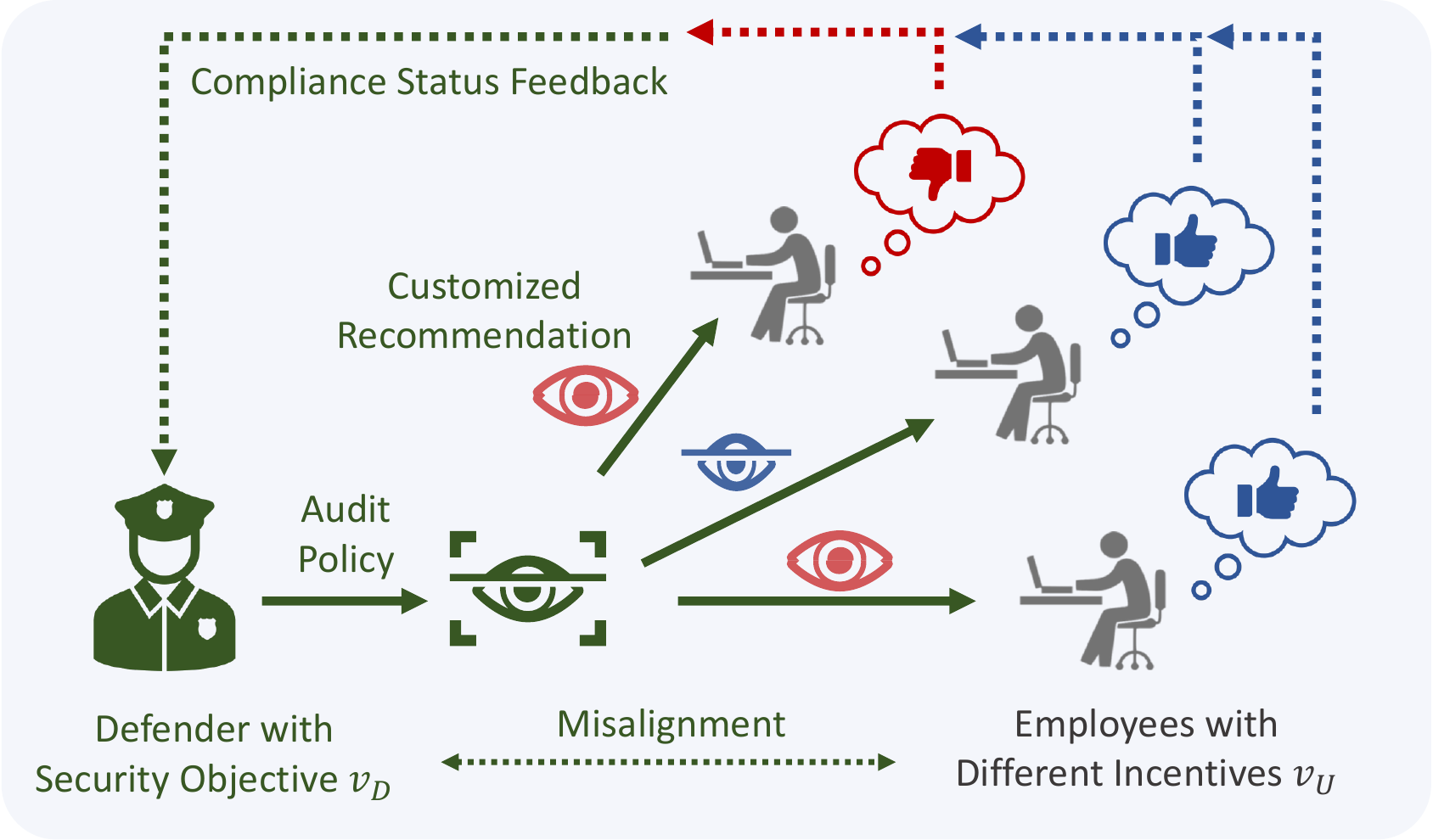}
\caption{ 
An illustration of the ZETAR feedback system: the defender of a corporate network audits insiders' compliance status and provides customized recommendations to insiders based on the learning of their incentives. ZETAR reduces the incentive misalignment between insiders and the defender. 
}
\label{fig:ScenarioDiag}
\end{figure}

%It is, however, challenging for a defender to \textit{model}, \textit{characterize}, and further \textit{influence} the insiders' incentives quantitatively. 
%First, there lacks formal models and design paradigms to measure the compliance gap and incentive misalignment. %resulting from the incentive misalignment between the insiders and the organization or the defender. 
%First, insiders' incentives are not directly \textit{controllable} but indirectly and restrictively affected through extrinsic factors such as reward, penalty, and information. 
%Systematically quantifying the impact of these extrinsic factors on the insiders' incentives and behaviors is the precondition of designing them to motivate (rather than command) an insider to act in the organization's interests. 
%Second, it is costly to \textit{customize} an incentive mechanism for a large population of insiders with varied incentives.  
%The automated and customized design of the incentive mechanism is a desideratum to achieve \textit{population-level compliance}. 

%\textcolor{black}{Besides the challenges of affecting insiders' incentives in a quantitative and customized manner,} 
Finally, since insiders' incentives are not directly \textit{observable},  
\textcolor{black}{we develop a feedback process for the defender to learn the insiders' incentives based on their compliance status, i.e., their behaviors under a selected recommendation mechanism,} as shown in Fig. \ref{fig:ScenarioDiag}. 
\textcolor{black}{Through a thorough theoretical characterization of the incentive design problem, we further identify the \textit{policy separability principle} and the \textit{convexity} of the feasible region, which enable us to develop efficient incentive learning algorithms that guarantee to converge in finite steps.} 

%it is instrumental for the learning process to be fast and adaptive to mitigate insider threats timely. 

\textcolor{black}{
\subsection{Summary of Contribution}
We summarize the contribution of this paper as follows.} 
\begin{itemize}
    \item \textbf{Modeling}: 
    \textcolor{black}{By integrating the social solution of strategic recommendations with the technical solution of audits,} we develop a social-technical paradigm called ZETAR that can reduce the compliance gap resulting from the incentive misalignment between the insiders and the organization (represented by the defender). 
    \item \textbf{Concept and Metrics}: 
    \textcolor{black}{We formally define abstracted concepts (e.g., compliance, trustworthiness, and satisfaction) to characterize their interrelationships and furnish transferable metrics. 
    Additionally, we give a formal delineation of an insider's incentive, categorizing them into amenable, malicious, and self-interested groups.} 
    \item \textbf{Computation and Analysis}: We formulate the design of ZETAR into a mathematical programming problem and \textcolor{black}{simplify the computation by recognizing} that Completely Trustworthy (CT) recommendation policies are sufficient for the optimal compliance improvement. 
    \textcolor{black}{The analysis of the problem enriches both the theory of compliance and incentive design. It also offers security insights and guidelines, including principles of information disclosure to insiders with different incentives.} 
    \item \textbf{Algorithm Design}: 
    \textcolor{black}{Compared to classical feedback learning methods that are universal yet inefficient,} we develop efficient feedback algorithms with a binary search procedure \textcolor{black}{that adequately exploit the structure of the problem. The algorithm itself further contributes to the field of learning within the Bayesian persuasion framework \cite{kamenica2011bayesian}.} 
    \item \textbf{Application and Validation}: We present a case study to corroborate the effectiveness of ZETAR in improving compliance for insiders with different risk and compliance attitudes. We decompose an insider's incentives into extrinsic and intrinsic ones. 
    The results show that ZETAR can well adapt to different types of insiders and achieve a structural improvement of compliance for risk-averse insiders. 
    Under binary action sets, we identify the \textit{belief thresholds} of these insiders to determine whether to take compliant or non-compliant actions. 
\end{itemize}

  %We develop ZETAR, an \textit{automated} and \textit{adaptive} paradigm to design a \textit{provably-compliant} incentive mechanism. 
%We prove that CT recommendation policies achieve a \textit{win-win situation}  by promoting insiders' satisfaction and improving the defender's security level. 
%We interpret the dual problem from the perspective of insiders who aim to minimize their efforts to satisfy the corporate security objective. 
%Then, we characterize the impact of linear transformations of incentives and security objectives, the geometric representation of CT policy sets, and the optimal level of compliance improvement. 

\subsection{Notations and Organization of the Paper}
We use the pronoun `he' for an insider and `she' for the defender \textcolor{black}{throughout this paper}. 
Calligraphic letter $\mathcal{Y}$ defines a set, and $\Delta \mathcal{Y}$ represents the set of probability distributions over set $\mathcal{Y}$. 
Superscripts and subscripts represent elements and different categories, respectively. 
\textcolor{black}{The rest of the paper is organized as follows.} 
We present the system model and computational framework of ZETAR in Sections \ref{sec:model} and \ref{sec:Computational Framework of ZETAR}, respectively. 
Section \ref{sec:characterize two} characterizes trustworthiness and compliance. 
These characterizations lead to efficient learning algorithms in Section \ref{sec:Model-Agnostic} when insiders' incentives are unknown to the defender. 
Section \ref{sec:case study} presents a case study, and Section \ref{sec:conclusion} concludes the paper.

\section{Related Works}

\subsection{Insider Threat Mitigation} %Insider Threat Mitigation %Contributing Factors and Mitigation Strategies for Insider Threat
%\subsubsection{Unintentional and intentional Insider Threats}
%Based on whether the insider misbehaves with or without malicious intent, 
Insider threats \textcolor{black}{are usually} classified into unintentional or intentional ones. For intentional insider threats, the authors in \cite{MooreTheCritical2016,theis2019common,Mitigation} have recognized incentives as a leading factor and incentive designs as a class of promising mitigation strategies. 
\textcolor{black}{
In their recent study \cite{Huang_Zhu_2021}, the authors employ defensive deception as a strategy to differentiate between non-compliant insiders and attacks, thereby mitigating insider threats. 
}

For unintentional insider threats, the authors in \cite{team2013unintentional,greitzer2014unintentional} have identified three contributing factors (i.e., organizational, human, and demographic factors) and a set of proactive mitigation strategies (e.g., awareness training, relieving workload pressure, and security tools to help overcome user errors). 
\textcolor{black}{
Many unintentional insider threat incidents result from various human cognitive vulnerabilities, such as limited rationality and attention \cite{Huang_Zhu_2023}. For example, most employees have the security knowledge to recognize a phishing email yet still fall victim to it due to a lack of attention \cite{dhamija2006phishing}. 
Recent works \cite{Huang_Zhu_2022,Huang_Jia_Balcetis_Zhu_2022} have attempted to mitigate the above attentional vulnerability by integrating real-time human biometric data (e.g., eye-tracking data and EEG) with AI and learning technologies.} 

\subsection{\textcolor{black}{Security Policy Compliance}}
\textcolor{black}{The existing research on security policy compliance has largely fallen into two distinct categories. 
The first category exploits technical methods, including deep learning \cite{yuan2021deep}, graph-based approaches \cite{eberle2010insider}, and game theory \cite{casey2015compliance}, to detect and manage policy violations. 
The second category focuses on human and social aspects to identify the critical factors for human compliance decisions (see, e.g., \cite{cram2019seeing}). 
While both categories lay solid foundations for security policy compliance, they fall short in offering} a holistic design of technical (e.g., audit and access control) and social (e.g., security policies and positive organizational culture) solutions \cite{hunker2011insiders,greitzer2010combining,saxena2020impact}. 
ZETAR provides a unique, quantitative, and automated framework to provide strategic and customized recommendation policies to elicit compliant behaviors from insiders. 
%Moreover, ZETAR quantifies the risk of non-compliance, which enables the provably-correct optimization of mitigation strategies. 

\subsection{Incentive Mechanisms for Cyber-Physical Security}% and information design

%\subsubsection{Incentive Design through Information}

There is a rich literature on designing incentive mechanisms to enhance Cyber-Physical System (CPS) security \cite{7265043,zhu2012guidex}. 
%, efficiency \cite{7180387}, and privacy \cite{8355763} of Cyber-Physical System (CPS).  %and privacy \cite{8355763}. %zhu2012guidex
%intrusion detection networks \cite{zhu2012guidex} and insurance \cite{zhang2017bi}. 
%A rich literature has used incentive mechanisms to enhance security \cite{7265043} for cyber-physical systems. 
%, efficiency \cite{7180387}, privacy \cite{8355763} of cyber-physical systems. 
%Compared to these incentive mechanisms that focus on design payoff rules and allocation rules, manipulating information \cite{HORAK2019101579,huang2020dynamic,huang2021dynamic} has become a critical complement to change users' or attackers' beliefs and incentivize these agents to behave in the defender's favor. 
%huang2020game
%In these works, the defender keeps their strategies secrete to affect the agents' beliefs. 
%Compared to the covert information manipulation, ZETAR implements a transparent and overt recommendation strategy to align insiders' incentives and meanwhile foster a culture of trust among insiders.  
%Aiming to foster a culture of trust among insiders, ZETAR implements a transparent and overt recommendation strategy to align insiders' incentives. 
Informational control is relatively less explored \textcolor{black}{compared to the classical design of payment and allocation rules \cite{myerson1989mechanism}.} It provides an affordable, scalable, \textcolor{black}{flexible} and complementary way to change agents' beliefs for compliance. 
\textcolor{black}{In the content of this paper, it is convenient to dynamically tailor the recommendation mechanism, operating at the information level, to insiders with different incentives. However, it is much more costly to customize or change the audit mechanism, which needs actual implementation.}

\textcolor{black}{Among the recent works that focus on the strategical design and control of information} \cite{HORAK2019101579,huang2020dynamic,huang2021dynamic}, the defender often keeps their strategies \textit{covert} to influence the agents’ reasoning.  In contrast, ZETAR employs a \textit{transparent} and \textit{overt} recommendation strategy to align with insiders’ incentives, \textcolor{black}{which can potentially foster a culture of trust between insiders and the organization.}  

%\subsubsection{Uncertainty and Unknown in Incentive Design}
%To deal with the ubiquitous uncertainty and unknown in incentive design, previous works have focused on worst case optimization \cite{dworczak2020preparing}, distributionally robust \cite{koccyiugit2020distributionally}, expectation \cite{restuccia2018incentme}, and learning \cite{zhan2020incentive}. In this work, ZETAR leverages feedback and the linear structure of the solution to develop an efficient incentive-learning algorithm that is proven to converge in a small number of steps. 
Uncertainties in incentive designs are challenging to deal with. Previous works have taken several approaches to address this issue, including robust methods \cite{xu2021biobjective}, Bayesian methods \cite{zou2015incentive}, and
learning methods \cite{zhan2020incentive}. In this work, ZETAR leverages the feedback of insiders' compliance status and the structures of the solution to develop efficient incentive learning algorithms that are provably convergent in finite steps.

\section{System Model of ZETAR}
\label{sec:model}

\textcolor{black}{As illustrated in Fig. \ref{fig:ScenarioDiag},} ZETAR provides customized recommendation designs for insiders with different incentives, \textcolor{black}{based on the same audit policy.}   
Each design involves two players, the defender $D$ and an insider $U$. 
%the defender $D$ (hereafter `she') and an insider $U$ (hereafter `he'). 
%{\bf We do not need to talk about SaaS. SaaS is an example or a special case.}
%In this work, the defender can assess resorts to \textit{Security-as-a-Service (SaaS)} providers to evaluate the organization's security posture in Section \ref{sec:SP} and audit insiders' actions in Section \ref{sec:audit and action}. 
The defender can assess the organization's security posture \textcolor{black}{(explained in Section \ref{sec:SP})} and audit insiders' behaviors either by himself or through a third-party service provider \textcolor{black}{(detailed in Section \ref{sec:audit and action})}. % (e.g., \textit{Security-as-a-Service (SaaS)} providers 
\textcolor{black}{The defender receives audit outcomes detailing each insider's compliance status and enhances compliance through effective incentive management and strategic recommendations.} 
 % non-compliant actions after they are taken. Then, the defender applies post-event mitigation methods. 
%To \textit{proactively} elicit compliant behaviors of the insiders, the defender creates incentive mechanisms and strategic recommendations in Section \ref{sec:Recommendation}. %improve the compliance. 

\subsection{An Organization's Security Posture}
\label{sec:SP}
Security Posture (SP) reflects an enterprise's overall cybersecurity strength and capacities to deter, detect, and respond to the dynamic threat landscape \cite{dukes2015committee}. 
Based on different scoring and categorization methodologies \cite{al2020gosafe,bahuguna2020country}, SP can be classified into finite categories (e.g., high-risk SP and low-risk SP). 
In this work, we consider a finite number of $J$ SP categories that compose the set $\mathcal{Y}:=\{y^j\}_{j\in \mathcal{J}}$, where $\mathcal{J}:=\{1,\cdots,J\}$. 

The current SP can be assessed based on penetration tests, honeypots, and alert analysis \cite{zhan2020nsaps}. 
Since an organization's SP changes probabilistically based on the dynamic behaviors of attackers, users, and defenders, we let $b_Y(y^j)\in [0,1]$ denote the probability of the organization to be in the state of SP  $y^j\in \mathcal{Y}, \forall j\in \mathcal{J}$. 
%an SP $y^j\in \mathcal{Y}$ happens with a probability $b_Y(y^j)\in [0,1], \forall j\in \mathcal{J}$. 
With a slight abuse of notation, we define $b_Y\in \Delta \mathcal{Y}$ as the probability distribution over $\mathcal{Y}$. 
% SP is affected by many factors (e.g., attack frequency and destructiveness, the size of the attack surface, and the resiliency of the existing defense methods) that depend on the changing behaviors of attackers, users, and defenders, respectively. 
% Thus, SP of the same organization changes probabilistically and lead to different SP states (e.g., high-risk SP and low-risk SP). 
%the audit team and the defender are not the same
%Since these contributing factors depend on the changing behaviors of attackers, users, and defenders, SP is random and needs to be evaluated repeatedly on a weekly or monthly basis. 
%The security status of an enterprise’s networks, information, and systems based on information security resources (e.g., people, hardware, software, policies) and capabilities in place to manage the defense of the enterprise and to react as the situation changes. 

\subsection{Zero-Trust Audit Policy}%Audit-as-a-Service
\label{sec:audit and action}

\textcolor{black}{The defender of an organization follows prescribed security rules to improve the organization's cyber hygiene.} 
These rules can be set and audited by regulatory agencies, cyber insurance providers, or the organization itself. 
%
%{\bf AUDIT FIRM CAN BE IN THE EXAMPLE STORY. HERE, WE WILL MAKE IT GENERAL. MOVE THE NEXT PART TO AND INTEGRATE INTO THE EXAMPLE.}For example, 
%the defender can resort to a cybersecurity audit firm for Audit as a Service (AaaS) to check the compliance of the organization. 
%
\textcolor{black}{In accordance with the zero-trust security principle (e.g., see \cite{rose2020zero}), every insider within the organization is subject to audit and not inherently trusted.}

%the audit is applied to all insiders in an organization without a prior trust assignment. 
%the assumption of trust. %This practice aligns with the zero-trust security principle (e.g., see \cite{rose2020zero}). 
Consider a finite set  of $I$ Audit Schemes (ASs), denoted by $\mathcal{X}$. 
Each AS contains the entire audit procedure. For example, for a given AS $x\in \mathcal{X}$, the audit involves the steps of (1) monitoring and checking the insider's behaviors, %action $a\in \mathcal{A}$, 
%{\bf Here a is confusing. a means the action of choosing to click or not, or means an abstract action to comply or not. We need to define it more properly. There are two definitions. One is here; another is in Example 1.} 
(2) assigning a compliance score to the insider,
%assessing the insiders' compliance score based on the action, 
and (3) informing (the defender) of the compliance score and action. A different AS $x'\in \mathcal{X}, x'\neq x$, can vary in the monitoring or scoring scheme. % Readers can refer to Example \ref{example} for an example of the action and the AS. 
%For example, a stringent AS $x^{sa} \in \mathcal{X}$ can assess the compliance score of an insider by monitoring the compliance to all the rules, while a tolerant AS $x^{tr}\in\mathcal{X}$ only assesses the compliance score based on a subset of rules. 

\textcolor{black}{The ASs are prescribed based on the SP of the organization.} Let $\psi\in \Psi:\mathcal{Y}\mapsto \Delta \mathcal{X}$ denote the audit policy, which probabilistically determines an AS $x\in \mathcal{X}$, where $|\mathcal{X}|=I$. The probability of choosing $x\in \mathcal{X}$ given the SP $y\in \mathcal{Y}$ is thus given by $\psi(x|y)\in [0,1]$. The outcome of the audit scheme is used by the defender to create penalties or rewards for the insiders to shape their incentives and elicit compliant behaviors. Hence, the incentives of the insiders and the security objective of the defender are naturally dependent on the prescribed audit scheme. They will be further elaborated on in Section \ref{sec:utilities}.

\textcolor{black}{We characterize the system model by system-level concepts, including SP, audit policy, and AS, that are designed to be versatile, allowing for further specification to cater to diverse scenarios of insider threat mitigation. 
We provide the following example to offer intuitive insights into the above mathematical formulation and demonstrate the practical applicability of our system model. 
The example typifies the system model's use in the stochastic audit of essential security rules.} 

\begin{example}[\textbf{Stochastic \textcolor{black}{Audits} of Critical Security Rules}]%for Security Rule Violations
\label{example}
%Non-malicious insiders can fail to adhere to security rules due to convenience or negligence; e.g., an insider may use a public network to transmit confidential data to save time. 
%prioritize security over convenience
Consider an organization that needs to comply with a finite set of $H$ critical security rules, denoted by $\mathcal{H}:=\{1,\cdots,H\}$, set by a U.S. regulatory agency. The rules entail proper behaviors for remote access, user accounts, and backups \cite{sarkar2010assessing}. 
The compliance of an insider is monitored by checking each rule. Its outcome, denoted by $o^h$, also known as the compliance status concerning rule $h\in \mathcal{H}$, is either full, partial, or no compliance, denoted by $\iota_{f}$, $\iota_{p}$, and $\iota_{n}$, respectively.
%We use $e_{fc}$, $e_{pc}$, and $e_{nc}$ to represent full, partial, and no compliance, respectively. 
%Three compliance statuses concerning rule $h\in \mathcal{H}$ are full, partial, and no compliance. They are represented by $e_{fc}$, $e_{pc}$, and $e_{nc}$, respectively.
%Let $e^h\in \mathcal{E}^h:=\{e_{fc},e_{pc},e_{nc}\}$ be an insider's compliance status concerning rule $h\in \mathcal{H}$ and 
By lumping the outcomes into a vector, we let  $a=(o^1,\cdots,o^H) \in \mathcal{A}:=\prod_{h\in \mathcal{H}} \mathcal{O}^h$, where $\mathcal{O}^h=\{\iota_{f},\iota_{p},\iota_{n}\}$, be the consolidated compliance status of an insider. An insider can choose his consolidated compliance status $a\in \mathcal{A}$ based on his incentives.
%is evaluated based on his behaviors. 
%Let $e^h\in \mathcal{E}^h:=\{e_{fc},e_{pc},e_{nc}\}$ be an insider's compliance status concerning rule $h\in \mathcal{H}$; $e^h$ takes finite values to represent full, partial, and no compliance. 
%and it is evaluated based on the behaviors of the insider. 
%For example, if rule $h\in \mathcal{H}$ requires the insider to perform full backups weekly and incremental backups daily, then $\mathcal{E}^h$ includes three compliance statuses, i.e., complete compliance, complete violation, or partial compliance (e.g., performing full backups less frequently).  
%The consolidated evaluation of the insider's compliance status is denoted by  $a=(e^1,\cdots,e^H) \in \mathcal{A}:=\prod_{h\in \mathcal{H}} \mathcal{E}^h$. 
%The compliance score of the insider can be assessed based on $a$, which can indicate the trustworthiness of the insider.  
%insider's non-compliance creates vulnerability surface that can be exploited by attackers. 
%To push changes in insiders' non-compliant behaviors, the defender can audit the violation of security rules at varying costs via automated violation detection system or human inspection. 
%Human inspection is more accurate yet also more costly.  
%Based on the defender's resource limitation, e.g., the number of human inspectors, the defender chooses an inspection scheme $x$ from the following $H+1$ potential schemes $\mathcal{X}=\{x^1,\cdots,x^{H+1}\}$. 
Let $\mathcal{X}=\{x^1,\cdots,x^{H+1}\}$ be the set of $I=H+1$ ASs. 
Each AS follows the same procedure of checking the compliance of the $H$ rules to report an insider's compliance status but differs in assessing compliance scores. % evaluation of the .  %mechanism of the action. 
AS $x^h\in \mathcal{X}, h=1, \cdots, H$, yields a compliance score $r^h\in \mathbb{R}$ solely based on the outcome $o^h\in \mathcal{O}^h$, i.e., $r^h=g^h(o^h)$, where $g^h:\mathcal{O}^h\rightarrow \mathbb{R}$ is the scoring function associated with AS $x^h$.
%evaluates the compliance based on the compliance status of rule $h\in \mathcal{H}$,
AS $x^{H+1}\in \mathcal{X}$ uses the outcomes associated with all the rules for the assessment, i.e., $r^{H+1}=g^{H+1}(a)$, where $g^{H+1}:\mathcal{A} \rightarrow \mathbb{R}$ is the scoring function associated with AS $x^{H+1}$.
It is clear that $x^{H+1}$ is the most stringent AS among all. The score is used as the criterion to penalize insiders and thus affects their incentives that will be formally defined in Section \ref{sec:utilities}.  
%{\bf NOT CLEAR WHY ??}Stringent ASs, e.g., $x^{H+1}$, intensify the confrontation between insiders and the defender and incurs a confrontation cost to the defender. 
%randomly selects $h$ rules (resp. no rules) from $\mathcal{H}$ to inspect. 
%the defender chooses to randomly inspect $h$ rules in scheme $x^h, h\in \{1,\cdots,H\}$ and inspect no rules in scheme $x^{H+1}$.  
%assigns human inspectors to audit $h$ rules randomly from the total $H$ rules. In scheme $x^{H+1}$, the defender switches to an automated violation detection system which is less accurate but also less costly. 
%The defender needs to choose AS strategically as inspecting more rules reduces cyber risks at a higher cost. 
%yet increases cost. %tradeoff of security and cost
\end{example}

 %Shown in Fig. \ref{fig:ScenarioDiag},  
%In Example \ref{example}, an audit policy can prescribe stringent audit $x^{H+1} \in \mathcal{X}$ more frequently under a high-risk SP than a low-risk SP. 
In Example \ref{example}, the audit policy $\psi\in \Psi$ is chosen based on a predetermined level of tolerance. 
A proper level of tolerance trades off between the organization's security and the compliance cost resulting from the overhead and the lack of flexibility 
\cite{MooreTheCritical2016}. 
An appropriate choice of tolerance depends on the SP; e.g., an audit policy can prescribe the stringent audit $x^{H+1} \in \mathcal{X}$ at a higher rate under high-risk SP than low-risk SP. 
% In Example \ref{example},  the audit company may 
% %In Example \ref{example}, defenders who are underfunded or understaffed need to 
% reduce $\psi(x^{H}|y^{lr})$ (i.e., the frequency of stringent audit $x^{H}\in \mathcal{X}$ at low-risk SP $y^{lr}\in \mathcal{Y}$) to save budgets and human resources for a stringent audit at the high-risk SP. 
%Since ZETAR adopts zero-trust principles and audit all insiders, 
We assume that the audit policy $\psi$ set by the organization or regulatory agencies remains the same for a sufficiently long time, making the policy more implementable and agreeable to insiders over the entire corporate network \cite{theis2019common}. 

%For a given audit scheme, an insider can take actions to comply or not apply. 
Let $\mathcal{A}$ denote the set (with cardinality $K$) of an insider's actions. In Example \ref{example}, an action $a\in \mathcal{A}$ is referred to as the rule compliance profile, i.e., $a=(o^1, \cdots, o^H)$, which is a result of the insiders' behaviors, including keystrokes, full application contents (e.g., email, chat, data import, and data export), and screen captures \cite{spooner2018navigating}. In the case where there is one rule, $i.e., H=1$, $\mathcal{A}$ is reduced to an action set that comprises three actions: full, partial, and no compliance. The insider's actions are monitored by the AS, and the defender is informed of the insider's compliance status to nudge compliant behaviors. 
\textcolor{black}{Note that ZETAR does not aim to design insiders' incentives to prevent all non-compliant behaviors, but focuses on the actions that can be audited with full fidelity and social acceptance. 
For instance, audits should not encompass sensitive behavioral data, such as the duration spent in restrooms. 
Consequently, $\mathcal{A}$ encompasses only those behaviors whose compliance status can be comprehensively audited without infringing on privacy norms or regulations.}

\subsection{Utilities of the Defender and Insiders}
\label{sec:utilities}
In the past five years, financially motivated insider threats have continued to be the most common motive of threat actors \cite{verizon2021}. 
We define utility functions $v_p:\mathcal{Y}\times \mathcal{X}\times\mathcal{A} \mapsto \mathbb{R}$ for $p\in \{U,D\}$ to capture an insider's incentive and the defender's security objective, respectively. 
%The utility functions of the defender and an insider, denoted by  $v_p:\mathcal{Y}\times \mathcal{X}\times\mathcal{A} \mapsto \mathbb{R}, p\in \{D,U\}$, depend on the current SP, the implemented AS, and the insider's action. 
The defender's utility $v_D(y,x,a)$ assesses the impact of an insider's action $a\in \mathcal{A}$ on network security under the SP $y\in \mathcal{Y}$ and AS $x\in \mathcal{X}$.   
Since the impact is assessed subjectively by the defender,  $v_D$ represents the defender's security objective. 
For example,  under life-critical scenarios with zero tolerance to non-compliance, the defender can assign $v_D(y,x,a^{ic})=-\infty, \forall y\in \mathcal{Y}, x\in\mathcal{X}$. 
 %In Fig. \ref{fig:overviewDiag}, we assign a high value to $v_D(y^{hr},x^{ta},a^{ic})$ if it is more beneficial for the defender to find and fix the vulnerability caused by the insider's non-compliant action $a^{ic}$ with random audit $x^{ta}$ under high-risk SP $y^{hr}$.    
%The utility function $v_D$ characterizes the defender's security objectives, and is used in Section \ref{sec:Recommendation} to obtain the optimal design of ZETAR. 

An insider's utility $v_U(y,x,a)$ models his extrinsic and intrinsic incentives to take action $a\in \mathcal{A}$ under SP $y\in \mathcal{Y}$ and AS $x\in \mathcal{X}$. 
On the one hand, $v_U$ can incorporate monetary incentives (through reward and recognition) and disincentives (through penalty and punishment) from the defender. 
%We can shape insider's extrinsic incentive by changing the amount of reward and penalty directly. 
On the other hand, $v_U$ can represent an insider's proclivity for compliant behaviors.
%internal inclination not to misbehave. 
%We can shape this type of utility by security awareness training and fostering a culture of responsibility for the organization, the work, and co-workers. %cyber hygiene
%Since it is hard to measure directly, we use RL to learn how to foster the culture. 
\textcolor{black}{The utility function may also capture other factors, including different risk attitudes.} 
Readers can refer to Section \ref{sec:casestudy_vD} and \ref{sec:casestudy_vU} for an example of $v_D$ and $v_U$, respectively. 

\textcolor{black}{The utility functions aptly capture the essence of the incentive design problem in Section \ref{sec:Recommendation}. They also provide the proper level of abstraction for incentive modeling and the derivation of theoretical insights. 
The detailed form of the utility functions $v_U$ and $v_D$ is beyond the scope of this paper.}

%Then, the implemented AS under SP $y\in \mathcal{Y}$ is a realization of the probability distribution $\psi(\cdot|y)$. 
%Since the insiders can repeatedly observe the implemented ASs and verify the truthfulness of the declared audit policy $\psi$

%However, if the SP is of low risk, ...
%as shown in Fig. \ref{fig:overviewDiag}. 
%For each security posture $x^n$, the defender choose to inspect a subset of $H$ policy, denoted by $\mathcal{H}_n\subseteq \mathcal{H}$. 
%The audit then bring reward and penalty based on the insider's action. Thus, the user's utility $v_U(x,a)$ depends on $x$ and $a$. 
%The defender utility $v_D(x,a)$ is affected by the security posture, the audit cost, and the user's action (e.g., it is beneficial to find non-compliance and then fix it). 
%represents a class of attack scenarios, including phishing, password attacks, and ransomware, which happens with a given probability $b_X(x)$. %then its realization cannot be known exactly. 

\subsection{Strategic Recommendations for Customized Compliance}
\label{sec:Recommendation}
%Sections \ref{sec:SP} and \ref{sec:audit and action} have introduced the organization's SP, the defender's AS, and an insider's action in Fig. 

Following Section \ref{sec:audit and action}, the audit policy  $\psi$ remains unchanged once determined. 
Since it is challenging for a fixed audit policy to achieve optimal inspection outcomes for insiders with different compliance requirements and incentives (as will be further elaborated in Section \ref{sec:insider's Initial Compliance}), the defender designs a customized incentive mechanism, encompassing a recommendation policy $\pi\in \Pi$ that results in a recommendation  $s\in \mathcal{S}$, \textcolor{black}{as shown in the third stage of Fig. \ref{fig:overviewDiag}.
%The basic intuition behind the recommendation mechanism is to strategically reveal information of the selected AS so that self-interested insiders 
}  
%Thus, we complement the fixed audit policy with a customized recommendation mechanism shown in the third stage of Fig. \ref{fig:overviewDiag}. 

 \begin{figure}[h]
\centering
\includegraphics[width=.5 \textwidth]{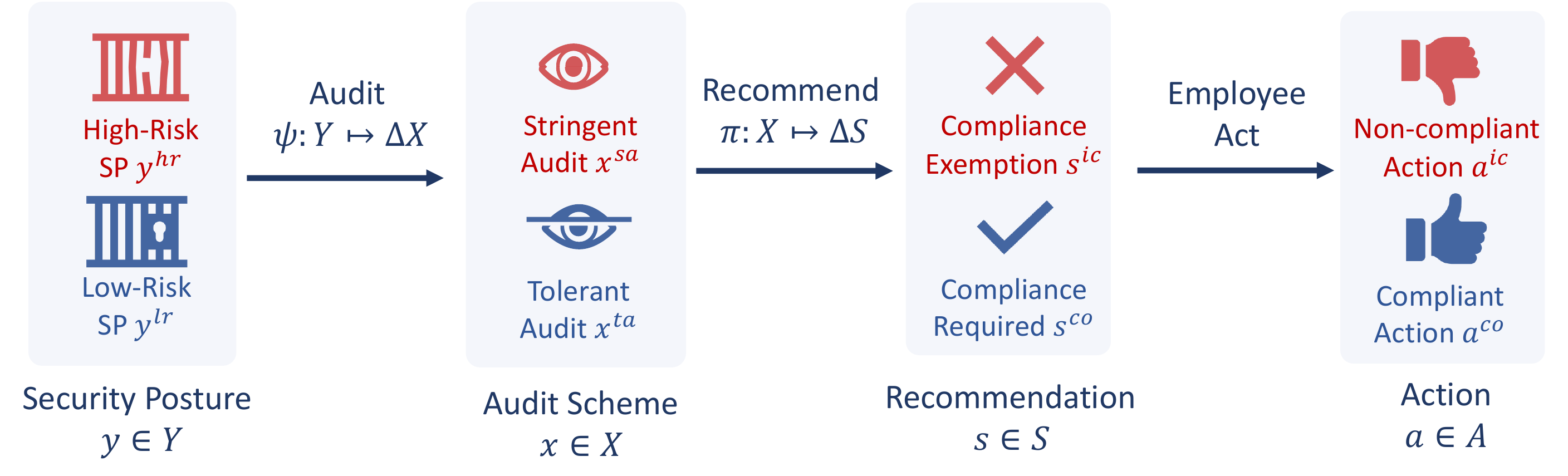}
\caption{ 
The timeline of ZETAR to enhance insiders' compliance in corporate networks. %Overview diagram 
The defender is informed of the SP and the audit outcomes of all insiders' behaviors.
The defender designs a recommendation policy $\pi\in \Pi$ to improve compliance. 
}
\label{fig:overviewDiag}
\end{figure}

\subsubsection{Information Structure and Timeline}
\label{sec:Information Structure and Dynamics}
As stated in \cite{MooreTheCritical2016}, transparent criteria for organizational policies can create a culture of trust and consequently serve as a positive incentive to reduce non-compliance. 
Thus, we assume that the sets $\mathcal{Y}, \mathcal{X}, \mathcal{S}, \mathcal{A}$, the prior statistics $b_Y\in \Delta \mathcal{Y}$, the audit policy $\psi\in \Psi$, and the recommendation policy $\pi\in \Pi$ are common knowledge. 
%commonly known by insiders. 
%Thus, the defender truthfully reports the audit policy $\psi\in \Psi$ to all insiders; i.e., $\psi(x|y)$, the probability of taking AS $x\in \mathcal{X}$ given the SP $y\in \mathcal{Y}$, is common knowledge. 
Since the organization's SP changes randomly as shown in Section \ref{sec:SP}, the SP is assessed repeatedly on a weekly or monthly basis, and it is unknown to insiders. 

Fig. \ref{fig:overviewDiag} illustrates the timeline of ZETAR as follows. 
Given the current SP $y\in \mathcal{Y}$ and the audit policy $\psi\in \Psi$, the chosen AS $x\in \mathcal{X}$ is known to the defender yet remains unknown to the insiders. 
Before implementing the chosen AS $x$, the defender recommends an action based on $x$ and the recommendation policy $\pi\in \Pi$. 
Then, the insider takes an action $a\in \mathcal{A}$ that is not necessarily the recommended one. 
Finally, the defender implements the chosen AS and penalizes insiders based on the audit outcome. 
Insiders observe the chosen AS after it is implemented. 
After the zero-trust audit, the insiders' actions become known to the defender. 
%{\bf WHEN DO WE IMPLEMENT IT? IS IT KNOWN TO THE INSIDERS BEFORE THEY BEHAVE?}yet remains unknown to the insiders until implemented. 
%as the result of the zero-trust audit, all insiders' actions are known to the defender. %after they are taken. 

\subsubsection{Insider's Initial Compliance} %under Misaligned Incentive
\label{sec:insider's Initial Compliance}
%The compliance issue in Fig. \ref{fig:ScenarioDiag} is an example of the agency dilemma where insiders take compliance actions on behave of the defender to improve organizational cyber hygiene. 
Without the recommendation mechanism, %illustrate by the third stage in Fig. \ref{fig:overviewDiag}, 
an insider takes an action $a_0\in \mathcal{A}$ to maximize his expected utility concerning the prior statistics $b_Y\in \Delta \mathcal{Y}$ and $\psi\in \Psi$, i.e., $a_0\in \text{arg}\max_{a\in \mathcal{A}} \sum_{y\in \mathcal{Y}} b_Y(y) \allowbreak \sum_{x\in\mathcal{X}} \psi(x|y) {v}_U (y,x,a)$. 
%$a_0\in \text{arg}\max_{a\in \mathcal{A}} \mathbb{E}_{y\sim b_{Y}}\mathbb{E}_{x\sim \psi(\cdot|y)} [ {v}_U (y,x,a)]$
Due to the misalignment between an insider's incentive $v_U$ and the defender's security objective $v_D$, the insider's initial compliance status represented by  $a_0\in \mathcal{A}$ may negatively affect corporate security. 
%Agency dilemma arises due to the misalignment between an insider's incentive $v_U$ and the defender's security objective $v_D$, and the insider's initial compliance status represented by  $a_0\in \mathcal{A}$ may negatively affect the corporate security. 
%Since the insider's incentive $v_U$ may not align with 
%is different from 
%the defender's security objective $v_D$, the defender can face the agency dilemma, and the insider's initial compliance status represented by  $a_0\in \mathcal{A}$ may negatively affect the corporate security. 
For example, a self-interested insider tends to break the security rules for convenience if the audit policy $\psi$ chooses a stringent audit (e.g., $x^{H+1}\in \mathcal{X}$ in Example \ref{example}) less frequently. 
%inspect more rules less frequently. . 

\subsubsection{Recommendation Mechanism}
\label{sec:recommendation mechanism}
%the entire ZETAR (both audit and recommendation) aims to improve compliance
%To mitigate the agency dilemma in Section \ref{sec:insider's Initial Compliance} and achieve a customized design for insiders with different incentives, we complement the zero-trust audit with the recommendation mechanism shown in the third stage of Fig. \ref{fig:overviewDiag}.  
To align insiders' incentives with the defender's security objective, the defender can recommend an action to an insider. 
Thus, set $\mathcal{S}:=\{s^k\}_{k\in \mathcal{K}}$ has the same cardinality with $\mathcal{A}$ and represents the finite set of $K$ recommendations where $s^k\in \mathcal{S}$ recommends the insider to take action $a^k\in \mathcal{A}$. 
%Sets $\mathcal{A}$ and $\mathcal{S}$ have the same cardinality. 
%{\bf ARE THE CARDINALITIES OF THE SET A AND S THE SAME?}
The defender recommends the action according to a stochastic recommendation policy $\pi\in \Pi: \mathcal{X}\mapsto \Delta \mathcal{S}$; i.e., given the chosen AS $x\in \mathcal{X}$, the defender chooses recommendation $s\in \mathcal{S}$ with probability $\pi(s|x)$. %at a frequency of $\pi(s|x)$ every unit time?  . 
As will be shown in Section \ref{sec:insider's Belief Update and Best-Response Action}, by strategically choosing the recommendation policy, the defender can manipulate an insider's belief of the current SP and the chosen AS, thus affecting his perception of the expected utility and enhancing compliance. 
%changing the compliance status on average. 
 
 \subsubsection{Insider's Belief Update and Best-Response Action}
 \label{sec:insider's Belief Update and Best-Response Action}
The received recommendation reveals the defender's knowledge of the SP and the chosen AS. An insider can form and update a belief of the unknowns by observing the recommendations.
 Denote $b_{Y,X}\in \mathcal{B}_{Y,X}\subseteq \Delta(\mathcal{X}\times\mathcal{Y})$ as the joint prior distribution of the current SP and the chosen AS, i.e., $b_{Y,X}(y,x):= b_Y(y)\psi(x|y), \forall x\in \mathcal{X},y\in \mathcal{Y}$.
 Analogously, we define 
 %We define the following shorthand notations: $b_{Y,X}(y,x):= b_Y(y)\psi(x|y)$ represents the joint probability of SP $y\in \mathcal{Y}$ and AS $x\in \mathcal{X}$; 
 $b_X(x):=\sum_{y'\in \mathcal{Y}}b_{Y,X}(y',x)$ as the marginal prior probability of AS $x\in \mathcal{X}$, $b_{Y|X}(y|x):=b_{Y,X}(y,x)/b_X(x)$ as the conditional prior probability of SP $y\in \mathcal{Y}$ under AS $x\in \mathcal{X}$, and $b^{\pi}_S(s):=\sum_{x'\in \mathcal{X}} b_X(x' )\pi(s|x')$ as the probability of recommendation $s\in \mathcal{S}$ under $\pi\in\Pi$, where
$b_X\in \mathcal{B}_{X}\subseteq \Delta \mathcal{X}$, $b_{Y|X}\in \mathcal{B}_{Y|X}$, and $b_S^\pi \in \Delta \mathcal{S}$. 
  
% Define set $\mathcal{B}_{Y,X}:=\{\{b_{Y,X}(y,x)\}_{\forall x\in \mathcal{X},y\in \mathcal{Y}}|b_{Y,X}(y,x)\in [0,1], \forall x\in \mathcal{X},y\in \mathcal{Y}, \sum_{x\in \mathcal{X}} b_{Y,X}(y,x)=b_Y(y), \forall y\in \mathcal{Y}\}$. 
 Following the requirement of transparent criteria in Section \ref{sec:Information Structure and Dynamics}, the recommendation policy $\pi\in \Pi$ is assumed to be common knowledge. 
 The assumption can be justified by the fact that an insider can learn the recommendation policy $\pi\in \Pi$ based on the repeated observations of the recommendation policy input (i.e., AS $x\in \mathcal{X}$) and the policy output (i.e., recommendation $s\in \mathcal{S}$) after they are implemented. 
 %the associated recommendations and ASs. %, as the defender needs to audit on a weekly or monthly basis. 
 %Since the audit is repeated on a weekly or monthly basis and an insider can know these audit schemes implemented, he can statistically verify or learn the recommendation policy based on the repeated observations of recommendation and the associated AS. 
Thus, for rational insiders who adopt Bayesian rules to update their beliefs, each recommendation $s\in \mathcal{S}$ under recommendation policy $\pi\in\Pi$ results in posterior belief $b_{Y,X}^{\pi}(y,x|s)\in \mathcal{B}_{Y,X}^{\pi} \subseteq  \Delta(\mathcal{X}\times\mathcal{Y})$, i.e., 
\begin{equation}
%\small
\label{eq:BayesUpdate}
b_{Y,X}^{\pi}(y,x|s)=\frac{b_{Y,X}(y,x)\pi(s|x)}{\sum_{x'\in \mathcal{X}} b_X(x' )\pi(s|x')}, \forall x\in \mathcal{X}, y\in \mathcal{Y}.  
\end{equation}
Then, we can obtain the insider's marginal posterior belief of AS $x\in \mathcal{X}$, his marginal posterior belief of SP $y\in \mathcal{Y}$, and the associated conditional posterior belief  under recommendation $s\in \mathcal{S}$ as 
$b_{X}^{\pi}(x|s):=\sum_{y\in \mathcal{Y}} b_{Y,X}^{\pi}(y,x|s)=\frac{b_{X}(x)\pi(s|x)}{b^{\pi}_S(s)} \in \mathcal{B}_X^{\pi} \subseteq \Delta \mathcal{X}$, 
$b_{Y}^{\pi}(y|s):=\sum_{x\in \mathcal{X}} b_{Y,X}^{\pi}(y,x|s) \in \Delta \mathcal{Y}^{\pi} \subseteq \Delta \mathcal{Y}$, 
and $b_{Y|X}^{\pi}(y|x,s):={b_{Y,X}^{\pi}(y,x|s)}/{b_{X}^{\pi}(x|s)}$, respectively. 
Since $b_{Y|X}^{\pi}(y|x,s)=b_{Y|X}(y|x), \forall s\in\mathcal{S}$, these recommendations under policy $\pi\in \Pi$ have no impact on the conditional probability $b^\pi_{Y|X}$. 
However, as it does not hold in general that $b_{Y,X}^{\pi}=b_{Y,X},b_X^{\pi}=b_X,b_Y^{\pi}=b_Y, \forall s\in \mathcal{S}$, the recommendation mechanism (i.e., $\pi\in \Pi$ and $s\in \mathcal{S}$) can change the insiders' marginal beliefs of the current SP and the implemented AS as well as their joint beliefs. 
We summarize the above observations in Lemma \ref{remark:Conditional Belief Invariance}. 
\begin{lemma}[\textbf{Invariance of Conditional Belief}]
\label{remark:Conditional Belief Invariance}
A recommendation policy $\pi\in \Pi$ has no impact on $b_{Y|X}^{\pi}$. 
%impacts on $b_{Y,X}^{\pi},b_X^{\pi},b_Y^{\pi}$ but not $b_{Y|X}^{\pi}$. 
\end{lemma}
% \begin{equation}
% \label{eq:BayesUpdate}
% b_{X}^{\pi}(x|s) =\sum_{y\in \mathcal{Y}}\frac{b_{Y,X}(y,x)\pi(s|x)}{b^{\pi}_S(s)}=\frac{b_{X}(x)\pi(s|x)}{b^{\pi}_S(s)}, \forall x\in \mathcal{X}.  
% \end{equation}

With a recommendation policy $\pi\in \Pi$, the insider takes a best-response action denoted by $a^*_{\pi,s}\in \mathcal{A}$ to maximize his posterior utility under recommendation $s\in \mathcal{S}$, i.e., %PU is taken, use PoU instead for posterior utility
\begin{equation}
\label{eq:optimalAction}
a^*_{\pi,s} \in \text{arg}\max_{a\in \mathcal{A}} \mathbb{E}_{y,x\sim b_{Y,X}^{\pi}(\cdot|,s)} [{v}_U (y,x,a)]. 
\end{equation}
%With shorthand notation 
Letting $\bar{v}_p(x,a):=\sum_{y\in \mathcal{Y}} b_{Y|X}(y|x)  v_p(y,x,a), 
\allowbreak
\forall x\in \mathcal{X}$ for $p\in \{U,D\}$ be an insider's expected incentive and the defender's expected security objective, respectively, we obtain 
\begin{equation}
\label{eq:bareq}
    \mathbb{E}_{y,x\sim b_{Y,X}^{\pi}(\cdot|,s)} [{v}_p (y,x,a)]= \sum_{x\in \mathcal{X}} b_{X}^{\pi}(x|s)  \bar{v}_p (x,a), \forall s\in \mathcal{S}. 
\end{equation}
%$\mathbb{E}_{y,x\sim b_{Y,X}^{\pi}(\cdot|,s)} [{v}_U (y,x,a)]= \sum_{x\in \mathcal{X}} b_{X}^{\pi}(x|,s)  \bar{v}_U (x,a)$ for all $s\in \mathcal{S}$. 

%We refer to $a^*_{\pi,s}\in\mathcal{A}$ as an action induced by recommendation $s\in \mathcal{S}$ under recommendation policy $\pi\in \Pi$
We refer to $a^*_{\pi,s}\in\mathcal{A}$ as an action induced by recommendation policy $\pi\in \Pi$ under recommendation $s\in \mathcal{S}$, which is in general different from the insider's initial compliance status $a_0\in\mathcal{A}$ in Section \ref{sec:insider's Initial Compliance}.  
For a given $\pi$, not all recommendations induce a compliant action. 
%action that improves the compliance status. 
However, by strategically choosing the recommendation policy, the defender can 
%send compliance-improvement recommendations with a high probability to 
improve compliance on average. 
We formally quantify the improvement of compliance and its average impact on the corporate security in Section \ref{sec:Defender's Optimal Recommendation Policy}.

%By revising the utility function $v_U$, we can make \eqref{eq:BayesUpdate} and \eqref{eq:optimalAction} applicable to the insider with bounded rationality. 
% Why this formulation makes sense?
% \begin{itemize}
%     \item It is normal to consider self-interested insiders. 
%     \item It provides a benchmark framework. We can add prospect theory later. 
%     \item Utility is a way to capture incentives. Finally, we need to know the best response that can be learned by human behavioral data. 
%     The model helps us to make use of data more efficiently. 
%     \item For insiders with bounded rationality, \eqref{eq:BayesUpdate} and \eqref{eq:optimalAction} are still applicable with a revised utility function $v_U$. 
% \end{itemize}

\subsubsection{Trustworthiness of the Recommendation Scheme}
Following Section \ref{sec:insider's Belief Update and Best-Response Action}, the defender's recommendation $s\in\mathcal{S}$ from a recommendation policy $\pi\in \Pi$ may not be trusted by an insider; i.e., the recommended action is not a best-response action. 
%an insider's best-response action $a_{\pi,s}^*$ under recommendation policy $\pi\in \Pi$ does not necessary coincide with the action associated with the recommendation $s\in \mathcal{S}$. 
We formalize the definitions of trustworthy recommendations and trustworthy recommendation policies in Definitions \ref{def:trustworthS} and \ref{def:trustworthyPi}, respectively. 

\begin{definition}[\textbf{Trustworthy Recommendations}]
\label{def:trustworthS}
A recommendation $s^k\in \mathcal{S}, k\in \mathcal{K}$, under a recommendation policy $\pi\in \Pi$ is trustworthy (resp. untrustworthy); i.e., the policy $\pi$ is trusted by an insider with an incentive $\bar{v}_U$, if the induced action follows (resp. does not follow) the recommended action  $a^k\in \mathcal{A}$, 
%the recommended action $a^k\in \mathcal{A}$ is (resp. is not) a best-response action, 
i.e., $a^k \in$ (resp. $\notin$) $\text{arg}\max_{a\in \mathcal{A}} \mathbb{E}_{x\sim b_{X}^{\pi}(\cdot|,s)} [ \bar{v}_U (x,a)]$. 
%An insider is said to be compliant (resp. non-compliant) with a recommendation $s^k, k\in \mathcal{K}$ under the recommendation policy $\pi\in \Pi$ if his best-response action $a_{\pi,s^k}^*\in \mathcal{A}$ coincides (resp. does not coincide) with the recommendation $s^k \in \mathcal{S}$, i.e.,  $a_{\pi,s}^*=a^k$ (resp. $a_{\pi,s}^*\neq a^k$). We refer to the best-response action $a_{\pi,s}^*=a^k$ (resp. $a_{\pi,s}^*\neq a^k$) as a compliant (resp. non-compliant) action. 
\end{definition}

\begin{remark}[\textbf{Compliance and Trustworthiness}]
\label{remark:compliance and trustworthy}
%Directly
Following Definition \ref{def:trustworthS}, an insider complies with a recommendation  (i.e., takes the recommended action) only if it is trustworthy. 
\end{remark}

\begin{definition}[\textbf{Trustworthy Recommendation Policies}]
\label{def:trustworthyPi}
%A recommendation policy is either completely trustworthy (CT), partially trustworthy (PT), or completely untrustworthy (CU). 
Recommendation policies under which recommendation $s^k\in \mathcal{S}, k\in \mathcal{K}$, is trustworthy (resp. untrustworthy) formulate the $k$-th Partially Trustworthy (PT) (resp. Partially Untrustworthy (PU)) policy set $\Pi_{pt}^k\subseteq \Pi$  (resp. $\Pi_{pu}^k\subseteq \Pi$). 
A recommendation policy $\pi\in \Pi$ is Completely Trustworthy (CT) (resp. Completely Untrustworthy (CU)) if all recommendations under $\pi$ are trustworthy (resp. untrustworthy). 
All CT (resp. CU) recommendation policies formulate the CT (resp. CU) policy set $\Pi_{ct}:=\cap_{k=1}^K \Pi_{pt}^k$ (resp. $\Pi_{cu}:=\cap_{k=1}^K \Pi_{pu}^k$). % \subseteq \Pi
%An insider's compliance status concerning a recommendation policy can be either full compliance, full non-compliance, or partial compliance. An insider is said to be fully compliant (resp. fully non-compliant) with a recommendation policy $\pi\in \Pi$ if he is compliant (resp. non-compliant) with all recommendations under the policy $\pi$. For all other cases, the insider is said to be partially compliant with the recommendation policy $\pi$. 
\end{definition}

% \begin{definition}[\textbf{Trustworthy Policy Set}]
% \label{def:compliance policy set}
% %We define $\Pi_{co}\subseteq \Pi$ as the compliance policy set that contains all the recommendation policies that the insider is fully compliant with. 
% \end{definition}

Different recommendation policies reveal varied amounts information about the AS and the SP, which consequently affect the insider's compliance status. 
Two extreme cases are defined in Definition \ref{def:zeroFullinfo}. 
Let the optimal action of an insider $U$ or the defender $D$ at AS $x\in \mathcal{X}$ and SP $y\in \mathcal{Y}$ be given by $\tilde{a}_p^{max}(y,x)\in \text{arg}\max_{a\in \mathcal{A}}{v}_p(y,x,a)$. 
Analogously, let ${a}_p^{max}(x)\in \text{arg}\max_{a\in \mathcal{A}}\bar{v}_p(x,a)$ for all $x\in \mathcal{X}$ and $p\in \{U,D\}$. 
A zero-information recommendation policy, denoted by $\pi_z\in \Pi$, recommends the same actions as an insider's initial compliance status in Section \ref{sec:insider's Initial Compliance} regardless of the chosen AS. 
Hence $\pi_z$ does not change the insider's belief, i.e., $b_X^{\pi_z}(x|s)=b_X(x), \forall s\in \mathcal{S}, \forall x\in \mathcal{X}$, and does not bring new information to the insider. 
%Equivalently, $b_{Y,X}^{\pi_z}=b_{Y,X}$. 
Meanwhile, a full-information recommendation policy denoted by $\pi_f\in \Pi$ recommends optimal action $a_U^{max}(x)$ under the chosen AS $x\in \mathcal{X}$. 
%Thus, the insider under $\pi_f$ takes the same actions as he does when he knows the chosen AS before implemented. 
Remark \ref{remark:CT Feasibility} shows that it is feasible for the defender to implement CT recommendation policies regardless of ZETAR settings in Sections \ref{sec:SP} to \ref{sec:utilities}. % \ref{sec:audit and action}, and
%ZETAR 

 \begin{definition}[\textbf{Zero- and Full-Information Recommendation Policy}] %there are different definition, we consider the trustworthy ones (CT)
\label{def:zeroFullinfo}
% If $\pi(s|x)=\pi(s|x'), \forall s\in \mathcal{S}, \forall x,x'\in \mathcal{X}$, then $\pi$ is a zero-trust generator. If the mapping $\pi: \mathcal{X} \mapsto \mathcal{S}$ is injective, then $\pi$ is a full-trust generator.
 A recommendation policy $\pi_z\in \Pi$ contains zero information if $\pi_z(s^k|x)=\mathbf{1}_{\{a^k=a_0\}}, \forall k\in \mathcal{K}, \forall x\in \mathcal{X}$.  A recommendation policy $\pi_f\in \Pi$ contains full information if $\pi_f(s^k|x)=\mathbf{1}_{\{a^k=a_U^{max}(x) \}}, \forall k\in \mathcal{K}, \forall x\in \mathcal{X}$. 
\end{definition}
\begin{remark}[\textbf{Feasibility}]%Existence of CT Recommendation Policies
\label{remark:CT Feasibility}
Following Definition \ref{def:zeroFullinfo}, zero- and full-information recommendation policies are CT, i.e., $\pi_z,\pi_f\in \Pi_{ct}$. 
Thus, $\Pi_{ct}$ is nonempty regardless of ZETAR settings. %parameters. %in Section \ref{sec:model}. 
\end{remark}

 \subsubsection{Defender's Optimal Recommendation Policy}
\label{sec:Defender's Optimal Recommendation Policy}

Following \eqref{eq:optimalAction} and \eqref{eq:bareq}, an insider's expected utility defined in \eqref{eq:insiderEPU} represents the insider's Acquired Satisfaction Level (ASaL) under recommendation policy $\pi\in \Pi$. 
\begin{equation}
%\small
\label{eq:insiderEPU}
    J_U(\pi,b_{X},\bar{v}_U):=\sum_{s\in \mathcal{S}} b^{\pi}_S(s)  \mathbb{E}_{y,x\sim b_{Y,X}^{\pi}(\cdot|,s)} [ {v}_U (y,x,a_{\pi,s}^*)]. 
\end{equation}
Since an insider's action induced by zero-information policy $\pi_z$ is his initial-compliance action $a_0\in \mathcal{A}$, $J_U(\pi,b_{X},\bar{v}_U)$ represents the insider's Innate Satisfaction Level (ISaL). 
To capture the average impact of an insider's compliance status on corporate security under different recommendations, we define the defender's Acquired Security Level (ASeL) under recommendation policy $\pi\in \Pi$ as
 \begin{equation}
 %\small
\begin{split}
         \label{eq:def_obj}
    & \tilde{J}_D(\pi,b_{Y,X}, {v}_D,{v}_U) :=\mathbb{E}_{y,x\sim b_{Y,X}(\cdot)} \mathbb{E}_{s\sim \pi(\cdot | x)} [{v}_D(y,x,a_{\pi,s}^*)] \\
    & =\sum_{x\in \mathcal{X}} b_X(x) \sum_{s\in \mathcal{S}} \pi(s|x) \bar{v}_D(x,a^*_{\pi,s}):={J}_D(\pi,b_X, \bar{v}_D,\bar{v}_U). 
\end{split}
 \end{equation}
 
%  \begin{equation}
% \begin{split}
%          \label{eq:def_obj}
%      J_D(\pi,b_Y,\psi, v_D,v_U)&:=\mathbb{E}_{y,x\sim b_{Y,X}(\cdot)} \mathbb{E}_{s\sim \pi(\cdot | x)} [{v}_D(y,x,a_{\pi,s}^*)] \\
%      &=\sum_{x\in \mathcal{X}} b_X(x) \sum_{s\in \mathcal{S}} \pi(s|x) \bar{v}_D(x,a^*_{\pi,s}). \\
%     &=\sum_{s\in \mathcal{S}} \sum_{x\in \mathcal{X}}  \pi(s|x) \hat{v}_D(x,a^*_{\pi,s}). 
% \end{split}
%  \end{equation}
%Without a recommendation scheme of $\mathcal{S}$ and $\Pi$ illustrated in the third stage in Fig. \ref{fig:overviewDiag}, an insider takes action $a_0$ based on their initial incentives $\bar{v}_U$ and  prior statistics  $b_X\in \mathcal{B}_X$ as shown in Section \ref{sec:insider's Initial Compliance}, which results in the same EPU as a zero-information recommendation policy $\pi_z\in \Pi$ does. 
Since an insider's best-response action $a_{\pi_z,s}^*\in \mathcal{A}$ remains the same as $a_0\in \mathcal{A}$ in Section \ref{sec:insider's Initial Compliance} under all recommendations $s\in \mathcal{S}$, a zero-information recommendation policy $\pi_z\in \Pi_{ct}$ has no impact on the insider's compliance. 
Hence
%zero-information recommendation policies do not change the insider's compliance 
$J_D(\pi_z,b_X,\bar{v}_D,\bar{v}_U)$ quantifies the impact of an insider's initial compliance status and represents the defender's Initial Security Level (ISeL). 
The difference in the defender's security level  ${J}_D^{acel}(\pi,b_{X}, \bar{v}_D,\bar{v}_U):=J_D(\pi,b_X,\bar{v}_D,\bar{v}_U)-J_D(\pi_z,b_X,\bar{v}_D,\bar{v}_U)$ measures the average impact of the insider's compliance status changes (under recommendation policy $\pi\in \Pi$) on the corporate security, and we refer to ${J}_D^{acel}$ as the Average Compliance Enhancement Level (ACEL) in Definition \ref{def:ACEL}.  

\begin{definition}[\textbf{Average Compliance Enhancement Level}] %Expected CEL = EXCEL
\label{def:ACEL}
For an insider with incentive $\bar{v}_U$ and the defender with security objective $\bar{v}_D$, we define $J_D^{acel}(\pi, b_X,\bar{v}_D,\bar{v}_U)\in \mathbb{R}$ as the Average Compliance Enhancement Level (ACEL) under the prior statistic $b_X\in \mathcal{B}_X$ defined in Section \ref{sec:insider's Belief Update and Best-Response Action} and recommendation policy $\pi\in \Pi$ defined in Section \ref{sec:recommendation mechanism}. 
\end{definition}
The defender's goal is to design the optimal recommendation policy $\pi^*\in \Pi$ that maximizes the ACEL, where $J_D^{acel,*}$ denotes the optimal ACEL, i.e., $J_D^{acel,*}(b_{X}, \bar{v}_D,\bar{v}_U):={J}_D^{acel}(\pi^*,b_{X}, \bar{v}_D,\bar{v}_U)=\max_{\pi\in \Pi} {J}_D^{acel}(\pi,b_{X}, \bar{v}_D,\bar{v}_U)\geq 0$. 
The optimal ACEL gauges the maximum improvement of an insider's compliance. 
\textcolor{black}{Thus, its value is proportional to the insider's persuadability under a recommendation scheme.
On the other hand, the ISeL quantifies the defender's expected utility in the presence of insider behaviors without any incentive mechanism. 
Thus, its value is proportional to the insider's initial compliance. 
These metrics are useful to develop scoring metrics to quantitatively categorize insiders, as shown in Remark \ref{remark:scoring}.} 
\begin{remark}[\textbf{Scoring Metrics}]
\label{remark:scoring}
The \textcolor{black}{values of} ISeL $J_D(\pi_z,b_X,\bar{v}_D,\bar{v}_U)$ and the optimal ACEL $J_D^{acel,*}(b_{X}, \bar{v}_D,\bar{v}_U)$ \textcolor{black}{measure} how compliant and persuadable, respectively, an insider with incentive $\bar{v}_U$ is under security objective $\bar{v}_D$. 
\end{remark}

\section{Computational Framework of ZETAR}  
\label{sec:Computational Framework of ZETAR}
In this section, we formulate the design of ZETAR into mathematical programming problems, where the defender has complete information of an insider's incentive $\bar{v}_U$. 

\subsection{Level of Recommendation Customization}
\label{sec:LoRC}
As illustrated in Section \ref{sec:audit and action} and \ref{sec:Recommendation} and also in Fig. \ref{fig:ScenarioDiag}, the defender determines a unified audit policy to inspect all insiders' behaviors yet designs customized recommendation policies. %for insiders with different incentives. 
Since the difference in these recommendation policies can lead to the perceptions of unfairness and distrust \cite{theis2019common}, the defender needs to strike a balance between the optimal ACEL and the Level of Recommendation Customization (LoRC). 
We let $\eta\in \mathbb{R}^{+}$ be the defender's LoRC,  $\pi_{d}\in \Pi$ be a default recommendation policy, and the KL divergence $KL(\pi||\pi_{d}):=\sum_{k\in \mathcal{K},x\in \mathcal{X}}\pi(s^k|x) \log \frac{ \pi(s^k|x) }{ \pi_{d}(s^k|x) }$ be the measure of policy difference, respectively.  
%The default recommendation policy can be either easy to implement (e.g., a policy that recommends actions randomly for all AS) or is beneficial to the corporate security on average (e.g., the defender can estimate the insider's average incentives to choose one unified recommendation policy as the default policy).  
If $\pi_{d}(s^k|x)=0$, then $\pi(s^k|x)=0$ by default, and $\pi(s^k|x) \log \frac{ \pi(s^k|x) }{ \pi_{d}(s^k|x) }=0$ as $\lim_{z\rightarrow 0^+} z \log z=0$. 
%\textcolor{black}{The defender can also find other ways to prevent }

\subsection{Primal Mathematical Programming}
\label{sec:LP formulation}

Without loss of generality, the defender can narrow the policy search space to $\Pi_{ct}\subseteq \Pi$ to achieve the optimal ACEL \cite{kamenica2011bayesian}, i.e., $J_D^{acel,*}(b_{X}, \bar{v}_D,\bar{v}_U)=\max_{\pi\in \Pi_{ct}} J_D^{acel}(\pi, b_{X}, \bar{v}_D,\bar{v}_U)$. %=\max_{\pi\in \Pi_{ct}} J_D(\pi, b_{X}, \bar{v}_D,\bar{v}_U) 
Under a CT recommendation policy, the insider complies to the recommendation and chooses $a^k\in \mathcal{A}$ when the recommendation is $s^k\in\mathcal{S}, \forall k\in \mathcal{K}$. Thus,  $\max_{\pi\in \Pi_{ct}} J_D^{acel}(\pi, b_{X}, \bar{v}_D,\bar{v}_U)=\max_{\pi\in \Pi_{ct}}\sum_{x\in \mathcal{X}} b_X(x) \sum_{k\in \mathcal{K}} \pi(s^k|x) \bar{v}_D(x,a^k)$. 
For a LoRC $\eta$, we formulate the following convex program denoted by $P_\eta$. %to design the optimal recommendation policy. 
\begin{equation*}
%\small 
\begin{split}
& [P_\eta]: r_\eta = \max_{\pi\in \Pi } \quad   \sum_{x\in \mathcal{X}} b_X(x) \sum_{k\in \mathcal{K}} \pi(s^k|x) \bar{v}_D(x,a^k)   - \frac{ KL(\pi||\pi_{d})}{\eta}  \\ 
 &  (a).  \    \pi (s^k|x)\geq 0, \forall k\in \mathcal{K}, \forall x\in\mathcal{X},
 \\
 & (b).   
 \sum_{k\in \mathcal{K}  }\pi (s^k|x)=1, \forall x\in\mathcal{X},
 \\
& (c) .  
\sum_{x\in \mathcal{X}} b_{X}(x) \pi(s^k| x) [ \bar{v}_U(x,a^k) - \bar{v}_U(x,a^l)  ]  \geq 0,
\forall k,l \in \mathcal{K}. 
%\\
%& (c) . \ \hat{\pi}^k [\hat{v}_U^k -\hat{v}_U^h]\geq 0, \forall k,h \in \mathcal{K}. 
\end{split}
\end{equation*} 
Let $\pi^*_{\eta}\in \Pi_{ct}$ and $r_\eta$ be the maximizer and the optimal value of ${P}_{\eta}$, respectively, for all $\eta\in \mathbb{R}^{+}$.
% \begin{equation*}
% \begin{split}
% & \text{(PLP):}   \quad 
% \sup_{\pi\in \Pi} \quad \sum_{j=1}^J \sum_{i=1}^I b_{Y,X}(y^j,x^i)\sum_{k=1}^K \pi(s^k| x^i) {v}_D(y^j,x^i,a^k) \\
%  &  (a).  \    \pi (s^k|x^i)\geq 0, \forall k\in \mathcal{K}, \forall i\in\mathcal{I},
%  \\
%  & (b).   
%  \sum_{k\in \mathcal{K}  }\pi (s^k|x^i)=1, \forall i\in\mathcal{I},
%  \\
% & (c) .  
% \sum_{j\in \mathcal{J}}\sum_{i\in \mathcal{I}} b_{Y,X}(y^j,x^i) \pi(s^k| x^i) [ {v}_U(y^j, x^i,a^k) \\
%  & \quad\quad\quad\quad \quad\quad\quad\quad - {v}_U(y^j,x^i,a^l)  ]  \geq 0,
% \forall k,h \in \mathcal{K}. \\
% & (c) . \ \hat{\pi}^k [\hat{v}_U^k -\hat{v}_U^h]\geq 0, \forall k,h \in \mathcal{K}. 
% \end{split}
% \end{equation*} 
Constraints (a), (b) explicitly describe the set $\Pi$, and constraint (c) limits the recommendation policy to be CT defined in Definition \ref{def:trustworthyPi}. 
All recommendation policies that satisfy constraints (a), (b), (c) compose the set $\Pi_{ct}\subseteq \Pi$. 
Due to the feasibility of CT policies in Remark \ref{remark:CT Feasibility} and the boundedness of $v_D$, the program $P_\eta$ is feasible and bounded for all $\eta\in \mathbb{R}^{+}$. 
When the defender aims to design CT recommendation policies closest to the default policy $\pi_d\in \Pi$ (i.e., $\eta\rightarrow 0^+$), then $\pi_0^*=\pi_{d}$ if and only if $\pi_{d}\in \Pi_{ct}$. 
%if eta=0 and pi_d is not feasible, then the obj is infinite
As the LoRC $\eta$ increases, the defender focuses more on compliance enhancement, and the optimizer of $P_{\infty}$ coincides with $\pi^*$ that achieves the optimal ACEL $J_D^{acel,*}$, i.e., $\pi^*_{\infty}=\pi^*$.  
By specifying $a^l\in \mathcal{A}$ in constraint (c) of $P_\eta$ as the initial-compliance action $a_0\in \mathcal{A}$, we prove that CT policies never decrease an insider's satisfaction level in Proposition \ref{proposition:win-win}.  

% If the regularization term is not included, as $\eta$ increases, the defender focuses more on improving compliance, so the original term of compliance enhance is increasing and the KL part decreases. 
% \begin{lemma}
% The objective function of PCP (the regularization term is included) is non-decreasing as we increase $\eta$. 
% That means that the increase of compliance enhancement dominates the increase of uniformed. 
% \label{lemma:non-decreasing}
% \end{lemma}
% \begin{proof}
% If $\eta^2\geq \eta^1$, then 
% \begin{equation}
%     \begin{split}
%       & \max_{\pi\in \Pi_{ct}} [ J_D^{acel}(\pi, b_{X}, \bar{v}_D,\bar{v}_U)  - \frac{ KL(\pi||\pi_{d})}{\eta^1} ] \\
%       - &  \max_{\pi\in \Pi_{ct}} [ J_D^{acel}(\pi, b_{X}, \bar{v}_D,\bar{v}_U)  - \frac{ KL(\pi||\pi_{d})}{\eta^2}] \\
%       \geq & -  \max_{\pi\in \Pi_{ct}} [ \frac{ KL(\pi||\pi_{d})}{\eta^2}  - \frac{ KL(\pi||\pi_{d})}{\eta^1} ] \\
%       = & ( \frac{1}{\eta^{R,1}}  -\frac{1}{\eta^{R,2}}) \cdot \min_{\pi\in \pi_{d}}  KL(\pi||\pi_{d}) \geq 0 , 
%     \end{split}
% \end{equation}
% due to the fact the KL is non-negative. 
% \end{proof}
% \begin{remark}
% When $\eta$ increases, the ACEL term increases while the KL term decreases. The above lemma shows that as $\eta$ increases, the increases of ACEL dominates the increase of KL. The defender shift the focus more on ACEL than difference level. 
% \end{remark}

\begin{proposition}[\textbf{Trustworthiness Promotes Satisfaction}]%Principles of Mutual Benefit is too general
\label{proposition:win-win}
An insider's ASaL $J_U(\pi,b_X,\bar{v}_U)$ is not lower than his ISaL $J_U(\pi_z,b_X,\bar{v}_U)$ for all $\pi\in\Pi_{ct}$ and $b_X\in \mathcal{B}_X$. 
\end{proposition}

\begin{proof}
An insider's ASaL in \eqref{eq:insiderEPU} under a CT recommendation policy $\pi\in \Pi_{ct}$ can be represented as 
{\small
$
     J_U(\pi,b_X,\bar{v}_U) \allowbreak =\sum_{s\in \mathcal{S}} b_S^\pi(s) \cdot \max_{a\in \mathcal{A}} \sum_{x\in \mathcal{X}} b^{\pi}_X(x|s)\bar{v}_U(x,a)
      = \sum_{x\in \mathcal{X}} b_X(x) \sum_{k\in \mathcal{K}} \pi(s^k|x)\bar{v}_U(x,a^k)
$. 
}
Based on constraint (c) of $P_\eta$, $\sum_{x\in \mathcal{X}} b_X(x) \pi(s^k|x) [\bar{v}_U(x,a^k)-\bar{v}_U(x,a_0)]\geq 0$ for all $k\in \mathcal{K}$ and $\pi\in \Pi_{ct}$. 
Hence, $J_U(\pi,b_X,\bar{v}_U)\geq \sum_{x\in \mathcal{X}} b_X(x)\bar{v}_U(x,a_0)= \max_{a\in \mathcal{A}} \sum_{x\in \mathcal{X}} b(x)\bar{v}_U(x,a) = J_U(\pi_z,b_X,\bar{v}_U)$. 
\end{proof}
\begin{remark}[\textbf{Win-Win Situation}]
Proposition \ref{proposition:win-win} shows that an insider's ASaL is not lower than his ISaL if a recommendation policy is CT. 
Based on Remark \ref{remark:CT Feasibility}, the defender's ASeL is not lower than her ISEL under the optimal recommendation policy, i.e., $J_D^{acel,*}(b_{X}, \bar{v}_D,\bar{v}_U)\geq 0$. 
Thus, the optimal policy achieves a win-win situation between the defender and insiders by promoting cyber hygiene and insiders' satisfaction. 
\end{remark}

\subsection{Dual Mathematical Programming}

Let $\alpha_\eta(s^k,x)\geq 0$, $\beta_\eta(x)\in \mathbb{R}$, and $\lambda_\eta(s^k, a^l)\geq 0$ denote the dual variables of the constraints (a), (b), and (c) in $P_\eta$, respectively. 
Define shorthand notation $\bar{\beta}_\eta(s^k , x,\lambda_\eta):=  \bar{v}_D(x,a^k) + \sum_{a^l \in \mathcal{A} } \lambda_\eta(s^k, a^l) [ \bar{v}_U(x,a^k) - \bar{v}_U(x,a^l)]$, where $\lambda_\eta(s^l, a^l), \forall l\in\mathcal{K}$, can take any finite values. 
%Note that $\lambda_\eta(s^k, a^l)$ can take arbitrary value as $[\bar{v}_U(x,a^k) - \bar{v}_U(x,a^l)]=0$ if $k=l$. 
The dual problem is denoted as $D_\eta$. 
The strong duality proved in Proposition \ref{proposition:dualconvex} yields the bounds for the optimal value $r_\eta$ of the primal problem $P_\eta$ in Proposition \ref{proposition:bound}. 
Define $\alpha^*_\eta(s^k,x)$, $\beta^*_\eta(x)$, and $\lambda^*_\eta(s^k, a^l)$ as the optimal dual variables and the shorthand notation $\underline{r}:=\sum_{x\in \mathcal{X}} [ \max_{k\in \mathcal{K}}  b_X(x)\bar{\beta}_\eta(s^k,x, \lambda_\eta)+\log(\pi_{d}(s^k|x))/\eta ]$. 
%and the bounds for its optimal value . %, denoted by $D_\eta$, of the primal problem $P_\eta$
\begin{equation*}
%\small
\begin{split}
%\label{eq:dualconvex}
& [D_\eta]:   \quad  
\min_{[\beta_\eta(x)\in \mathbb{R}]_{x\in \mathcal{X}},  [\lambda_\eta(s^k, a^l)\in \mathbb{R}^{0+}]_{k,l\in \mathcal{K}}}  \quad  \sum_{x\in\mathcal{X}} [\beta_\eta(x)+\frac{1}{\eta}]\\
%&  \text{ s.t. } \\
& (a). \quad 
\sum_{x\in \mathcal{X}}  [ \bar{v}_U(x,a^k) - \bar{v}_U(x,a^l)  ]   b_X(x)  \pi_{d}(s^k|x)  \\
&  \quad \quad  \quad \quad \quad
\cdot e^{\eta [ b_X(x)\bar{\beta}_\eta(s^k,x, \lambda_\eta)-\beta_\eta(x)] }  \geq 0,
\forall k,l \in \mathcal{K}, \\
& (b).  \quad    \sum_{k\in \mathcal{K}} \pi_{d}(s^k|x) e^{ \eta [b_X(x) \bar{\beta}_\eta(s^k,x, \lambda_\eta)-\beta_\eta(x)]-1} =1,  \forall x\in \mathcal{X}.
\end{split}
\end{equation*}

\begin{proposition} [\textbf{Strong Duality}]
\label{proposition:dualconvex}
For all $\eta\in \mathbb{R}^{+}$, $D_\eta$ is the dual problem of $P_\eta$, and the optimal value of $D_\eta$ is $r_\eta$. 
%strong duality holds. 
%DCP stands for dual convex problem

\end{proposition}
% \begin{equation*}
% \begin{split}
% & \text{(DLP)}   \quad  
% \min_{\beta_\eta(x)\in \mathbb{R}, \ \lambda_\eta(s^k, a^l) \geq 0} \quad \quad  \sum_{x\in\mathcal{X}} b_X(x) \beta_\eta(x)\\
% &  \text{ s.t. } \quad 
% \beta_\eta(x)\geq \bar{\beta}_\eta(s^k, x,\lambda_\eta),  \forall k\in \mathcal{K} , \forall x\in\mathcal{X}.
% \end{split}
% \end{equation*}

\begin{proof}
Since all constraints in $P_\eta$ are linear, Slater's condition reduces to the feasibility of $D_\eta$ \cite{boyd2004convex}, and strong duality holds. Thus, $P_\eta$ and $D_\eta$ achieve the same optimal value.   
% The Lagrangian of $P_\eta$ is given by
% \begin{equation*}
% \begin{split}
% & L(\pi;\alpha_\eta,\beta,\lambda_\eta)=
%   \sum_{x\in \mathcal{X}} b_X(x) \sum_{k\in \mathcal{K}  }\pi (s^k|x) \bar{v}_D(x,a) -\frac{KL(\pi||\pi^d)}{\eta}  \\ 
% & 
% + \sum_{k\in \mathcal{K}}\sum_{x\in\mathcal{X}}
% \alpha_\eta(s^k,x) \pi (s^k|x)
% - \sum_{x\in\mathcal{X}} \beta_\eta(x) (\sum_{k\in \mathcal{K}  }\pi (s^k|x)-1)
% \\
% & 
% +\sum_{k\in \mathcal{K}}\sum_{l \in \mathcal{K}} \lambda_\eta(s^k,a^l) \bigg[ 
% \sum_{x\in \mathcal{X}}b_X(x)\pi (s^k|x)  [ \bar{v}_U(x,a^k) - \bar{v}_U(x,a^l)  ]
% \bigg]. 
% \end{split}
% \end{equation*}
Setting the gradient of the Lagrangian function of $P_\eta$ concerning $\pi$ to $0$ yields 
%the system of equations in 
% \begin{equation}
% \begin{split}
% %\label{eq:gradient of lagrangian}
%   &  \frac{1}{\eta} (\log \frac{ \pi(s^k|x) }{ \pi_{d}(s^k|x)}+1)=
%      b_X(x) \bar{v}_D(x,a^k) + \alpha_\eta(s^k,x) - \beta_\eta(x)   \\
% & \quad \quad \quad \quad \quad
% +\sum_{l \in \mathcal{K}} \lambda_\eta(s^k, a^l) 
% b_X(x) [ \bar{v}_U(x,a^k) - \bar{v}_U(x,a^l) ],  
% \end{split}
% \end{equation}
$\frac{1}{\eta} (\log \frac{ \pi(s^k|x) }{ \pi_{d}(s^k|x)}+1)=
     b_X(x) \bar{v}_D(x,a^k) + \alpha_\eta(s^k,x) - \beta_\eta(x)   +\sum_{l \in \mathcal{K}} \lambda_\eta(s^k, a^l) 
b_X(x) [ \bar{v}_U(x,a^k) - \bar{v}_U(x,a^l) ], $
for all $k\in \mathcal{K}, x\in \mathcal{X}$, which leads to  
%Thus, for all $k\in \mathcal{K}, x\in \mathcal{X}$, 
\begin{equation}
\label{eq:closedformOptimalpolicy}
    \pi_\eta^*(s^k|x)=\pi_{d}(s^k|x) \cdot e^{ \eta [ b_X(x) \bar{\beta}_\eta(s^k,x, \lambda_\eta)+ \alpha_\eta(s^k,x) - \beta_\eta(x) ] -1 } . 
\end{equation}

Since $\pi_\eta^*(s^k|x)$ in \eqref{eq:closedformOptimalpolicy} is non-negative for all $k\in \mathcal{K},x\in \mathcal{X}$, constraint (a) of $P_\eta$ holds. 
Moreover, the complementary slackness implies the optimal dual variables $\alpha_\eta^*(s^k,x)=0, \forall k\in \mathcal{K},x\in \mathcal{X}$. 
Plugging $\pi_\eta^*(s^k|x)$ in \eqref{eq:closedformOptimalpolicy} into constraints (b) and (c) of $P_\eta$ leads to constraints (a) and (b) of $D_\eta$, respectively. 
 Then, by strong duality, $D_\eta$ minimizes the Lagrangian function $L(\pi_\eta^*,\alpha_\eta^*,\beta_\eta,\lambda_\eta)=\sum_{x\in \mathcal{X}}[\beta_\eta(x)+1/\eta]$ over dual variables ${\beta_\eta(x)\in \mathbb{R},\lambda_\eta(s^k, a^l)\in \mathbb{R}^{0+}}, \forall k,l\in \mathcal{K}, x\in \mathcal{X}$. 
% \begin{equation}
%     \begin{split}
%         \min_{\beta_\eta(x)\in \mathbb{R},\lambda_\eta(s^k, a^l)\geq 0} L(\pi^*,\alpha_\eta^*,\beta,\lambda_\eta)= \min_{\beta_\eta(x)\in \mathbb{R},\lambda_\eta(s^k, a^l)\geq 0} \sum_{x\in\mathcal{X}} [\beta_\eta(x)+\frac{1}{\eta}]
%     \end{split}
% \end{equation}
\end{proof}

%Note that if we plug (b) into (a), we also need to plug (b) into the objective function! 
% \begin{remark}
% Note that constraints (a) and (b) of DCP together imply the softmax intepretation. \begin{equation*}
%     \begin{split}
%         \sum_{x\in \mathcal{X}}  [ \bar{v}_U(x,a^k) - \bar{v}_U(x,a^l)  ]   b_X(x) 
%         \frac{\pi_{d}(s^k|x)  \cdot e^{\eta b_X(x)\bar{\beta}_\eta(s^k,x, \lambda_\eta) } }{\sum_{k'\in \mathcal{K}} \pi_{d}(s^{k'}|x)  \cdot e^{\eta b_X(x)\bar{\beta}_\eta(s^{k'},x,\lambda_\eta) }  }
%          \geq 0,
%     \end{split}
% \end{equation*}
% for all $k\in \mathcal{K} , a^l \in \mathcal{A}$.
% \end{remark}

\begin{proposition}[\textbf{Bounds of the Optimal Value}]
\label{proposition:bound}
The lower and upper bounds of $r_\eta$ %for the optimal value of $D_\eta$ (or equivalently $P_\eta$) 
are $\underline{r}$ and $\underline{r}+\log(K)/\eta$, respectively. 
\end{proposition}
%$\underline{r}:=\sum_{x\in \mathcal{X}} [ \max_{k\in \mathcal{K}}  b_X(x)\bar{\beta}_\eta(s^k,x, \lambda_\eta)+\log(\pi_{d}(s^k|x))/\eta ]$

\begin{proof}
Constraint (b) in $D_\eta$ is equivalent to the log-sum-exp expression: 
$
    \beta_\eta(x)=\frac{1}{\eta} \log ( \sum_{k\in \mathcal{K}} \pi_{d}(s^k|x) e^{\eta     b_X(x) \bar{\beta}_\eta(s^k,x, \lambda_\eta) -1 } ). 
$
% or
% \begin{equation}
%     e^{\eta \beta_\eta(x)}= \sum_{s^k} \pi_{d}(s^k|x) e^{ \eta b_X(x) \bar{\beta}_\eta(s^k,x, \lambda_\eta) -1}, \forall k\in \mathcal{K}, x\in \mathcal{X}. 
% \end{equation}
Thus, 
$
    \max_{k\in \mathcal{K}} \eta b_X(x)\bar{\beta}_\eta(s^k,x, \lambda_\eta)-1+\log(\pi_{d}(s^k|x))\leq \eta \beta_\eta(x) \leq \max_{k\in \mathcal{K}} \eta b_X(x)\bar{\beta}_\eta(s^k,x, \lambda_\eta)-1+\log(\pi_{d}(s^k|x)) + \log(K)
$  for all $x\in\mathcal{X}$. 
%The first inequality is strict unless $M=1$ and the second inequality is strict unless all elements are equal. 
Since strong duality holds, we obtain the bounds for $r_\eta$ in $P_\eta$. 
%Note that the bound depends on dual variable $\lambda_\eta$. When eta goes to inf, the bound is not useful as it is equivalent to the linear constraint of DLP. 
\end{proof}

% \begin{remark}
% \label{remark:closed form pi}
Following \eqref{eq:closedformOptimalpolicy} and Proposition \ref{proposition:bound}, the optimal policy $\pi_\eta^*$ has the closed-form expression in \eqref{eq:optimalpolicy_softmax} concerning the optimal dual variables $\lambda^*_\eta(s^k,a^l)\in \mathbb{R}^{0+},l,k\in \mathcal{K}$, and the default recommendation policy $\pi_d\in \Pi$; i.e., for all $x\in \mathcal{X},s^k\in \mathcal{S}, k\in \mathcal{K}$, 
 \begin{equation}
 %\small
        %   \pi^*(s^k|x)=  \frac{ \pi_{d}(s^k|x) \cdot e^{ \eta b_X(x) \bar{\beta}_\eta(s^k,x, \lambda_\eta) -1} }{
        %  \sum_{k\in \mathcal{K}} \pi_{d}(s^k|x) \cdot e^{ \eta b_X(x) \bar{\beta}_\eta(s^k,x, \lambda_\eta) -1}
        %  }. 
     \label{eq:optimalpolicy_softmax}
         \pi_\eta^*(s^k|x)=  \frac{ \pi_{d}(s^k|x) \cdot e^{ \eta b_X(x) \bar{\beta}_\eta(s^k,x, \lambda^*_\eta) } }{
         \sum_{k\in \mathcal{K}} \pi_{d}(s^k|x) \cdot e^{ \eta b_X(x) \bar{\beta}_\eta(s^k,x, \lambda^*_\eta) }
         }. % \forall x\in \mathcal{X},k\in \mathcal{K} 
 \end{equation}
% \end{remark}
%Note that when $\eta\rightarrow \infty$, $\pi^*$ does not necessarily be a pure policy as the optimal $\lambda_\eta$ will make some terms indifferent. 

\subsection{Interpretation of ZETAR from Insiders' Perspectives}
When ZETAR designs a fully customized recommendation policy (i.e., LoRC $\eta$ goes to infinity), then the dual problem  $D_{\infty}$ is a linear program as shown in Proposition \ref{proposition:dualLP}. 
Define $\hat{\beta}_{\infty}(x):=\beta_{\infty}(x)/b_X(x)$ where $\hat{\beta}_{\infty}(x)=0$ if $b_X(x)=0$. 
%this definition has no problem as $\beta_{\infty}(x)=0$ if $b_X(x)=0$ based on the LP constraint. 

\begin{proposition}
\label{proposition:dualLP}
When LoRC $\eta$ goes to infinity, the dual problem $D_{\infty}$ degenerates to the following linear program: 
\begin{equation*}
%\small
\begin{split}
& [D_{\infty}]:    \quad  
\min_{[\hat{\beta}_{\infty}(x)\in \mathbb{R}]_{x\in \mathcal{X}}, [\lambda_{\infty}(s^k, a^l)\in \mathbb{R}^{0+}]_{k,l\in \mathcal{K}}}  \quad \quad  \sum_{x\in\mathcal{X}} b_X(x)\hat{\beta}_{\infty}(x)\\
&  \text{ s.t. } \quad 
\hat{\beta}_{\infty}(x)\geq \bar{\beta}_{\infty}(s^k, x,\lambda_{\infty}),  \forall k\in \mathcal{K} , \forall x\in\mathcal{X}. 
\end{split}
\end{equation*} 

\end{proposition}
% Define $\hat{\beta}(x):=\beta_\eta(x)/b_X(x)$ and rewrite $[D_{\infty}]$ as 
% \begin{equation*}
% \begin{split}
%     %& \max_{ \lambda_\eta \geq 0}  \sum_{x\in\mathcal{X}} b_X(x) \cdot \min_{k\in \mathcal{K}} - \bar{v}_D(x,a^k) + \sum_{a^l \in \mathcal{A} } \lambda_\eta(s^k, a^l) [ \bar{v}_U(x,a^l) - \bar{v}_U(x,a^k)]
%   \min_{ \lambda_\eta \geq 0}  \sum_{x\in\mathcal{X}} b_X(x) \cdot \max_{k\in \mathcal{K}} \bar{\beta}_\eta(s^k, x,\lambda_\eta). 
% \end{split}
% \end{equation*}

\begin{proof}
When $\eta\rightarrow \infty$, the upper and lower bounds for $\beta_\eta(x)$ in Proposition \ref{proposition:bound} attain the same value, which implies  $\beta_\eta(x)= \max_{k\in \mathcal{K}} \eta b_X(x)\bar{\beta}_\eta(s^k,x, \lambda_\eta)$. 
Thus, constraint (a) of $D_\eta$ is feasible. Constraint (b) and the objective function of the convex program $D_\eta$ are equivalent to the constraint and the objective function of the linear program $D_\infty$, respectively. 
\end{proof}

%The duality of BP and Information design have been studied in \cite{kolotilin2018optimal,galperti2018dual,dworczak2019persuasion}. 
The dual problem $D_{\infty}$ provides an interpretation of ZETAR with fully customized recommendation policies from an insider's perspective; i.e., each insider aims to minimize his effort to satisfy the security objective of the corporate network. 
%As a dual of the defender's design problem, the insider aims to minimize his effort to  satisfy the corporate security objective. 
Variable $\lambda_{\infty}(s^k,a^l)$ represents the insider's frequency to take action $a^l\in \mathcal{A}$ under recommendation $s^k\in \mathcal{S}$. 
%it has to be some thing that can be determined by the insider. 
%the attractiveness of action $a^l\in \mathcal{A}$ to the insider under recommendation $s^k\in \mathcal{S}$. 
The variable $\bar{\beta}_{\infty}(s^k, x,\lambda_{\infty})$ represents the mixed security objective of the corporate network under AS $x\in \mathcal{X}$ and recommendation $s^k\in \mathcal{S}$,  %, and the frequency $\lambda_{\infty}(s^k,a^l)\in \mathbb{R}^{0+}$ to take action $a^l\in \mathcal{A}$. 
%The mixed security objective 
which involves the sum of the defender's utility $\bar{v}_D$ and the insider's expected utility, i.e., $\sum_{a^l \in \mathcal{A} } \lambda_{\infty}(s^k, a^l) [ \bar{v}_U(x,a^k) - \bar{v}_U(x,a^l)]$. 
The variable $\hat{\beta}_{\infty}(x)$ represents the insider's effort at AS $x\in \mathcal{X}$, and the effort is required to satisfy the security objective at each AS for all recommendations. An insider who prioritizes convenience over security chooses the rate of actions to minimize his expected effort $\sum_{x\in\mathcal{X}} b_X(x)\hat{\beta}_{\infty}(x)$. 

\section{Characterization of Trust and Compliance}%Characterization of Trustworthy Recommendation Policies and Compliance Status
\label{sec:characterize two}
Section \ref{sec:Computational Framework of ZETAR} provides a unified computational framework to design the optimal CT recommendation policy under any LoRC. 
In this section, we consider fully customized recommendation policies, i.e., $\eta=\infty$. 
%Section \ref{sec:Geometric Characterization of CT Policy Sets} and \ref{sec:Characterizing Compliance under Misalignment}
We characterize the invariance of an insider's compliance status and the defender's optimal recommendation policy under linear utility transformations in Section \ref{sec:Linear Utility Transformation}. %of the insider's incentive and the defender's security objectives. 
In Section \ref{sec:Geometric Characterization of CT Policy Sets}, we provide a geometric characterization of the CT policy set based solely on an insider's incentive $v_U$. The characterizations are useful to develop efficient algorithms in Section \ref{sec:Model-Agnostic} when insiders' incentives are unknown. 
%We provide a graphical illustration of these characterizations in Section \ref{sec:Graphical Illustration of Learning Algorithms} when $K=I=J=2$. 
In Section \ref{sec:Characterizing Compliance under Misalignment}, we characterize the optimal ACEL under different levels of misalignment between the defender's security objective $v_D$ and an insider's incentive $v_U$. 

\subsection{Impact of Linear Utility Transformations}
\label{sec:Linear Utility Transformation}
We define the linear utility transformation for the defender and an insider in Definition \ref{def:LT}. 
Following Remark \ref{remark:compliance and trustworthy}, if a recommendation $s\in\mathcal{S}$ is trustworthy (or untrustworthy) to both two insiders, then they have the same compliant status under $s$. 
Lemma \ref{lemma:LIinsider} illustrates the preservation of an insider's compliance status under linear transformations of $v_U$.  The proof directly follows from Definition \ref{def:trustworthS}.

\begin{definition}[\textbf{Linear Utility Transformation}]
\label{def:LT}
Define the linear transformation of a player's utility with a scaling factor $\rho_p^{sa}\in \mathbb{R}$ and translation factors $[\rho_p^{tr}(y,x)\in \mathbb{R}]_{x\in\mathcal{X},y\in \mathcal{Y}}$ as $v_p^{lt}(y,x,a):=\rho_p^{sa} v_p(y,x,a)+\rho_p^{tr}(y,x)$ for all $x\in\mathcal{X},y\in \mathcal{Y},a\in \mathcal{A}, p\in \{D,U\}$. %\rho_p^{sa},\rho_p^{tr}
\end{definition}

\begin{lemma}[\textbf{Preservation of Compliance Status}]
\label{lemma:LIinsider}
%Consider a linear transform of the insider's utility defined in Definition \ref{def:LT}. 
%Then, 
Interacting with the same defender, two insiders with incentives $v_U$ and $v_U^{lt}$, respectively, have the same compliance status for all recommendation $s\in \mathcal{S}$ under any recommendation policy $\pi\in \Pi$. 
Moreover, the defender applies the same optimal recommendation policy to both insiders. 
\end{lemma}
% \begin{lemma}[\textbf{Linear Invariant}]
% \label{lemma:LIinsider}
% Consider a linear transform of the insider's utility defined in Definition \ref{def:LT}. 
% The following two statements hold: 
% \begin{itemize}
%     \item[a.] The insiders with utilities $v_U$ and $v_U^{lt}$ have the same compliant status with any recommendation policy $\pi\in \Pi$.
%     \item[b.] The EPUs of the insiders with utilities $v_U$ and $v_U^{lt}$ are linearly dependent, i.e., 
%     $J_U(\pi,b_{Y,X}, v^{lt}_U)=\rho_U^{sa} J_U(\pi,b_{Y,X}, v_U)+\sum_{x\in\mathcal{X}}\sum_{y\in \mathcal{Y}} b_{Y,X}(y,x) \rho_U^{tr}(y,x)$. 
% \end{itemize}
% \end{lemma}
% \begin{proof}
% The proof directly follows from Definition \ref{def:trustworthS}.  %from \eqref{eq:optimalAction} and 
% \end{proof}

Lemma \ref{lemma:LIdefender} characterizes ZETAR from the perspective of linear systems; i.e., a linear transformation of $v_D$ results in a linear transformation of the defender's ASeL in \eqref{eq:def_obj} and the ACEL in Definition \ref{def:ACEL} for any  recommendation policy $\pi\in \Pi$. 
The proof directly follows from \eqref{eq:def_obj} and the fact that  
$\max_{\pi\in \Pi} J_D(\pi,b_{X}, \bar{v}^{lt}_D,\bar{v}_U)=\max_{\pi\in \Pi} J_D(\pi,b_{X}, \bar{v}_D,\bar{v}_U)$. 
\begin{lemma}[\textbf{Preservation of Linearity}]
\label{lemma:LIdefender}
%Consider a linear transform of the defender's utility with a scaling factor $\rho_D^{sa}\in \mathbb{R}$ and translation factors $\rho_D^{tr}(y,x)\in \mathbb{R}, \forall x\in\mathcal{X},y\in \mathcal{Y}$, i.e., $v_D^{lt}(y,x,a)= \rho_D^{sa} v_D(y,x,a)+\rho_D^{tr}(y,x)$ for all $x\in\mathcal{X},y\in \mathcal{Y},a\in \mathcal{A}$. 
%Consider a linear transform of the insider's utility defined in Definition \ref{def:LT}. 
Interacting with the same insider, the ASeL of two defenders with security objectives $v_D$ and $v_D^{lt}$ are linearly dependent, i.e.,    $J_D(\pi,b_{X}, \bar{v}^{lt}_D,\bar{v}_U)=\rho_D^{sa} J_D(\pi,b_{X}, \bar{v}_D,\bar{v}_U)+\sum_{x\in\mathcal{X}}\sum_{y\in \mathcal{Y}} b_{Y,X}(y,x) \rho_D^{tr}(y,x)$, for all $\pi\in \Pi$. %including the optimal pi
Moreover, the two defenders use the same optimal recommendation policy. 
\end{lemma}
% \begin{proof}
% %The same insider leads to the same best-response action under any $\pi\in \Pi$. 
% The proof directly follows from \eqref{eq:def_obj} and the fact that  
% $\max_{\pi\in \Pi} J_D(\pi,b_{X}, \bar{v}^{lt}_D,\bar{v}_U)=\max_{\pi\in \Pi} J_D(\pi,b_{X}, \bar{v}_D,\bar{v}_U)$. 
% \end{proof}
%Note that this lemma is not true with $\eta$ and $\pi^d$. 
\begin{remark}[\textbf{Policy Invariance}]
\label{remark:Policy Invariance}
Lemmas \ref{lemma:LIinsider} and \ref{lemma:LIdefender} show that linear utility transformation does not affect the optimal recommendation policy. 
The structures of an insider's incentive and the defender's security objective play more critical roles in compliance status and ASeL than their absolute values. 
\end{remark}

\subsection{Geometric Characterization of CT Sets, ASaL, and ASeL}%of CT Policy Sets
\label{sec:Geometric Characterization of CT Policy Sets}

 %$\hat{v}_i(x,a):=b_X(x)\bar{v}_p(x,a)$ for all $x\in \mathcal{X}, a\in \mathcal{A}, p\in \{D,U\}$
%To simplify, we introduce the notation as follows. 
We define the following notations for the matrix representations of recommendation policies and utilities in Section \ref{sec:Geometric Characterization of CT Policy Sets} and \ref{sec:Model-Agnostic}.  
%Define the following vectors for the matrix representation. 
Define $\hat{v}_p^k:=[b_X(x^1)\bar{v}_p(x^1,a^k), \allowbreak b_X(x^2)\bar{v}_p(x^2,a^k), \allowbreak \cdots, \allowbreak b_X(x^I)\bar{v}_p(x^I,a^k)]^T\in \mathbb{R}^{1\times I}, p\in \{D,U\}$, for all $k\in \mathcal{K}$. 
%Define $\hat{b}^i:=[b(y^1,x^i),\cdots,b(y^J,x^i))]$ and $\hat{v}_U^{i,k}:=[v_U(y^1, x^i,a^k), \cdots,v_U(y^J, x^i,a^k)]'$ for all $i\in \mathcal{I}$ and $k\in \mathcal{K}$. Then, let $\hat{v}_U^k=[\hat{b}^1\hat{v}_U^{1,k},\hat{b}^2 \hat{v}_U^{2,k}, \cdots, \hat{b}^I\hat{v}_U^{I,k}]'$ for all $k\in \mathcal{K}$. 
For each recommendation $s^k\in\mathcal{S}, k\in \mathcal{K}$, and the shorthand notation $\hat{\pi}^{k,i}:=\pi(s^k|x^i)$, we can define an $I$-dimension vector $\hat{\pi}^{k}=[\hat{\pi}^{k,1},\cdots,\hat{\pi}^{k,I}]\in \hat{\Pi}^k$. By definition, $\sum_{k=1}^K \hat{\pi}^{k}=[1,1,\cdots,1]\in \mathbb{R}^I$.  
%where $\mathbf{1}^I$ represents a all-ones vector $1$-vector of size $I$. 
Then, a recommendation policy $\pi\in \Pi$ has an equivalent matrix form as $\hat{\pi}:=[\hat{\pi}^{1},\cdots,\hat{\pi}^{K}]^T\in \hat{\Pi}$. 
Analogously, the $k$-th PT, $k$-th PU, and CT recommendation policies in matrix forms compose sets $\hat{\Pi}_{pt}^{k}$, $\hat{\Pi}_{pu}^{k}$, and $\hat{\Pi}_{ct}$, respectively. 
In Proposition \ref{proposition:separability}, we identify $\hat{\pi}^k\in\hat{\Pi}^k$ as the sufficient component of $\hat{\pi}\in \hat{\Pi}$ to determine the trustworthiness of recommendation $s^k\in \mathcal{S}$.
%Since it is sufficient to determine the trustworthiness of a recommendation based on the $k$-th row of the policy matrix, we say a recommendation is trustworthy under $\hat{\pi}^k$ for short in the following sections. 

\begin{proposition}[\textbf{Minimal Sufficiency for Trustworthy Recommendations}]
\label{proposition:separability}
Policy vector $\hat{\pi}^k\in\hat{\Pi}^k$ is the minimal sufficient component of the policy matrix $\hat{\pi}\in \hat{\Pi}$ to determine the trustworthiness of recommendation $s^k\in \mathcal{S}$. 
%insider's compliance with the recommendation $k\in \mathcal{K}$. 
\end{proposition}
\begin{proof}
Based on \eqref{eq:optimalAction} and Definition \ref{def:trustworthS}, a recommendation $s^k\in \mathcal{S}, k\in \mathcal{K}$, is trustworthy if and only if $\hat{\pi}^k [\hat{v}_U^k -\hat{v}_U^l]\geq 0, \forall l\in \mathcal{K}$ (i.e., the matrix representation of constraint (c) in $P_\eta$). 
%Thus, $\hat{\pi}^k\in \hat{\Pi}^k$ is the minimal sufficient statistic. 
\end{proof}

Proposition \ref{proposition:separability} leads to the policy separability principle in Remark \ref{remark:separable}; i.e., the defender can design the $k$-th policy vector  $\hat{\pi}^k\in \hat{\Pi}^k$ separately for all $k\in \mathcal{K}$ to learn the $k$-th PT policy set.  
%CT policy set. 
%design CT recommendation policies by designing the $k$-th policy vector of the $k$-th PT policy sets, i.e., $\hat{\pi}^k\in \hat{\Pi}^k$, sequentially for all $k\in \mathcal{K}$. 
The policy separability contributes to efficient CT policy set learning algorithms in Section \ref{sec:Model-Agnostic}.  
We characterize an insider's ASaL, the convexity of the CT policy set, and  the defender's ASeL in Lemmas \ref{lemma:PWLC of the insider's EPU}-\ref{lemma:PWLC of the Defender's EPU}, respectively. 
Section \ref{sec:impacts of recommendation policies} illustrates these characterizations when $I = J = K = 2$. 
%We provide a graphical illustration of these characterizations in Section \ref{sec:impacts of recommendation policies}. 
%\ref{lemma:convex}. 

\begin{remark}[\textbf{Policy Separability}] %Separable Design for CT Recommendation Policy
\label{remark:separable}
The defender can determine the $k$-th PT policy set, i.e., $\hat{\Pi}_{pt}^k$, independently from other PT policy sets $\hat{\Pi}_{pt}^{k'}, \forall k'\in \mathcal{K}\setminus \{k\}$, to determine CT policy set $\hat{\Pi}_{ct}$. 
%The defender can determine the $k$-th policy vector of the $k$-th PT recommendation policy, i.e., $\hat{\pi}^k\in \hat{\Pi}_{pt}^k$, independently from $\hat{\pi}^{k'}\in \hat{\Pi}_{pt}^{k'}, \forall k'\in \mathcal{K}\setminus \{k\}$, to design a CT recommendation policy $\hat{\pi}\in \hat{\Pi}_{ct}$. 
\end{remark}

% \begin{equation}
% \label{eq:insiderEPU}
% \begin{split}
%     &J_U(\pi,b_{X},\bar{v}_U)=\sum_{s\in \mathcal{S}} b^{\pi}_S(s) \max_{a\in \mathcal{A}} \mathbb{E}_{y,x\sim b_{Y,X}^{\pi}(\cdot|,s)} [ {v}_U (y,x,a)]\\
%     & = \sum_{k\in \mathcal{K}}  \max_{a\in \mathcal{A}} [\sum_{i\in \mathcal{I}}\sum_{j\in \mathcal{J}} b_{Y,X}(y^j,x^i)\pi(s^k|x^i)v_U(y^j,x^i,a)]\\
%     & = \sum_{k\in \mathcal{K}}  \max_{a\in \mathcal{A}} [\sum_{x\in \mathcal{X}}b_X(x) \pi(s^k|x)\bar{v}_U(x,a)]\\
%     &=\sum_{x\in \mathcal{X}}b_X(x) \sum_{s\in \mathcal{S}} \pi(s|x)\bar{v}_U(x,a_{\pi,s}^*). 
% \end{split}
% \end{equation}

\begin{lemma}[\textbf{PWL and Convex of ASaL}] %the PWLC is different from the duplicity game paper
\label{lemma:PWLC of the insider's EPU}
\textcolor{black}{For any recommendation policy $\hat{\pi}\in \hat{\Pi}$, an insider's} ASaL $J_U(\pi,b_{X}, \bar{v}_U)$ is PieceWise Linear (PWL) and convex in $\hat{\pi}^{k}\in \hat{\Pi}^{k}, \forall k\in \mathcal{K}$. 
\end{lemma}
\begin{proof}
Following \eqref{eq:insiderEPU}, we can represent an insider's ASaL as $J_U(\pi,b_{X},\bar{v}_U)=\sum_{k\in \mathcal{K}}  \max_{a\in \mathcal{A}} [\sum_{x\in \mathcal{X}}b_X(x) \pi(s^k|x)\bar{v}_U(x,a)]$. 
Since 
%$ \sum_{i\in \mathcal{I}}\sum_{j\in \mathcal{J}} b_{Y,X}(y^j,x^i)\allowbreak \pi(s^k|x^i) \allowbreak v_U(y^j,x^i,a)$ 
$\sum_{x\in \mathcal{X}}b_X(x) \pi(s^k|x)\bar{v}_U(x,a)$ is an linear function in $\hat{\pi}^{k}\in \hat{\Pi}^{k}, \forall k\in \mathcal{K}$, and $a\in \mathcal{A}$, the point-wise maximum of a group of linear functions in \eqref{eq:insiderEPU} leads to a PieceWise Linear (PWL) and convex function concerning $\hat{\pi}^{k}\in \hat{\Pi}^{k}, \forall k\in \mathcal{K}$. Then, the sum of a group of PWL and convex functions remains PWL and convex.  
\end{proof}

Denote $\mathcal{C}_{l}^{k}:=\{\hat{\pi} \in \hat{\Pi} | \hat{\pi}^k [\hat{v}_U^l -\hat{v}_U^h] \geq 0, \forall h\in \mathcal{K}\}$ as the set of recommendation policies that induce action $a^l\in \mathcal{A}$ under recommendation $s^k\in \mathcal{S}$. 
Define $\mathcal{C}_{\{l_1,\cdots,l_K\}}:= \cap_{k=1}^K \mathcal{C}_{l_k}^{k}$ as the set of  recommendation policies that induce action $a^{l_k}\in \mathcal{A}$ under recommendation $s^k\in \mathcal{S}$ for all $k\in \mathcal{K}$.  
Within each (possibly empty) set $\mathcal{C}_{\{l_1,\cdots,l_K\}}, \forall l_1,\cdots,l_K\in \mathcal{K}$, we can represent the defender's ASeL in \eqref{eq:def_obj} equivalently as the matrix form $\hat{J}_D(\hat{\pi}, \hat{v}_D, v_U) := \sum_{k\in \mathcal{K}} \hat{\pi}^k \hat{v}_D^{l_k}, \ \forall \hat{\pi}\in \mathcal{C}_{\{l_1,\cdots,l_K\}}$. 

\begin{lemma}
\label{lemma:convex}
The $K^K$ sets $\mathcal{C}_{\{l_1,\cdots,l_K\}}, \forall l_1,\cdots,l_K\in \mathcal{K}$, are mutually exclusive and convex. 
The union of these sets composes the entire recommendation policy set, i.e., $\hat{\Pi}=\cup_{l_1,\cdots,l_K\in \mathcal{K}} \mathcal{C}_{\{l_1,\cdots,l_K\}}$. 
%Sets $\mathcal{C}_{l}^{k}, \forall l,k\in \mathcal{K}$, the 
The $k$-th PT policy set $\mathcal{C}_{k}^{k}, \forall k\in \mathcal{K}$, and the CT policy set $\hat{\Pi}_{ct}=\mathcal{C}_{\{1,\cdots,K\}}= \cap_{k=1}^K \mathcal{C}_{k}^{k}$ are convex. 
% The intersection leads to the CT policy set, i.e., $\Pi_{ct}=\mathcal{C}_{\{1,\cdots,K\}}= \cap_{k=1}^K \mathcal{C}_{k}^{k}$, which is convex.  
\end{lemma}
\begin{proof}
The convexity of set $\mathcal{C}_{l}^{k}$ directly follows its definition. 
The properties of the mutual exclusiveness and the union directly come from the definition of $\mathcal{C}_{\{l_1,\cdots,l_K\}}$. 
Definition \ref{def:trustworthyPi} leads to $\hat{\Pi}_{ct}=\cap_{k=1}^K \mathcal{C}_{k}^{k}$. 
Since the intersection of any collection of convex sets is convex, sets $\mathcal{C}_{\{l_1,\cdots,l_K\}}$ and $\hat{\Pi}_{ct}$ are convex. 
\end{proof}
% \begin{remark}
% The sets $\mathcal{C}_{\{l_1,\cdots,l_K\}},l_1,\cdots,l_K\in \mathcal{K}$, only depends on the insider's incentives $v_U$. 
% \end{remark}

% \begin{equation}
% \label{eq:defenderEPU}
%     \hat{J}_D(\hat{\pi}, \hat{v}_D, v_U) = \sum_{k\in \mathcal{K}} \hat{\pi}^k \hat{v}_D^{l_k}, \ \forall \hat{\pi}\in \mathcal{C}_{\{l_1,\cdots,l_K\}}. 
%     %J_D(\pi,b,\psi, v_D) = \sum_{k\in \mathcal{K}} \sum_{i\in \mathcal{I}}\sum_{j\in \mathcal{J}} b_{Y,X}(y^j,x^i)\pi(s^k|x^i)v_D(y^j,x^i,a^{l_k}), \forall \pi\in \mathcal{C}_{\{l_1,\cdots,l_K\}}. 
% \end{equation}

\begin{lemma}[\textbf{PWL of ASeL}] %the PWLC is different from the duplicity game paper
\label{lemma:PWLC of the Defender's EPU}
\textcolor{black}{For any recommendation policy $\hat{\pi}\in \hat{\Pi}$, the defender's} ASeL $\hat{J}_D(\hat{\pi},\hat{v}_D, v_U)$ is (possibly discontinuous)  piecewise linear concerning $\hat{\pi}^{k}\in \hat{\Pi}^{k}, \forall k\in \mathcal{K}$. 
\end{lemma}
\begin{proof}
Based on Lemma \ref{lemma:convex}, the entire recommendation policy set $\hat{\Pi}$ is divided into $K^K$ mutually exclusive (possibly empty) sets determined by an insider's incentive $v_U$.  
Within each set $\mathcal{C}_{\{l_1,\cdots,l_K\}}$, the defender's ASeL $\hat{J}_D(\hat{\pi}, \hat{v}_D, v_U)$ in matrix form is linear in $\hat{\pi}^{k}\in \hat{\Pi}^{k}, \forall k\in \mathcal{K}$. %, as shown in \eqref{eq:defenderEPU}. 
\end{proof}

\subsection{Optimal ACEL under Incentive Misalignment}%Compliance under Incentive Misalignment %for Amendable and Malicious insiders
\label{sec:Characterizing Compliance under Misalignment}
%According to the convex program formulation in Section \ref{sec:Computational Framework of ZETAR}, enlarging the CT policy set will not reduce the defender's EPU or the ACEL. 
%However, we have characterized the fundamental limit of the cardinality of the CT policy set in Section \ref{sec:Geometric Characterization of CT Policy Sets} and have shown that the maximum cardinality is achieved only at a restrictive condition of prior indifference in Corollary \ref{Coro: prior indifference}. 
%When two CT policy sets have the same cardinality, then the elements of two CT policy sets lead to different ACEL. 

We first classify the insiders into three incentive categories in Definition \ref{def:three category} based on the alignment of their incentives with the defender's security objective. 
%Amendable insiders are naturally willing to follow security rules and have a strong sense of responsibility to enhance security. 
%Malicious insiders, e.g., the insider whose computer has been hacked, use their inside access privilege for fraud, sabotage, espionage, revenge, and blackmail. 
%Others are self-interested and prioritize their convenience or benefits over corporate security. 
%Following Section \ref{sec:utilities}, the defender's utility evaluates the impact of an action on corporate security. For a given $\bar{v}_D$, 
%We classify the insiders into the following three categories based on their action preference in Definition \ref{def:three category}. 
%A cooperative insider, e.g., a network administrator, treats corporate security as his own goal. 
Denote $\chi(\mathcal{K}):=[{\chi(1)},{\chi(2)},\cdots,{\chi(K)}]$ as a permutation of set $\mathcal{K}$, i.e., $\chi(k)\in \mathcal{K},\forall k\in \mathcal{K}$, and $\chi(k)\neq \chi(k')$ if $k\neq k',\forall k,k'\in \mathcal{K}$. 
\begin{definition} [\textbf{Incentive Categories}]
\label{def:three category}
Consider the defender with security objective ${v}_D$.  
An insider is categorized as amenable (resp. malicious) if he shares the same (resp. opposite) preference ranking with the defender concerning actions for each SP and AS; i.e., for any given $x\in \mathcal{X}, y\in \mathcal{Y}$, if ${v}_U(y,x,a^{\chi(1)})\geq {v}_U(y,x,a^{\chi(2)})\geq \cdots\geq {v}_U(y,x,a^{\chi(K)})$, then  ${v}_D(y,x,a^{\chi(1)})\geq (resp. \leq){v}_D(y,x,a^{\chi(2)})\geq (resp. \leq)\cdots\geq (resp. \leq) {v}_D(y,x,a^{\chi(K)})$. 
An insider is self-interested if he is neither amenable nor malicious. 
\end{definition}

An amenable insider has a strong sense of responsibility to enhance security and prioritizes security over convenience. 
A malicious insider, e.g., a disgruntled insider or an insider whose credentials have been stolen, can misbehave or sabotage corporate security on purpose. 
Self-interested insiders represent the  majority of insiders who are willing to follow security rules when there is no conflict of interests. %only
Following Definition \ref{def:LT} and Remark \ref{remark:Policy Invariance}, linear transformations of a malicious, self-interested, or amenable insider's incentive do not change his incentive category. 
Lemma \ref{lemma:linearutility0} characterizes the optimal recommendation policy and the ACEL when an insider's incentive and the defender's security objective are linearly dependent. 

\begin{lemma}
\label{lemma:linearutility0}
Consider linearly dependent incentives of an insider and the defender with a scaling factor ${\rho}_{D,U}^{sa}\in \mathbb{R}$ and translation factors $[{\rho}_{D,U}^{tr}(y,x)\in \mathbb{R}]_{y\in \mathcal{Y}, x\in \mathcal{X}}$, i.e., ${v}_D(y,x,a)={\rho}_{D,U}^{sa} {v}_U(y,x,a)+{\rho}_{D,U}^{tr}(y,x)$ for all $y\in \mathcal{Y}, x\in \mathcal{X},  a\in\mathcal{A}$.
Then, the following two statements hold. 
\begin{enumerate}
    \item  $J_D^{acel,*}(b_{X},\bar{v}_D,\bar{v}_U)=0, \forall b_{X}\in \mathcal{B}_{X}$, if and only if ${\rho}_{D,U}^{sa}\leq 0$. Zero-information recommendation policy $\pi_z\in \Pi_{ct}$ achieves the optimal ACEL. %of $J_D^{acel,*}(b_{X},\bar{v}_D,\bar{v}_U)=0$. 
    \item $J_D^{acel}(\pi,b_{X},\bar{v}_D,\bar{v}_U)\geq 0, \forall \pi\in \Pi, \forall b_{X}\in \mathcal{B}_{X}$,  if and only if ${\rho}_{D,U}^{sa} > 0$. 
    Full-information recommendation policy $\pi_f\in \Pi_{ct}$ achieves the optimal ACEL, and we have  $J_D^{acel,*}(b_{X},\bar{v}_D,\bar{v}_U)={\rho}_{D,U}^{sa} \sum_{x\in \mathcal{X}}b_X(x)  \max_{a\in\mathcal{A}} \bar{v}_U(x,a)-{\rho}_{D,U}^{sa} \max_{a\in \mathcal{A}} \sum_{x\in \mathcal{X}} b_{X}(x)  \bar{v}_U(x,a)$. 
    %and $J_D(\pi_z,b_{X},\bar{v}_D,\bar{v}_U)$ is PWL and convex concerning $b_X$
    %If $K\geq I$, then any full-information recommendation policy achieves the optimal ACEL. If $K<I$, then for each $x\in \mathcal{X}$, there exist $s\in \mathcal{S}$ such that $\pi^*(s|x)=1$. 
\end{enumerate}
\end{lemma}

% \begin{lemma}
% \label{lemma:linearutility0}
% Consider linearly dependent expected incentives of an insider and the defender with a scaling factor ${\rho}_{D,U}^{sa}\in \mathbb{R}$ and translation factors ${\rho}_{D,U}^{tr}(y,x)\in \mathbb{R}, \forall x\in \mathcal{X},y\in \mathcal{Y}$, i.e., ${v}_D(y,x,a)={\rho}_{D,U}^{sa} {v}_U(y,x,a)+{\rho}_{D,U}^{tr}(y,x)$ for all $x\in \mathcal{X},  a\in\mathcal{A}$.
% Then, the following two statements hold. 
% \begin{enumerate}
%     \item  $J_D^{acel,*}(b_{Y,X},{v}_D,{v}_U)=0, \forall b_{Y,X}\in \mathcal{B}_{Y,X}$, if and only if ${\rho}_{D,U}^{sa}\leq 0$. zero-information recommendation policy $\pi_z$ achieves the optimal ACEL. 
%     \item $J_D^{acel}(\pi,b_{Y,X},{v}_D,{v}_U)\geq 0, \forall \pi\in \Pi, \forall b_{Y,X}\in \mathcal{B}_{Y,X}$,  if and only if ${\rho}_{D,U}^{sa} > 0$. 
%     Full-Information CT recommendation policy $\pi_f$ achieves the optimal ACEL  $J_D^{acel,*}(b_{Y,X},{v}_D,{v}_U)={\rho}_{D,U}^{sa} \sum_{x\in \mathcal{X},y\in \mathcal{Y}}b_{Y,X}(y,x)  \max_{a\in\mathcal{A}} {v}_U(y,x,a)-{\rho}_{D,U}^{sa} \max_{a\in \mathcal{A}} \sum_{x\in \mathcal{X},y\in \mathcal{Y}} b_{Y,X}(y,x)  {v}_U(y,x,a)$. 
% \end{enumerate}
% \end{lemma}

\begin{proof}
Under $\pi_z\in \Pi_{ct}$ and the linear dependency condition, the defender's ASeL in \eqref{eq:def_obj} becomes
$
        \tilde{J}_D(\pi_z,b_{Y,X}, {v}_D,{v}_U) \allowbreak
        %=\sum_{y\in \mathcal{Y}}\sum_{x\in \mathcal{X}} b_{Y,X}(y,x) \sum_{s\in \mathcal{S}}\pi_z(s|x) v_D(y,x,a^*_{\pi_z,s})  
        = \allowbreak
        \sum_{y\in \mathcal{Y},x\in \mathcal{X}} b_{Y,X}(y,x)  \allowbreak [\rho_{D,U}^{sa} {v}_U(y,x,a_0) 
        +\rho_{D,U}^{tr}(y,x)] = \rho_{D,U}^{sa}\max_{a\in \mathcal{A}} \sum_{y\in \mathcal{Y},x\in \mathcal{X}} b_{Y,X}(y,x)  {v}_U(y,x,a)  \allowbreak
        + 
        \sum_{y\in \mathcal{Y}, x\in \mathcal{X}} b_{Y,X}(y,x) \rho_{D,U}^{tr}(y,x)
$. 
Based on the concavification technique in \cite{kamenica2011bayesian},  $\tilde{J}_D(\pi^*,b_{Y,X},{v}_D,{v}_U)$ is the concave closure of $\tilde{J}_D(\pi_z,b_{Y,X},{v}_D,{v}_U)$ over $b_{Y,X}\in \mathcal{B}_{Y,X}$. 
Since $\max_{a\in \mathcal{A}} \sum_{y\in \mathcal{Y}}\sum_{x\in \mathcal{X}} b_{Y,X}(y,x)  {v}_U(y,x,a) $ is PWL and convex concerning $b_{Y,X}(y,x), \forall y\in \mathcal{Y}, x\in \mathcal{X}$, we obtain that $\tilde{J}_D(\pi_z,b_{Y,X},{v}_D,{v}_U)$ is PWL and convex (resp. concave) in $b_{Y,X}\in \mathcal{B}_{Y,X}$ if and only if ${\rho}_{D,U}^{sa}>0$ (resp. ${\rho}_{D,U}^{sa}\leq 0$). 
Then, the convex hull of a concave function is itself, and the optimal ACEL equals $0$. 
Meanwhile, the convex hull of a convex function depends only on the vertices of the convex set $\mathcal{B}_{Y,X}$, i.e., ${v}_D(y,x,\tilde{a}^{max}(y,x)), \forall x\in \mathcal{X}, y\in \mathcal{Y}$. 
We arrive at the result: {\small $\tilde{J}_D(\pi^*, b_{Y,X},{v}_D,{v}_U)=\sum_{x\in \mathcal{X}}b_{X}(x) \bar{v}_D(x,a^{max}_U(x))
%=\sum_{y\in \mathcal{Y}}\sum_{x\in \mathcal{X}}b_{Y,X}(y,x) [{\rho}_{D,U}^{sa}{v}_U(y,x,a^{max}_U(x))+{\rho}_{D,U}^{tr}(y,x)]
= {\rho}_{D,U}^{sa}  \sum_{x\in \mathcal{X}}b_{X}(x) \max_{a\in\mathcal{A}} \bar{v}_U(x,a)+\sum_{y\in \mathcal{Y},x\in \mathcal{X}}b_{Y,X}(y,x) {\rho}_{D,U}^{tr}(y,x)$}, under the optimal recommendation policy $\pi_f\in \Pi_{ct}$. 
%Thus, we obtain the two statements for $\pi^*$ and $J_D^*$.  
\end{proof}

% \begin{remark}
% An insider is amenable if ${\rho}_{D,U}^{sa}>0$ and malicious if ${\rho}_{D,U}^{sa}\leq 0$. 
% Lemma \ref{lemma:linearutility0} provides an closed form solution of the optimal ACEL for malicious and amendable insiders. 
% \end{remark}

\begin{remark}
\label{remark:nonlinear}
Since a recommendation policy $\pi\in \Pi$ has impact on $b_{Y,X}^{\pi}$ as shown in Lemma \ref{remark:Conditional Belief Invariance}, the incentive $\bar{v}_D$ is not a constant as $b_Y$ changes. 
Thus, $\tilde{J}_D(\pi^*,b_{Y,X},{v}_D,{v}_U)$ is linear in $b_{Y,X}\in \mathcal{B}_{Y,X}$ but not $b_X\in \mathcal{B}_X$ (or $b_Y\in \Delta \mathcal{Y}$) in general. 
If mapping $\psi\in \Psi$ is non-stochastic, then ZETAR degenerates to the Bayesian persuasion model in \cite{kamenica2011bayesian}, and $J_D(\pi^*,b_X,\bar{v}_D,\bar{v}_U)$ is linear in $b_X\in \mathcal{B}_X$ (or $b_Y\in \Delta \mathcal{Y}$). 
%NOTE that \mathcal{B}_X is a subset of \Delta \mathcal{X}
%we define b_X\in \mathcal{B}_X as the probability of X is restricted by psi. 
%b_Y\in \Delta \mathcal{Y} as Y can change freely in [0,1]. 
%under some restricted psi, X may only be able to change in [0.3,0.8]. 
\end{remark}

According to Definition \ref{def:three category}, an insider is amenable if ${\rho}_{D,U}^{sa}>0$ and malicious if ${\rho}_{D,U}^{sa}\leq 0$. 
Therefore, Lemma \ref{lemma:linearutility0} provides a closed-form solution of the optimal ACEL for malicious and amendable insiders. 
We extend the discussion on the optimal ACEL concerning amenable and malicious insiders in  Proposition \ref{proposition:Maxaction0} and \ref{proposition:pertuationNochange}, respectively. For an amendable insider, Proposition \ref{proposition:Maxaction0} shows that within the entire action preference, the optimal action is the decisive factor. 

\begin{proposition} 
\label{proposition:Maxaction0}
%Consider $K\geq I$. %no need this
If the incentives of an insider and the defender share the same optimal action $\tilde{a}^{max}(y,x)\in \mathcal{A}$ for each AS $x\in\mathcal{X}$ and SP $y\in \mathcal{Y}$,  %i.e., $a^{max}(x)\in \text{arg}max_{a\in \mathcal{A}}\bar{v}_U(x,a)$ and $a^{max}(x)\in \text{arg}max_{a\in \mathcal{A}}\bar{v}_D(x,a)$, 
then for all $b_{Y,X}\in \mathcal{B}_{Y,X}$, full-information recommendation policy $\pi_f\in \Pi_{ct}$ achieves the optimal ACEL, and
\begin{equation}
\label{eq:maxTrans}
        J_D^*(b_{X},\bar{v}_D,\bar{v}_U)= J_D^*(b_{X},\bar{v}_U,\bar{v}_U)
    -\sum_{x\in \mathcal{X}} b_{X}(x) \delta(x), 
\end{equation} 
where $J_D^*(b_X,\bar{v}_U,\bar{v}_U)=\sum_{x\in \mathcal{X}}b_X(x) \max_{a\in \mathcal{A}} \bar{v}_U(x,a)$ and 
$\delta(x):=\bar{v}_U(x,a^{max}(x) ) -\bar{v}_D(x,a^{max}(x)), \forall x\in \mathcal{X}$. 
Moreover, $J_D^{acel}(\pi,b_{X},\bar{v}_D,\bar{v}_U)\geq 0, \forall \pi\in \Pi, b_X\in \mathcal{B}_X$. 
\end{proposition}

\begin{proof}
Based on Lemma \ref{lemma:linearutility0}, if we construct ${v}_D^0:={v}_U$, then $\tilde{J}_D(\pi_z,b_{Y,X},{v}_D^0,{v}_U)$ is PWL and convex in $b_{Y,X}\in \mathcal{B}_{Y,X}$, and $\tilde{J}_D(\pi^*,b_{Y,X},{v}_D^0,{v}_U)$ only depends on ${v}_D^0(y,x,\tilde{a}^{max}(y,x)), \forall x\in \mathcal{X},y\in \mathcal{Y}$, for all $b_{Y,X}\in \mathcal{B}_{Y,X}$.  
Thus, we can construct ${v}_D^1$ from ${v}_D^0$ to make  $\tilde{J}_D(\pi^*,b_{Y,X},{v}_D^{1},{v}_U)=\tilde{J}_D(\pi^*,b_{Y,X},{v}_D^{0},{v}_U)$ as long as  ${v}_D^{1}(y,x,\tilde{a}^{max}(y,x))={v}_D^{0}(y,x,\tilde{a}^{max}(y,x)), \forall y\in \mathcal{Y}, x\in \mathcal{X}$. 

If an insider with incentive ${v}_U$ and the defender with security objective ${v}_D$ prefer the same optimal action for each AS and SP, we can construct  $\bar{v}_D^{eq}(x,a):=\bar{v}_D(x,a)+\delta(x)$ such that $\bar{v}_D^{eq}(x,a^{max}(x))=\bar{v}_U(x,a^{max}(x)), \forall x\in \mathcal{X}$. 
Then, $J_D^*(b_{X},\bar{v}_D^{eq},\bar{v}_U)=J_D^*(b_{X},\bar{v}_U,\bar{v}_U)$. 
Based on Lemma \ref{lemma:LIdefender},  $J_D^*(b_{X},\bar{v}_D^{eq},\bar{v}_U)=J_D^*(b_{X},\bar{v}_D,\bar{v}_U)+\sum_{x\in \mathcal{X}} \delta(x)$, which leads to \eqref{eq:maxTrans}. 
Since $J_D(\pi, b_{X},\bar{v}_D^{eq},\bar{v}_U)\geq J_D(\pi_z,b_{X},\bar{v}_D^{eq},\bar{v}_U), \forall \pi\in \Pi, b_X\in \mathcal{B}_X$, based on Lemma \ref{lemma:linearutility0}, Lemma \ref{lemma:LIdefender} yields $J_D(\pi,b_{X},\bar{v}_D,\bar{v}_U)\geq J_D(\pi_z, b_{X},\bar{v}_D,\bar{v}_U)$. 
\end{proof}

\begin{remark}[\textbf{Sufficiency under Aligned Action Preference}]%Weak Preference Alignment
Proposition \ref{proposition:Maxaction0} shows that if an insider and the defender share the same optimal action $\tilde{a}^{max}(y,x)\in \mathcal{A}$ for each AS $x\in\mathcal{X}$ and SP $y\in \mathcal{Y}$, then ${v}_D(y,x,\tilde{a}^{max}(y,x)), \forall x\in \mathcal{X}, y\in \mathcal{Y}$, are the minimal sufficient component of the defender's security objective ${v}_D$ to determine the defender's optimal ASeL. % and the optimal ACEL. 
\end{remark}

\begin{remark}[\textbf{Simplified ZETAR Problem}]
\label{remark:Degeneration of the Principal-Agent Problem0}
When ${v}_D={v}_U$, the principal-agent problem $\tilde{J}_D(\pi^*,b_{Y,X},{v}_U,{v}_U)$ is equivalent to a single-agent decision problem that directly solves $\sum_{x\in \mathcal{X}}b_X(x) \max_{a\in \mathcal{A}} \bar{v}_U(x,a)$. Thus, Proposition \ref{proposition:Maxaction0} contributes to an efficient computation of the defender's optimal ASeL when she shares the same $\tilde{a}^{max}$ with an insider. 
\end{remark}

\begin{remark}[\textbf{Full Information Disclosure to Amendable Insiders}]% for Federated Defense
\label{remark:full info0}
%Full information disclosure means sending $\pi_f\in \Pi_{ct}$.  
The defender should share full information (i.e., adopt $\pi_f\in \Pi_{ct}$) with amendable insiders. 
By synchronizing information with compliant insiders, the defender encourages amendable insiders to contribute to corporate security.  
%closely collaborates with amendable insiders. 
%achieves federated defense with the compliant insiders. 
%Need to talk about Federated defense before. 
\end{remark}

For malicious insiders, we first introduce a class of invariant perturbations of the defender's security objective that achieve the same optimal ACEL of $J_D^{acel,*}(b_{X},\bar{v}_D,\bar{v}_U)=0$ in Proposition \ref{proposition:pertuationNochange}. 
Define shorthand notation $a^{min}(y,x)\in \text{arg}min_{a\in \mathcal{A}}{v}_D(y,x,a)$ for all $x\in \mathcal{X}, y\in \mathcal{Y}$. 

\begin{proposition}[\textbf{Compliance Equivalency under Security Objective Perturbations}]%equivalent class of malicious insiders
\label{proposition:pertuationNochange}
%If an insider and the defender are preference-divergent, then they are information-divergent. %this is not true
We construct ${v}_D^{ip}$ as a copy of $\bar{v}_D$ with the following revision: for each $x^i\in \mathcal{X}, i\in \mathcal{I}, y\in \mathcal{Y}$, if $a^{min}(y,x^j)\neq a^{min}(y,x^i)$, then ${v}_D^{ip}(y,x^j,a^{min}(y,x^i))\leq {v}_D(y,x^j,a^{min}(y,x^i)), \forall j\in \mathcal{I}\setminus \{i\}$. 
If there exist ${\rho}_{D,U}^{sa}<0$ and ${\rho}_{D,U}^{tr}(y,x)\in \mathbb{R}$ such that ${v}_D(y,x,a)={\rho}_{D,U}^{sa} {v}_U(y,x,a)+{\rho}_{D,U}^{tr}(y,x)$ for all $y\in \mathcal{Y}, x\in \mathcal{X},  a\in\mathcal{A}$, then $J_D^*(b_{X},\bar{v}_D^{ip},\bar{v}_U)=J_D^*(b_{X},\bar{v}_D,\bar{v}_U)$. 
\end{proposition}
\begin{proof}
Lemma \ref{lemma:linearutility0} shows that  $\tilde{J}_D(\pi^*,b_{Y,X},{v}_D,{v}_U)$ as the concave closure of $\tilde{J}_D(\pi_z,b_{Y,X},{v}_D,{v}_U)$ is PWL and concave in $b_{Y,X}\in \mathcal{B}_{Y,X}$ if ${\rho}_{D,U}^{sa}<0$. 
Based on the geometry, changing ${v}_D$ to ${v}_D^{ip}$ does not affect the concave closure (yet $\pi^*\in \Pi$ can change and does not contain zero information), i.e.,  $J_D^*(b_{X},\bar{v}_D^{ip},\bar{v}_U)=J_D^*(b_{X},\bar{v}_D,\bar{v}_U)$. 
%The condition $a^{min}(y,x^j)\neq a^{min}(y,x^i)$ guarantees that changing $v_D$ to any $v_D^{ip}$ does not affect the convex hull. 
\end{proof}

Remark \ref{remark:full info0} provides the defender with the guidance of full-information disclosure to amendable insiders. 
Based on Lemma \ref{lemma:linearutility0} and Proposition \ref{proposition:Maxaction0}, it is natural to conjecture that the defender should disclose zero information to malicious insiders. %tempting to conjecture
However, it does not hold, and we present a counterexample in Proposition \ref{proposition:pertuationNochange}; i.e., although an insider with incentive $v_U$ is malicious to both the defender with security objective ${v}_D$ and the one with ${v}_D^{ip}$, zero information recommendation policy is not the optimal policy for the defender with ${v}_D^{ip}$. 
Thus, Proposition \ref{proposition:pertuationNochange} leads to the strategic information disclosure guideline in Remark \ref{remark:strategic info0}. 
It further shows that ZETAR can improve an insider's compliance even if he is malicious based on the incentive categorization in Definition \ref{def:three category}.  

\begin{remark}[\textbf{Strategic Information Disclosure to Malicious Insiders}]
\label{remark:strategic info0}
%Revealing no information (i.e., adopting $\pi_z\in \Pi$) to a malicious insider may not always result in the best outcome.  
The defender should disclose information strategically rather than hide information (i.e., adopting $\pi_z\in \Pi$) even when the insider is malicious and tends to take an action that results in the least utility to the defender. 
\end{remark}

\section{Feedback Design for Unknown Incentives}%Model-Agnostic Algorithm to Design CT Recommendation Policies %insiders of 
\label{sec:Model-Agnostic}

When the defender knows an insider's incentive $v_U$, we can use primal and dual convex programs in Section \ref{sec:Computational Framework of ZETAR} to compute the optimal recommendation policy $\pi^*_{\eta}$ for a given LoRC $\eta \in \mathbb{R}^{0+}$. 
In practice, however, insiders' incentives usually remain unknown to the defender, and the defender may not be able to formulate constraint (c) in $P_\eta$ to determine the CT policy set. 
%determine the CT policy set by formulating the constraints in $P_\eta$.  
To this end, we provide a feedback design approach in Section \ref{sec:Model-Agnostic} for the defender to learn the optimal recommendation policy based on the insiders' responses to recommendations, as shown in Fig. \ref{fig:ScenarioDiag}. 
In the proposed learning algorithms, the defender needs no prior knowledge nor trust in an insider's incentives. 
The zero-trust audit of all insiders provides the defender with their behaviors to learn incentives. 
%The defender applies the zero-trust audit to all insiders to learn their behaviors and incentives. % from scratch based on these behaviors. 

A straightforward feedback learning paradigm optimizes the defender's ASeL $J_D(\pi,b_X,\bar{v}_D,\bar{v}_U)$ over all recommendation policies in set $\hat{\Pi}$ directly. 
For a new insider with an unknown incentive, the defender at stage $m\in \{1,2,\cdots\}$  recommends actions according to a recommendation policy $\hat{\pi}_m\in \hat{\Pi}$. 
Then, the defender observes the insider's responses to these recommendations and evaluates her ASeL.  
At stage $m+1$, the defender uses her ASeL evaluation to update the recommendation policy from $\hat{\pi}_m\in \hat{\Pi}$ to $\hat{\pi}_{m+1}\in \hat{\Pi}$. 
The update rule depends on bespoke optimization methods (e.g., simulated annealing, Bayesian optimization, and reinforcement learning). 
%(e.g., simulated annealing \cite{kirkpatrick1983optimization}, Bayesian optimization \cite{huang2021advert}, and reinforcement learning \cite{huang2021radams}). 
The above learning paradigm is universal yet inefficient and does not guarantee global optimality. 
In Algorithms \ref{algorithm1} and \ref{algorithm2}, we design efficient feedback learning algorithms by exploiting the ZETAR features characterized in Section \ref{sec:characterize two}. 
In particular, the defender only learns the CT policy set $\hat{\Pi}_{ct}$ and then uses the primal convex program $P_\eta$ in Section \ref{sec:Computational Framework of ZETAR} to compute the optimal recommendation policy $\hat{\pi}^*_\eta$ and the optimal ACEL $J_D^{acel,*}$. 
The defender can achieve it as she knows her security objective $v_D$ to compute the objective function in $P_\eta$. 

%In Algorithm \ref{algorithm2}, we adopt the binary search algorithm that is guaranteed to find the CT policy set within a finite number of $\log_2 (1/\epsilon)$ steps for an accuracy $\epsilon>0$. 
% Our algorithm contributes to data-driven Bayesian persuasion. (Our method is simple as we do not minimize the regret, i.e., the learning process can be regretful and costly. We only consider asymptotic behavior of convergent to the right solution.)
% \textbf{Moreover, our algorithm contributes to the learning of the constraint of linear program from samples, and achieve perfect solution given sufficiently large samples. }
% I only find two references. 
% 1. Learning Linear Programs from Optimal Decisions
% 2. Learning Linear Programs from Data
%Also relates to inverse linear programming (inverse optimization). where we adjust some predefined parameters.
Based on Definition \ref{def:trustworthyPi}, we only need to learn all the PT policy sets $\hat{\Pi}_{pt}^k, \forall k\in \mathcal{K}$, to determine the CT policy set. 

Following the matrix representation in Section \ref{sec:Geometric Characterization of CT Policy Sets}, the $k$-th row vector $\hat{\pi}^k\in \hat{\Pi}^k$ of a recommendation policy $\hat{\pi}\in \hat{\Pi}$ can be equivalently represented a point, denoted by $(p_k^1,\cdots,p_k^I)$, in the unit hypercube of dimension $I$, where the coordinate $p_k^i=\hat{\pi}^{k,i} \in [0,1], \forall k\in \mathcal{K},i\in \mathcal{I}$. 
We refer to a point in the $k$-th hypercube as a $k$-th PT point if it represents the $k$-th row vector $\hat{\pi}^k$ of a $k$-th PT recommendation policy $\hat{\pi}\in \hat{\Pi}_{pt}^k$.
Since the $k$-th row $\hat{\pi}^k$ of $\hat{\pi}$ is sufficient to determine whether  $\hat{\pi}$ is PT based on Proposition \ref{proposition:separability}, learning the $k$-th PT set $\hat{\Pi}_{pt}^k$ is equivalent to determining the region formulated by the $k$-th PT points. 
We refer to the region as the $k$-th PT region, which is a convex polytope in the $k$-th hypercube of dimension $I$ based on Lemma \ref{lemma:convex}. 
Since a convex polytope can be uniquely represented by its vertices, we develop the following two algorithms to obtain the vertex representation (V-representation) of these regions. 
Due to the policy separability principle in Remark \ref{remark:separable}, we can determine the V-representation of the $k$-th PT region independently for each $k\in \mathcal{K}$. 
For any point $(p_k^1,\cdots,p_k^I)$ of the $k$-th hypercube, we define $\Omega(p_k^1,\cdots,p_k^I)\subseteq \hat{\Pi}$ as the set of recommendation policies whose $k$-th row vectors satisfy $\hat{\pi}^k=[p_k^1,\cdots,p_k^I]$. 
In Algorithm \ref{algorithm1},  we determine the whether the $2^I$ vertices of the $k$-th hypercube are $k$-th PT points. 
Let $V^k:=\{ (p_k^1,\cdots,p_k^I) |p^{i}_k\in \{0,1\}, \forall i\in \mathcal{I}\}$ be the set of these $2^I$ vertices. 
Among these $2^I$ vertices, the $k$-th PT ones compose the $k$-th PT cube-vertex set denoted as 
$V^k_{pt}\subseteq V^k$. 

\setlength{\algomargin}{1.07em}
\begin{algorithm}[h]
\SetAlgoLined
%\small 
%\footnotesize    
 \SetKwFor{ForPar}{for}{do in parallel}{endfor}
\SetAlgoVlined %no end to make it more compact
\textbf{Initialize} the $k$-th PT cube-vertex set $V_{pt}^k \leftarrow \emptyset$\; 
    \ForEach{vertex $(p_k^1,\cdots,p_k^I)\in V^k$}{ 
    \While{$s^k$ has not been recommended, i.e., $k'\neq k$}{
    Recommend $s^{k'}\in \mathcal{S}$ randomly based on a recommendation policy $\hat{\pi}\in \Omega(p_k^1,\cdots,p_k^I)$\;
    }
    \lIf{recommendation $s^k$ induces $a^k\in \mathcal{A}$}
    {$V_{pt}^k \leftarrow V_{pt}^k\cup \{(p_k^1,\cdots,p_k^I)\}$} 
    }
    \textbf{Return} the $k$-th PT cube-vertex set $V_{pt}^k$\;
 \caption{Algorithm to learn the $k$-th PT cube-vertex set $V_{pt}^k$ for a given insider. 
 \label{algorithm1}}
\end{algorithm}

In Algorithm \ref{algorithm2}, we determine the vertex coordinates of the $k$-th PT region. 
As a convex polytope, the region contains a finite set of polytope-vertices defined as $\bar{V}^k_{pt}$. 
Since the $k$-th PT region is determined by a hyperplane in the $k$-th hypercube, its vertices are on the edges, and it contains all the elements in the $k$-th PT cube-vertex set ${V}^k_{pt}$ as shown in the initialization step in line $7$. 
Each cube-vertex $(p_k^1,\cdots,p_k^I)\in V_{pt}^k$ has $I$ neighboring cube-vertices, and the coordinate of its $i$-th neighboring cube-vertex is $(p_k^1,\cdots,p_k^{i-1},1-p_k^{i},p_k^{i+1},\cdots, p_k^I)$. 
After we select a $k$-th PT cube-vertex in line $8$, we search over its $I$ neighboring cube-vertices in line $9$. 
If the neighboring vertex is also $k$-th PT, then the points on the edge of these two cube-vertices are both $k$-th PT. 
If the neighboring vertex is not $k$-th PT as shown in line $10$, then there is an additional polytope-vertex on the edge of these two cube-vertices. 
As shown in lines $11-17$, we use the binary search to learn the coordinate of this additional polytope-vertex and add it to the $k$-th polytope-vertex set $\bar{V}_{pt}^k$ in line $18$. 
In particular, the binary search adopts an accuracy $\epsilon>0$ used in the stopping criteria shown in line $12$. 
For the worst case where a polytope-vertex is close to a cube-vertex, we need $N\in \mathbb{Z}^+$ iterations to reach the stop criteria, i.e., $(1/2)^N\leq \epsilon$, which leads to $N\geq \log_2 (1/\epsilon)$. 
Since an $I$-dimensional hypercube has $2^{n-1}n$ edges, Algorithm \ref{algorithm2} is guaranteed to stop within $2^{n-1}n \log_2 (1/\epsilon)$ steps.

%Since the vertices of the polytope result from a hyperplane cut of a hypercube of dimension $I$, we can further reduce the number of polytope-vertices from $2^{n-1}n$ to $(n-[n/2]) {n \choose [n/2]}$ \cite{o1971hyperplane}. 

\setlength{\algomargin}{1.07em}
\begin{algorithm}[h]
\SetAlgoLined
%\small 
%\footnotesize    
\SetAlgoVlined %no end to make it more compact
\textbf{Initialize} $\bar{V}_{pt}^k \leftarrow V_{pt}^k$, and accuracy $\epsilon>0$\;
%Step 2: Start from the trustworthy vertex, search for the turning point 
    \ForEach{$k$-th PT vertex $(p_k^1,\cdots,p_k^I)\in V_{pt}^k$} %start from each trustworthy vertex
    {
    \For{$i \leftarrow 1$ \KwTo $I$} %search the neighboring nodes 
        {
            \If{$(p_k^1,\cdots,p_k^{i-1},1-p_k^{i},p_k^{i+1},\cdots, p_k^I) \notin V_{pt}^k$} %if =co, then Add constraint $\hat{\pi}^{k,i}\in [0,1]$, which is not useful. 
            {
            $lb \leftarrow 0$ and  $ub \leftarrow 1$\;
                    \While{$ub-lb>\epsilon$}
                    {
                    Recommend $s\in \mathcal{S}$ randomly based on  
                    {\small
                    $\hat{\pi}\in \Omega(p_k^1,\cdots,p_k^{i-1},\frac{lb+ub}{2},p_k^{i+1},\cdots, p_k^I)$}\;
                    \eIf{$s=s^k$ and $p_k^{i}=0$} %then $\bar{\pi}^{k,i}=1$, from 0 to 1
                    {
                    \leIf{insider takes action $a^k\in \mathcal{A}$ }{$lb=\frac{lb+ub}{2}$}{$ub=\frac{lb+ub}{2}$}
                    }
                    ($s=s^k$ and $p_k^{i}=1$) %then $\bar{\pi}^{k,i}=0$, from 1 to 0
                    {
                    {\leIf{insider takes action $a^k\in \mathcal{A}$ }{$ub=\frac{lb+ub}{2}$}{$lb=\frac{lb+ub}{2}$}
                    }
                    }
                    }
                    %Add constraint $\pi^{k,i}\in [lb,1]$\;
                    {\small
                    $\bar{V}_{pt}^k \leftarrow \bar{V}_{pt}^k\cup \{(p_k^1,\cdots,p_k^{i-1},\frac{lb+ub}{2},p_k^{i+1},\cdots, p_k^I)\}$}\;
            } 
        }
    }
    \textbf{Return} the $k$-th PT polytope-vertex set $\bar{V}_{pt}^k$\;
 \caption{Algorithm to learn the polytope-vertex set $\bar{V}_{pt}^k$ for a given insider. 
 \label{algorithm2}}
\end{algorithm}

After we obtain the V-representation of the $k$-th PT policy set, i.e., $\bar{V}_{pt}^k$, for each $k\in \mathcal{K}$,  we can use facet enumeration methods (e.g., \cite{avis1997good}) to obtain the half-space representation (H-representation) that can be directly used to construct the constraints in the primal convex program $P_\eta, \forall \eta\in \mathbb{R}^+$, in $\hat{\pi}\in \hat{\Pi}$. 
We provide a graphical illustration in Section \ref{sec:Graphical Illustration of Learning Algorithms} when $I=J=K=2$.

\section{Case Study}
\label{sec:case study}
In this section, we illustrate the design of ZETAR in Fig. \ref{fig:ScenarioDiag} under fully customized recommendation policies (i.e., $\eta=\infty$) to improve compliance for insiders with different incentives. %and consequently enhance corporate security. 
%We consider fully customized recommendation policies, i.e., $\eta=\infty$. 

\subsection{Model Description}
\label{sec:Model Description}
Following Fig. \ref{fig:overviewDiag}, we consider the binary SP, i.e., $\mathcal{Y}=\{y^{hr},y^{lr}\}$, where $y^{hr}$ and $y^{lr}$ represent the high-risk SP and the low-risk SP, respectively. 
%The state can also represent security/risk score/risk posture. See \url{https://reciprocity.com/risk/}. 
For illustration purposes, we consider binary audit schemes, i.e., $\mathcal{X}=\{x^{sa},x^{ta}\}$, where $x^{sa}$ and $x^{ta}$ represent stringent audit and tolerant audit, respectively. 
%Since frequent inspection can heighten tension and stress in the organization and cost additional budget burden, the defender chooses to inspect the insider's fulfilment of security rules only under $x^H$. 
The insider's behaviors are categorized into binary actions, i.e., $\mathcal{A}=\{a^{ic},a^{co}\}$, where $a^{ic}$ and $a^{co}$ represent non-compliant and compliant behaviors, respectively. 
\textcolor{black}{Since insiders have different risk attitudes toward gains and losses, we introduce a risk perception function $\kappa^{\gamma}$ with parameter $\gamma:=[\gamma_d,\gamma_s]$, where $\kappa^{\gamma}(v):=v^{\gamma_d},  v\geq 0$, and  $\kappa^{\gamma}(v):=-\gamma_s (-v)^{\gamma_d},v< 0$, based on Cumulative Prospect Theory (CPT) \cite{CPT}.}  
%We set $\gamma_d =1$ and $\gamma_s=1$ for risk-neutral insiders. 
%where $\gamma_s =2.25$ and $\gamma_d=0.88$. 
%Consider the asymmetric case with the same $\gamma$ and let $\bar{\gamma}>1$ to represent the loss aversion. 

\subsubsection{Insider's Intrinsic and Extrinsic Incentives}
\label{sec:casestudy_vU}
As shown in Table \ref{tab:utility}, we separate an insider's incentive $v_U^{\gamma}$ under $\kappa^{\gamma}$ into intrinsic incentive $v_{U,I}$ and extrinsic incentive $v^{\gamma}_{U,E}$. %, denoted by $v_{U,I}$ and $v^{\gamma}_{U,E}$, respectively. 

\begin{table}[h]
    \begin{subtable}[h]{0.24\textwidth}
        \centering
        \begin{tabular}{|c|c|c|}
\hline
$v^{\gamma}_{U,E}$ & $x^{sa}$ & $x^{ta}$        \\ \hline
$a^{ic}$     & $\kappa^{\gamma}(-c_D^{ic})$ & $0$          \\ \hline
$a^{co}$     & $\kappa^{\gamma}(r_D^{co})-c_U^{co}$ & $-c_U^{co}$ \\ \hline
\end{tabular}
        \subcaption{Extrinsic incentive. }
        \label{tab:Extrinsic insider}
     \end{subtable}
     \begin{subtable}[h]{0.24\textwidth}
        \centering
        \begin{tabular}{|c|c|c|}
\hline 
$v_{U,I}$ & $y^{hr}$        & $y^{lr}$            \\ \hline 
$a^{ic}$     & $c_U^{hr}$      & $c_U^{lr}$          \\ \hline
$a^{co}$     & $r_U^{hr}$   & $r_U^{lr}$ \\ \hline
\end{tabular}
       \subcaption{Intrinsic incentive. }
       %This utility matrix means that insider tends to be good when he understand the risk. %He tends to trade security for efficiency when he underestimate the risk 
       \label{tab:Intrinsic insider}
    \end{subtable}
     \caption{Insider's utility $v_U^{\gamma}=v_{U,I}+v^{\gamma}_{U,E}$. }
     \label{tab:utility}
\end{table}

The extrinsic incentive in  Table I(a) is independent of SP and captures the impact of AS on compliance. %Table \ref{tab:Extrinsic insider}
Compliant action $a^{co}\in \mathcal{A}$ introduces a compliance cost $c_U^{co}\in \mathbb{R}^+$ to an insider regardless of the AS. 
For example, an insider compliant with the air-gap rule has to spend additional time and effort to transfer data using a CD rather than a USB. 
%it is more convenient and time-saving for insiders to violate the air-gap rules and transfer data using USBs rather than CDs. 
Under AS $x^{sa}\in \mathcal{X}$, the defender introduces a reward $r_D^{co}\in \mathbb{R}^+$ and a penalty $c_D^{ic}\in \mathbb{R}^+$ to compliant action $a^{co}$ and non-compliant action $a^{ic}$, respectively. 
We assume that the tolerant audit scheme $x^{ta}$ introduces a reward and penalty of $0$ to $a^{co}$ and $a^{ic}$, respectively. 
% \begin{figure}[htb]
%     \centering % <--a
% \begin{subfigure}{0.45\textwidth}
%   \includegraphics[width=\linewidth]{fig/utilitydist}
%   \caption{\label{fig:CPT_u} 
% Choose $\kappa^{\gamma_d}(v)=v^{\gamma_d},  v\geq 0$; $\kappa^{\gamma_d}(v)=-\bar{\gamma_d} (-v)^{\gamma_d},v< 0$. where $\gamma_s =2.25$ and $\gamma_d=0.88$. 
%  }
% \end{subfigure}\hfil % <-- c
% \begin{subfigure}{0.45\textwidth}
%   \includegraphics[width=\linewidth]{fig/CPTw.eps} %ProbabilityDistort
%   \caption{\label{fig:CPT_P} 
% We distort the CDF of $P$ as 
% $w^{\chi}(P)=\exp(-(-\log(P))^{\chi})$. 
% The original form is $w^{+}(p)=\frac{p^{\bar{\delta}}}{\left(p^{\bar{\delta}}+(1-p)^{\bar{\delta}}\right)^{1 / \bar{\delta}}}$ and $w^{-}(p)=\frac{p^{\delta}}{\left(p^{\delta}+(1-p)^{\delta}\right)^{1 / \delta}}$, where $\delta=0.69, \bar{\delta}=0.61 \in [0,1]$.  
%   }
% \end{subfigure}\hfil % <--b
% \caption{Cumulative Prospect Theory}
% \label{fig:CPT}
% \end{figure}
The intrinsic incentive in Table I(b) is independent of AS and captures an insider's internal inclination to comply under different SP realizations. %Table \ref{tab:Intrinsic insider}
Under high-risk SP $y^{hr}$ (resp. low-risk SP $y^{lr}$), an insider receives an intrinsic penalty (e.g., the guilty of misconduct) denoted by $c_U^{hr}$ (resp. $c_U^{lr}$) to take non-compliant action $a^{ic}$ and an intrinsic reward (e.g., the gratification of being compliance-seeking) denoted by $r_U^{hr}$ (resp. $r_U^{lr}$) to take compliant action $a^{co}$. 
Based on Lemma \ref{lemma:LIinsider}, we can choose $\rho_U^{tr}(y^{hr},x)=-r_U^{hr}$ and $\rho_U^{tr}(y^{lr},x)=-r_U^{lr}$ for all $x\in \mathcal{X}$ without affecting the insider's compliance and the optimal recommendation policy $\pi^*\in \Pi$. 
Thus, without loss of generality, we calibrate $r_U^{hr}=r_U^{lr}=0$ and $c_U^{hr},c_U^{lr}\in \mathbb{R}$ and characterize the following three compliance attitudes of insiders in Definition \ref{def:Compliance Attitudes}. 

\begin{definition}[\textbf{Compliance Attitudes}]
\label{def:Compliance Attitudes}
An insider is said to be compliance-seeking, compliance-averse, and compliance-neutral if both $c_U^{hr}$ and $c_U^{lr}$ are positive, negative, and zero, respectively. 
\end{definition}
%Then, if $c_U^{hr}>0$ and $c_U^{lr}>0$, then an insider can be treated as compliance-seeking. 
%If $c_U^{hr}<0$ and $c_U^{lr}<0$, then an insider can be treated as compliance-averse. 

\subsubsection{Defender's Security Objective}
\label{sec:casestudy_vD}
Table \ref{tab:defender's utility} illustrates the defender's security objective $v_D$. 
Following Section \ref{sec:audit and action}, stringent audit increases insiders' pressures and reduces their working efficiency. We capture the efficiency reduction with a cost $c_D^{ca}\in \mathbb{R}^+$. 
%{\bf NOT CLEAR TO ME WHAT IT REFERS TO discord??} between insiders and the defender, which is captured by a confrontation cost $c_D^{ca}\in \mathbb{R}^+$. 
%When the insider takes the non-compliant action $a^{ic}$, then the manual audit $x^{sa}$ can identify it and patch the induced vulnerability accordingly, which yields a reward of $r_D^{hr}\in \mathbb{R}^+$ in Table \ref{tab:hr defender} and $r_D^{lr}\in \mathbb{R}^+$ in Table \ref{tab:lr defender} under the high-risk SP $y^{hr}$ and low-risk SP $y^{lr}$, respectively. 
When an insider takes a non-compliant action $a^{ic}$, the stringent audit $x^{sa}$ requires an immediate correction from the insider to patch the induced vulnerability, which yields a reward of $r_D^{ca}\in \mathbb{R}^+$ in Table \ref{tab:defender's utility} regardless of the SP realizations. 
Meanwhile, the tolerant audit $x^{ta}$ introduces no cost of efficiency reduction %{\bf what is this confrontation cost?} 
but additional risks of insider threats captured by the cost $c_D^{hr}\in \mathbb{R}^+$ in Table II(a) and $c_D^{lr}\in \mathbb{R}^+$ in Table II(b) under high-risk SP $y^{hr}$ and low-risk SP $y^{lr}$, respectively. % Table \ref{tab:hr defender} %Table \ref{tab:lr defender} 
When an insider takes a compliant action $a^{co}$, the risk of insider threats is reduced to a minimum and is represented as the defender's payoff $r_D^{sa}\in \mathbb{R}$.  
%Then, $v_D$ under two audit scheme share the same value and is calibrated to $0$. 
\begin{table}[h]
    \begin{subtable}[h]{0.24\textwidth}
        \centering
        \begin{tabular}{|c|c|c|}
\hline
$v_D$ & $x^{sa}$        & $x^{ta}$          \\ \hline
$a^{ic}$ & $r_D^{ca}-c_D^{ca}$        & $-c_D^{hr}$          \\ \hline
$a^{co}$ & $-c_D^{ca}$ &        $r_D^{sa}$ \\ \hline
\end{tabular}
        \caption{$v_D$ at high-risk SP.}
        \label{tab:hr defender}
     \end{subtable}
     \begin{subtable}[h]{0.24\textwidth}
        \centering
        \begin{tabular}{|c|c|c|}
\hline
$v_D$ & $x^{sa}$        & $x^{ta}$          \\ \hline
$a^{ic}$ & $r_D^{ca}-c_D^{ca}$        & $-c_D^{lr}$          \\ \hline
$a^{co}$ & $-c_D^{ca}$ &        $r_D^{sa}$ \\ \hline
\end{tabular}
       \caption{$v_D$ at low-risk SP.}
       %This utility matrix means that insider tends to be good when he understand the risk. %He tends to trade security for efficiency when he underestimate the risk 
       \label{tab:lr defender}
    \end{subtable}
     \caption{Defender's utility $v_D$ at two SP states. }
     \label{tab:defender's utility}
\end{table}

% \begin{table}[h]
% \centering
%         \begin{tabular}{|c|c|c|}
% \hline
% $v_D$ & $x^{sa}$        & $x^{ta}$          \\ \hline
% $a^{ic}$ & $r_D^{ic}(y)-c_D^{ca}$        & $-c_D^{sa}(y)$          \\ \hline
% $a^{co}$ & $-c_D^{ca}$ &        $0$ \\ \hline
% \end{tabular}
% \caption{Defender's Utilities $v_D$. }
% \label{tab:defender's utility}
% \end{table}

\subsection{Graphical Illustration of Learning Algorithms}
\label{sec:Graphical Illustration of Learning Algorithms}

In this case study, the defender has no prior knowledge of an insider's risk and compliance attitudes. 
Moreover, the defender assigns no prior trust to the insiders and applies the algorithms in Section \ref{sec:Model-Agnostic} to learn their incentives. 
Fig. \ref{fig:compliance-seekingBR2} and Fig. \ref{fig:compliance-averseBR2} illustrate the Algorithms under compliance-seeking and compliance-averse insiders, respectively. 
%the CT policy set characterization in Section \ref{sec:Geometric Characterization of CT Policy Sets} and
Since $K=2$, the recommendation policy $\hat{\pi}\in \hat{\Pi}$ can be equivalently represented as a point ($p^1,p^2$) in the unit hypercube of dimension $I=2$ (i.e., a unit square), where $p^1=\pi(s^{ic}|x^{sa}), p^2=\pi(s^{ic}|x^{ta})$ as shown in Section \ref{sec:Model-Agnostic}. 
In the hypercube space illustrated by Fig. \ref{fig:BRset}, the green and blue regions represent the CT and CU policy sets, respectively. 
%, which are shown to be convex polytopes. 
%to be convex (in Lemma \ref{lemma:convex}), centrosymmetric (in Corollary \ref{corollary:Centrosymmetry}), and with volume ranging from $0$ to $1/2$ (in Corollary \ref{corollary:Volume Bounds}). 
%\vspace{-4mm}
\begin{figure}[h]
    \centering % <--a
    \begin{subfigure}{0.22\textwidth}
\includegraphics[width=1\columnwidth,height=.8\columnwidth]{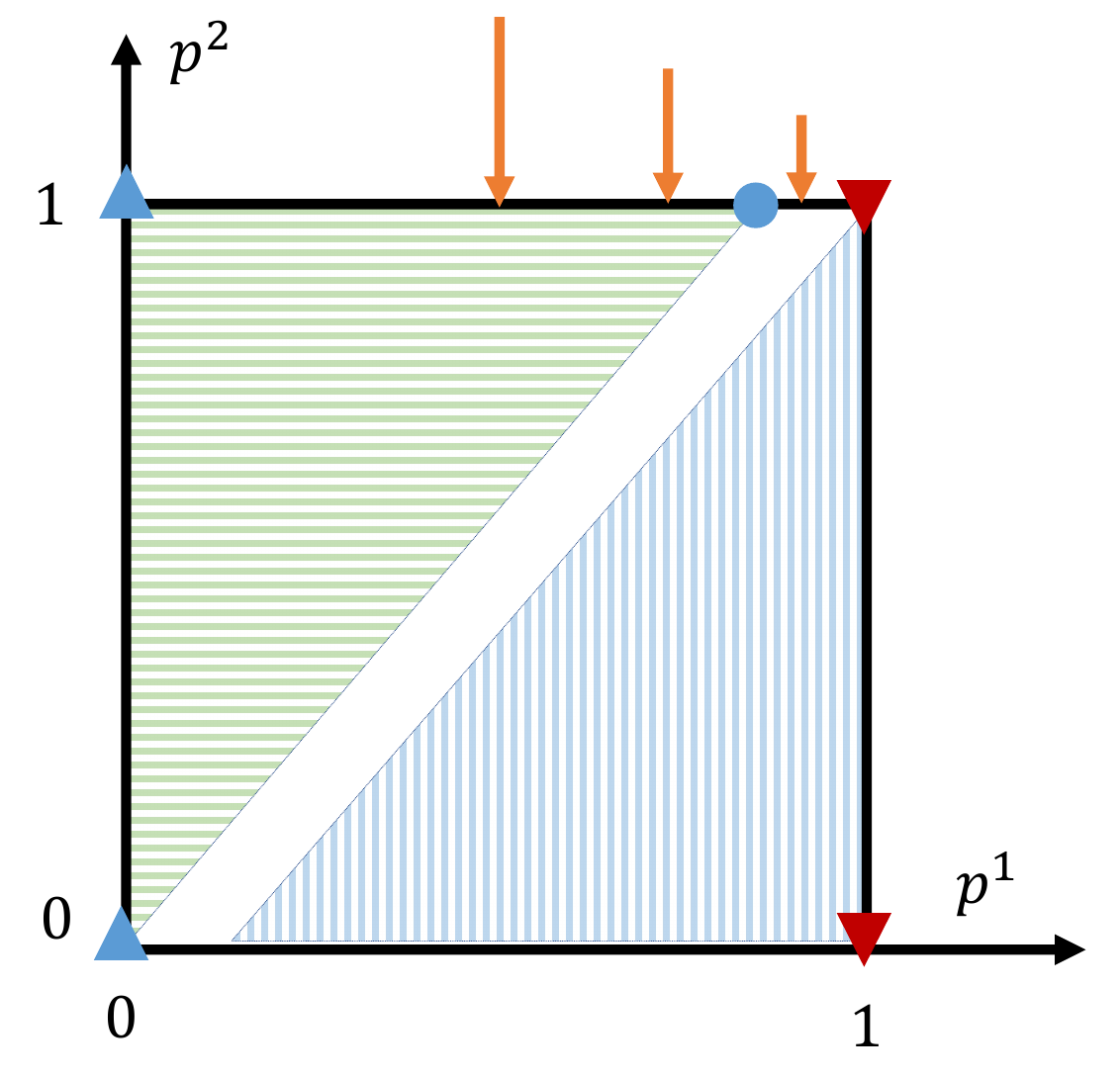}
\caption{ 
Compliance-seeking insiders. 
}
\label{fig:compliance-seekingBR2}
\end{subfigure}\hfil % <--b
    \begin{subfigure}{0.22\textwidth}
\includegraphics[width=1\columnwidth,height=.8\columnwidth]{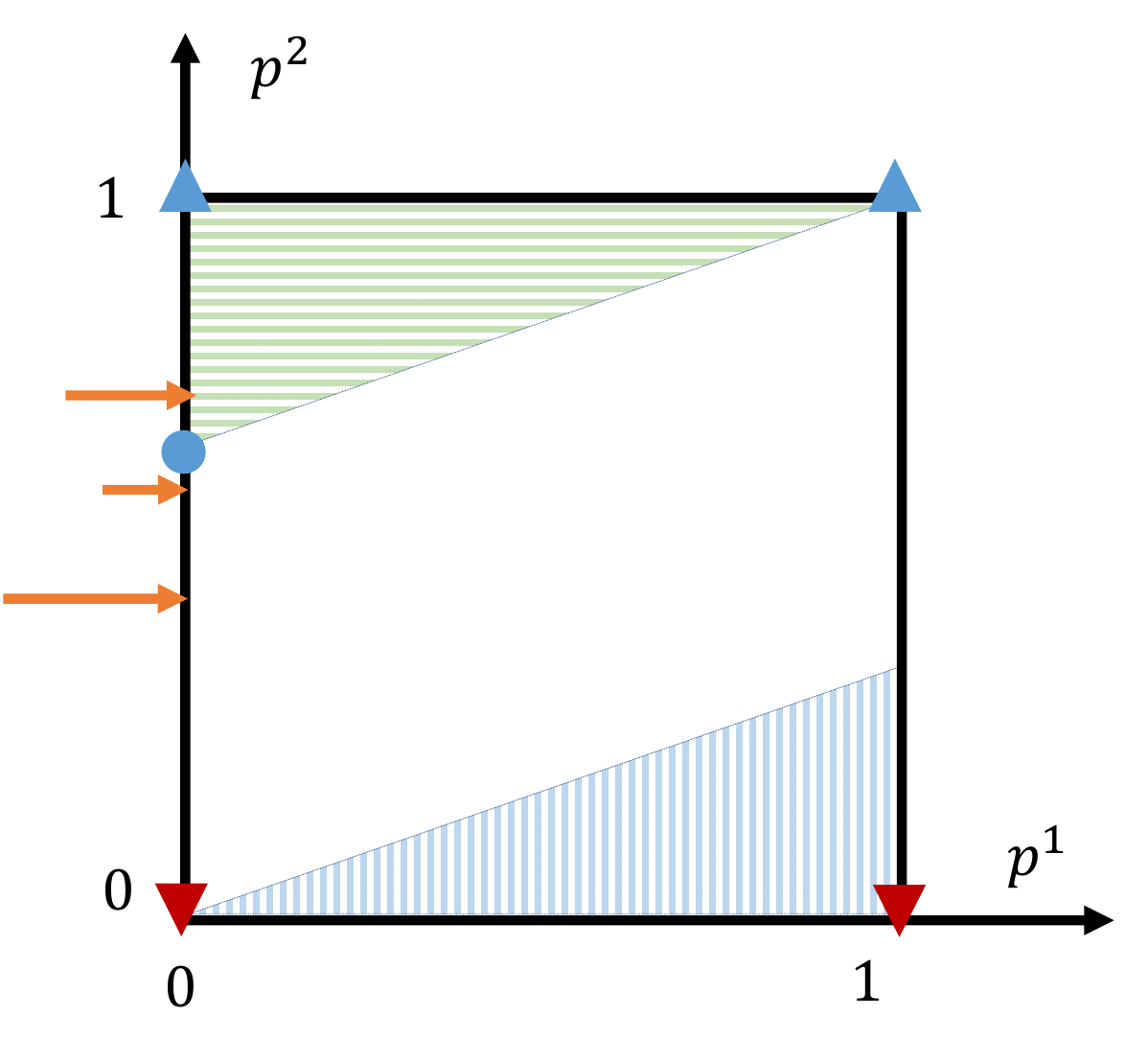} %height=.5\columnwidth
\caption{ 
Compliance-averse insiders. 
}
\label{fig:compliance-averseBR2}
\end{subfigure}\hfil % <--b
\caption{
The blue upward and the red downward triangles represent the CT and CU recommendations policies, respectively. 
The green (resp. blue) region with horizontal (resp. vertical) lines represents the CT (resp. CU) policy sets, respectively. 
The orange lines show the steps of the binary search in Algorithm \ref{algorithm2}.  
\label{fig:BRset}
}
\end{figure}

Algorithm \ref{algorithm1} determines whether the four vertices of the hypercube are the $1$-st PT or not, which are illustrated by the blue upward and red downward triangles, respectively, in Fig. \ref{fig:BRset}. 
Algorithm  \ref{algorithm2} further determines the additional polytope-vertex (represented by the blue circle in Fig. \ref{fig:BRset}) of the $1$-st PT polytope in green. 
From each blue triangle (line $8$), if a neighboring cube-vertex (line $9$) is also a blue triangle, then no additional polytope-vertices are needed to determine the green region (line $10$). 
If the neighboring cube-vertex is red, then binary search is applied to determine the additional polytope-vertex. 
In Fig. \ref{fig:compliance-seekingBR2} (resp. Fig. \ref{fig:compliance-averseBR2}), from the blue cube-vertex $(0,1)$, the red neighboring cube-vertex is $(1,0)$ (resp. $(0,0)$), and the binary search adopts line $15$ (resp. line $17$). 
We use the orange lines in Fig. \ref{fig:compliance-averseBR2} to illustrate the binary search process, i.e., line $11$ to $17$ in Algorithm \ref{algorithm2}. 
The first step of the binary search (represented by the longest orange line) evaluates the recommendation policy represented by the point $(0,1/2)$, and the policy is not the $1$-st PT. Thus, we update the lower bound $lb$ based on the \textit{else} condition in line $16$ of Algorithm \ref{algorithm2}. 
The second step (represented by the second-longest orange line) evaluates the recommendation policy represented by the point $(0,3/4)$, and the policy is the $1$-st PT. 
Thus, we update the upper bound $ub$ based on the \textit{then} condition in line $14$. 
The third step (represented by the third-longest orange line) evaluates the recommendation policy represented by the point $(0,5/8)$, and the policy is not the $1$-st PT. 
Thus, we update the lower bound again. 
We repeat the above process of binary search until $ub-lb\leq \epsilon$ as shown in line $12$, and we find the additional polytope vertex represented by the blue circle in Fig. \ref{fig:compliance-averseBR2}. 
After we obtain all the vertices of the polytope that represent the CT policy set, we can use facet enumeration methods to obtain the $H$-representation and construct the constraints of $P_\eta$ concerning $p^1,p^2\in [0,1]$. 
For example, if the coordinate of the blue circle in Fig.  \ref{fig:compliance-averseBR2} is $(0,w), w\in [0,1]$, then the constraint is $p^2\geq (1-w)p^1+w$. 

\subsection{Numerical Results}
%Following Section \ref{sec:audit and action}, 
We choose $\psi(x^{sa}|y^{hr})=0.8$ and $\psi(x^{sa}|y^{lr})=0.3$; i.e., the audit firm chooses a stringent audit with probability $0.8$ and $0.3$ under high-risk SP $y^{hr}$ and low-risk SP $y^{lr}$, respectively. 
%Under the above psi, as bY change from 0 to 1, bX(x^{sa}) changes from 0.3 to 0.8 linearly. 

\subsubsection{Compliance Threshold}
\label{sec:casestudy_initial compliance}
%Select $cUhr=cUlr=-1$ and $cUhr=cUlr=1$ and $cUhr=cUlr=0$ for three types of insiders.
%Select rDco=3; %reward for compliance

Following Section \ref{sec:insider's Initial Compliance}, we investigate the initial compliance of an insider with three compliance attitudes in Definition \ref{def:Compliance Attitudes} and different risk perception parameters $\gamma$. 
Define $t_{ze}\in \mathbb{R}$ as the zero of the function $f(b_Y(y^{hr})):=\sum_{y\in \mathcal{Y}} b_Y(y)\sum_{x\in \mathcal{Y}}\psi(x|y)[v_U(y,x,a_1)-v_U(y,x,a_2)]$. 
Let $t_{bt}:=\max{\{\min{\{t_{ze},1\}},0\}}$ be the belief threshold of an insider. 
For binary actions, an insider adopts a \textit{threshold policy} where $a_0=a^{co}$ if $b_Y(y^{hr})\geq t_{bt}$ and $a_0=a^{ic}$ if $b_Y(y^{hr})<t_{bt}$.  
Fig. \ref{fig:initial compliance} illustrates the belief threshold versus the non-compliance penalty $c_D^{ic}\in \mathbb{R}^+$. 
\begin{figure}[h]
    \centering % <--a
    \begin{subfigure}{0.22\textwidth}
\includegraphics[width=1\columnwidth]{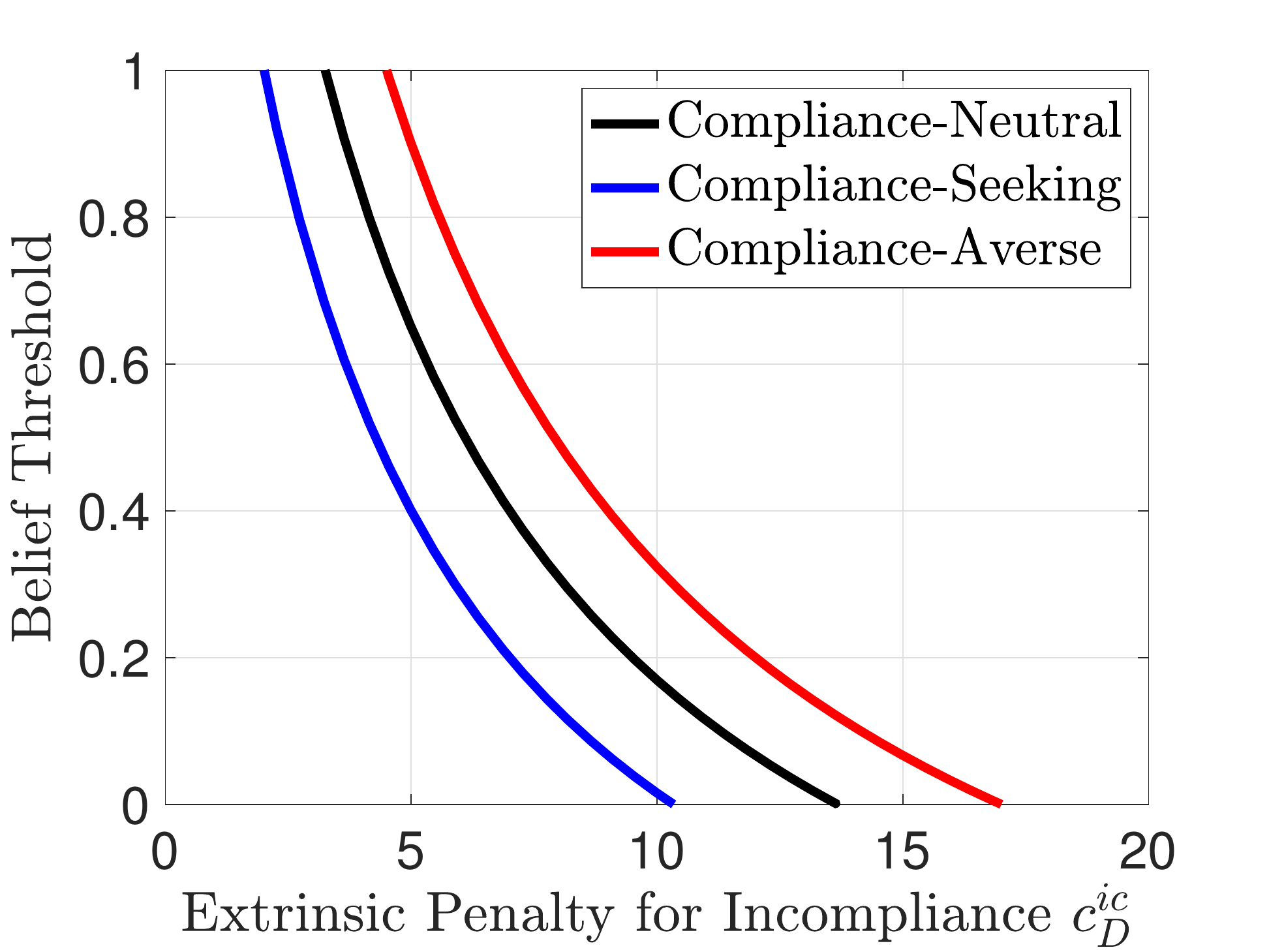}
\caption{ 
Insiders with three compliance attitudes under $\gamma_d=\gamma_s=1$.  
\label{fig:threshold_noCPT}
}
\end{subfigure}\hfil % <--b
    \begin{subfigure}{0.22\textwidth}
\includegraphics[width=1\columnwidth]{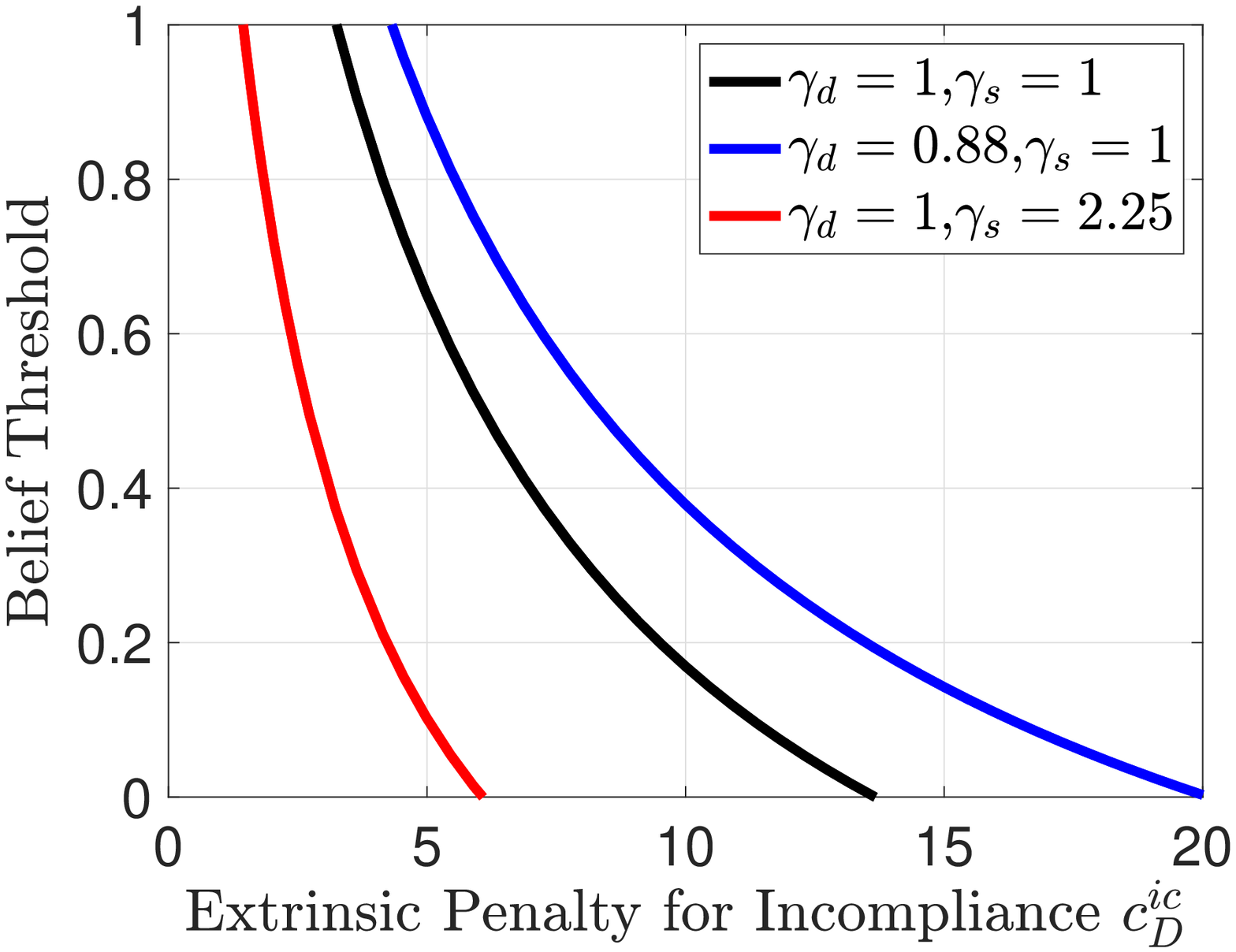} %height=.5\columnwidth
\caption{ 
Compliance-neutral insiders under distorted risk perceptions. 
\label{fig:threshold_withCPT}
}
\end{subfigure}\hfil % <--b
\caption{
Insiders' belief thresholds $t_{bt}\in [0,1]$ to compliant actions versus the value of non-compliance penalty $c_D^{ic}\in \mathbb{R}^+$. 
\label{fig:initial compliance}
}
\end{figure}
The plots show that increasing penalty $c_D^{ic}$ can make insiders more likely to take compliant action $a^{co}$ (i.e., a smaller belief threshold). 
Fixing the penalty value, compliance-averse (resp. compliance-seeking) insiders are the least (resp. most) likely to comply, i.e., the largest (resp. smallest) belief thresholds, among insiders with three compliance attitudes, as shown in Fig. \ref{fig:threshold_noCPT}.   
%The value of the zero point is the penalty required to make insider's full compliant.
In Fig. \ref{fig:threshold_withCPT}, a larger $\gamma_s$ in red represents a higher degree of loss aversion, which makes an insider more likely to comply. 
A small $\gamma_d$ in blue enhances the effect of diminishing sensitivity, which makes a large penalty less effective to induce compliant behaviors.

\subsubsection{Impacts of Recommendation Policies}
\label{sec:impacts of recommendation policies}
%Enhancement
Here, we specify $b(y^{hr})=0.2$ and $c_D^{ic}=10$ to inspect the impact of recommendation policies on an insider's behaviors. %decisions and actions. %perception and 
%add back
\textcolor{black}{
Fig. \ref{fig:posteriorbelief} illustrates the impact of different recommendation policies $\hat{\pi}\in \hat{\Pi}$ on an insider's posterior belief $b_X^{\pi}$ under Bayesian belief update in Section \ref{sec:insider's Belief Update and Best-Response Action}. 
The posterior belief is independent of $v_D, v_U,\eta$ and $\gamma$. %CPT parameters  
The plot illustrates the Bayesian plausibility \cite{kamenica2011bayesian} where $\sum_{s\in \mathcal{S}} b_S^{\pi}(s)b_X^{\pi}(x|s)=b_X(x), \forall x\in \mathcal{X}, \pi\in \Pi$.}  
\begin{figure}[h]
\centering
\includegraphics[width=.45 \textwidth,height=.6\columnwidth]{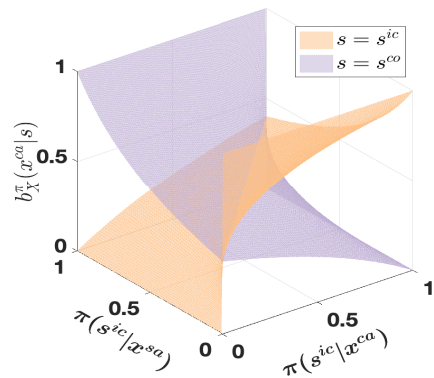}
\caption{ 
\textcolor{black}{An insider's posterior belief $b_X^{\pi}(x^{sa}|s)$ under recommendation $s=s^{ic}$ and $s=s^{co}$ in brown and pink, respectively, vs. $\pi(s^{ic}|x^{sa})\in [0,1]$ in $x$-axis and $\pi(s^{ic}|x^{ta})\in [0,1]$ in $y$-axis.} 
\label{fig:posteriorbelief}
}
\end{figure} 
Fig. \ref{fig:ACEL_nocpt} illustrates the impact of different recommendation policies $\hat{\pi}\in\hat{\Pi}$ on the ACEL under insiders with two compliance attitudes, which corroborate the PWL property in Lemma \ref{lemma:PWLC of the Defender's EPU}. 
%\textcolor{black}{and the centrosymmetry of the defender's ASeL in Lemma \ref{lemma:PWLC of the Defender's EPU} and Corollary \ref{corollary:Centrosymmetry}, respectively.}   
Different compliance attitudes only affect the policy set partition denoted by $\mathcal{C}_l^k, \forall l,k\in \{ic,co\}$, following Section \ref{sec:Geometric Characterization of CT Policy Sets}. 
In Fig. \ref{fig:ACEL_nocpt_good}, the policy sets (also illustrated in Fig. \ref{fig:compliance-seekingBR2}) illustrated by the contour plots on the $xy$-plane are sets $\mathcal{C}_{ic,co}$, $\mathcal{C}_{co,co}$, and $\mathcal{C}_{co,ic}$, respectively, from left to right. %2,4,3
In Fig. \ref{fig:ACEL_nocpt_bad}, the policy sets (also illustrated in Fig. \ref{fig:compliance-averseBR2}) illustrated by the contour plots on the $xy$-plane are sets $\mathcal{C}_{ic,co}$, $\mathcal{C}_{ic,ic}$, and $\mathcal{C}_{co,ic}$, respectively, from left to right. %2,1,3
These policy sets are convex as shown in Lemma \ref{lemma:convex}. The sets $\mathcal{C}_{ic,co}$ and $\mathcal{C}_{co,ic}$ are CT and CU, respectively.  
%\textcolor{black}{The results show that the volume of trustworthy policy set $\Pi_{pt}^{co}$ is larger under compliance-seeking insiders than compliance-averse insiders.} 

%    if bestResponse(1)==1 && bestResponse(2)==1
    %     fl=1; 
    % elseif bestResponse(1)==1 && bestResponse(2)==2
    %     fl=2; %fully obident
    % elseif bestResponse(1)==2 && bestResponse(2)==1    
    %     fl=3; %fully rebel
    % elseif bestResponse(1)==2 && bestResponse(2)==2   
    %     fl=4;
    % end

Fig. \ref{fig:ACEL_nocpt} illustrates that an improper recommendation policy may lead to a negative ACEL, but the optimal ACEL represented by the red star is always non-negative, as shown in Section \ref{sec:Defender's Optimal Recommendation Policy}. 
%Moreover, improper recommendation policy mislead compliance-seeking insiders more than the compliance-averse insiders as the extreme value is -5, compared to -3. 
For compliance-seeking insiders, the defender's ISeL  $J_D(\pi_z, b_X,\bar{v}_D,\bar{v}_U)$ and the optimal ASeL $J_D(\pi^*, b_X,\bar{v}_D,\bar{v}_U)$ are both $1.8$. 
For compliance-averse insiders, the defender's ISeL  $J_D(\pi_z, b_X,\bar{v}_D,\bar{v}_U)$ and the optimal ASeL $J_D(\pi^*, b_X,\bar{v}_D,\bar{v}_U)$ are $-0.64$ and $0.73$, respectively. 
%relative difference of $(0.73-(-0.64))/(\max(|0.73|,|-0.64|))=188\%$. 
%\cite{tornqvist1985should}. 
%percentage points. 
%https://en.m.wikipedia.org/wiki/Relative_change_and_difference
\begin{remark}[\textbf{Adaptivity and Structural Improvement}]
The above results show that ZETAR can well adapt to insiders with different compliance attitudes and achieve a structural improvement of compliance (from a negative ISeL to a positive ASeL) for compliance-averse insiders.
\end{remark}

%The red star represents the optimal ACEL, and the contour plots illustrate the CT and CU policy sets (also shown in Fig. \ref{fig:BRset}).  
%In Fig. \ref{fig:ACEL_nocpt_good}, 
% The defender's prior utility $J_D(\pi_z, b_X,\bar{v}_D,\bar{v}_U)$ is $1.8$, $-0.64$, $1.8$, and $1.8$ in Fig. \ref{fig:ACEL_nocpt_good}, \ref{fig:ACEL_nocpt_bad}, \ref{fig:EPU_CPTgamma1_good}, and \ref{fig:EPU_CPTgamma2_good}, respectively. 
% The defender's optimal EPU $J_D(\pi^*, b_X,\bar{v}_D,\bar{v}_U)$ is $1.8$, $0.7244$, $1.8$, and $1.8$ in Fig. \ref{fig:ACEL_nocpt_good}, \ref{fig:ACEL_nocpt_bad}, \ref{fig:EPU_CPTgamma1_good}, and \ref{fig:EPU_CPTgamma2_good}, respectively. 

%The plots in Fig. \ref{fig:ACEL_nocpt} and \ref{fig:EPU_CPTgamma_good} collaborate Lemma \ref{lemma:convex}, Lemma \ref{lemma:PWLC of the insider's EPU}, and Corollary \ref{corollary:Centrosymmetry}. 
%They also show that ZETAR can well adapt to insiders of different compliance attitudes and risk-attitudes. 
%\vspace{-6mm}
\begin{figure}[h]
    \centering % <--a
    \begin{subfigure}{0.22\textwidth}
\includegraphics[width=1\columnwidth]{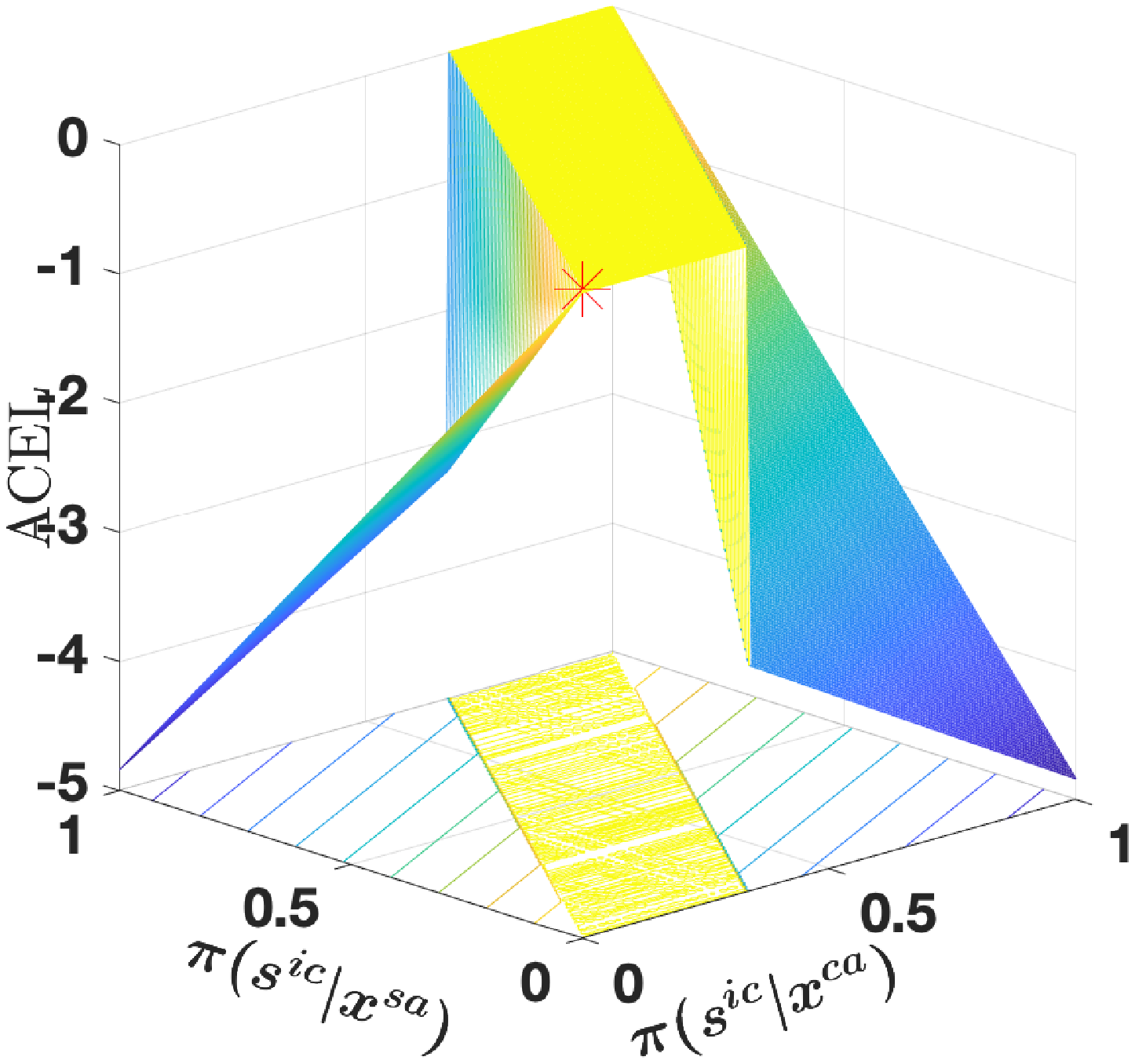}
\caption{ 
Compliance-seeking insiders. 
%ACEL=1.8-1.8
}
\label{fig:ACEL_nocpt_good}
\end{subfigure}\hfil 
%   \begin{subfigure}{0.3\textwidth} % <--b
% \includegraphics[width=1\columnwidth,height=.8\columnwidth]{fig/EPU_neutral_NoCPT_eta0.eps} %height=.5\columnwidth
% \caption{ 
% Compliance-neutral insiders. 
% }
% \label{fig:ACEL_nocpt_neutral}
% \end{subfigure}\hfil 
    \begin{subfigure}{0.22\textwidth} % <--c
\includegraphics[width=1\columnwidth]{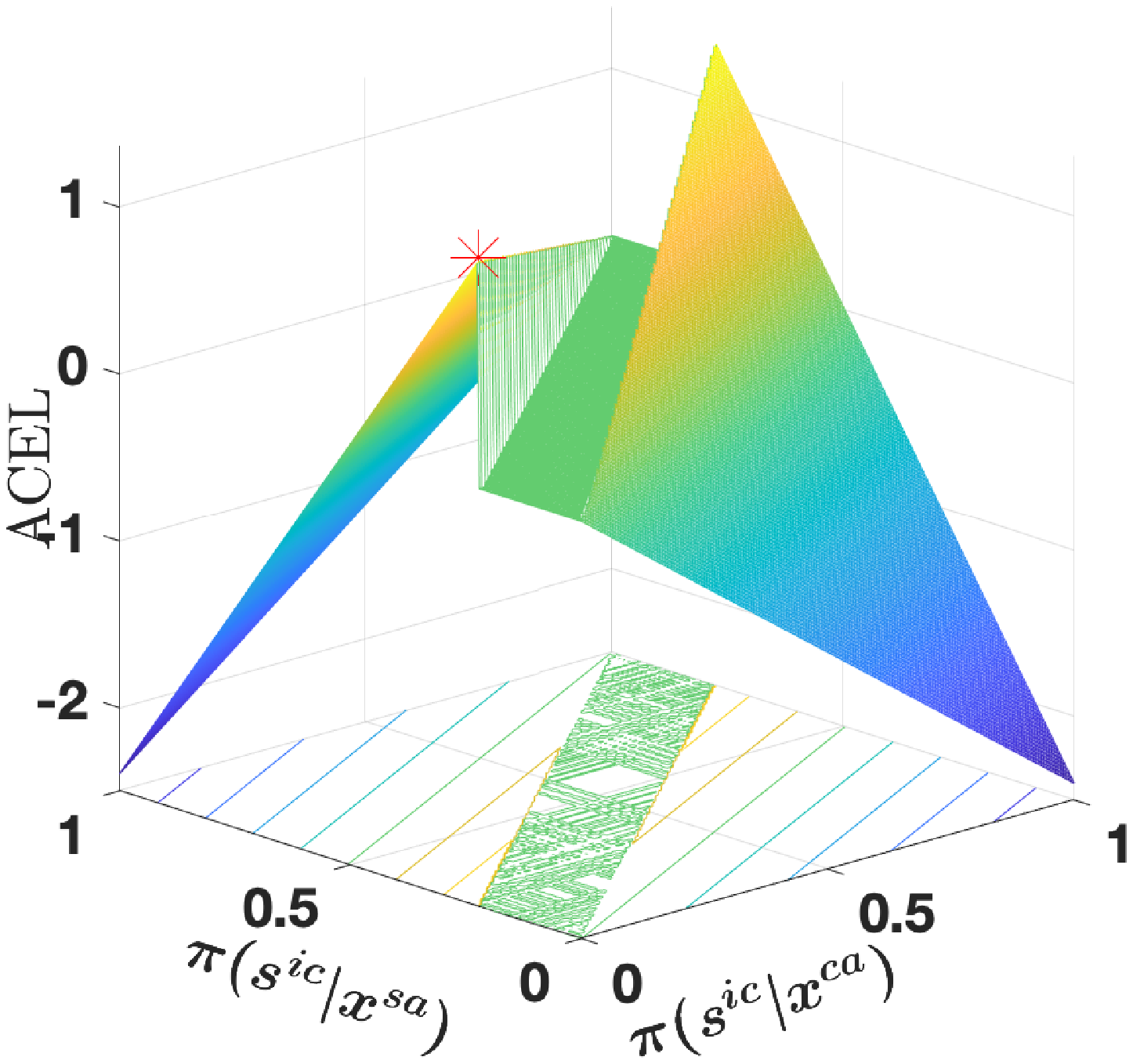} %height=.5\columnwidth
\caption{ 
Compliance-averse insiders. 
%ACEL=0.7244-(-0.64)
}
\label{fig:ACEL_nocpt_bad}
\end{subfigure}\hfil 
\caption{
ACEL $J_D^{acel}(\pi,b_X,\bar{v}_D,\bar{v}_U)$ versus $\pi(s^{ic}|x^{sa})\in [0,1]$ in $x$-axis and $\pi(s^{ic}|x^{ta})\in [0,1]$ in $y$-axis when $\gamma_d=\gamma_s=1$.  
%From left to right, it is CT set, PT, and CU set. 
\label{fig:ACEL_nocpt}
}
\end{figure}

\subsubsection{The Optimal ACEL}%under different bY and different types of insiders and CPT

% \begin{figure*}[h]
%     \centering % <--a
%     \begin{subfigure}{0.5\textwidth}
% \includegraphics[width=1\columnwidth]{fig/Self-interested_noCPT3by3subplot.eps} 
% \caption{ 
% Insiders of three compliance attitudes with $\gamma_s=\gamma_d=1$. 
% }
% \label{fig:3by3_noCPT}
% \end{subfigure}\hfil % <--b
%     \begin{subfigure}{0.5\textwidth}
% \includegraphics[width=1\columnwidth]{fig/Self-interested_CPT3by3subplot.eps} 
% %height=.5\columnwidth
% \caption{ 
% Compliance-neutral insiders with different CPT parameters. 
% }
% \label{fig:3by3_CPT}
% \end{subfigure}\hfil % <--b
% \caption{
% Utilities, the optimal ACEl, and the optimal recommendation policies in the first, second, and third rows, respectively, versus prior statistic $b_Y(y^{hr})\in [0,1]$. 
% The defender's prior utility $J_D(\pi_z,b_X,\bar{v}_D,\bar{v}_U)$, her EPU $J_D(\pi^*,b_X,\bar{v}_D,\bar{v}_U)$, and an insider's EPU $J_U(\pi^*,b_X,\bar{v}_U)$  in blue, red, and black, respectively. 
% Two elements of the optimal recommendation policy, $\pi^*(s^{ic}|x^{sa})$ and $\pi^*(s^{ic}|x^{ta})$, are illustrated in orange and pink, respectively. 
% The vertical dashed black lines represent the belief threshold $t_{bt}\in [0,1]$ in Section \ref{sec:casestudy_initial compliance}.
% \label{fig:3by3}
% }
% \end{figure*}
We illustrate the impacts of the optimal recommendation policy $\pi^*$ on the defender's and an insider's utilities under different likelihoods of the high-risk SP In Fig. \ref{fig:3by3}. 
%\vspace{-4mm}
\begin{figure}[h]
    \centering % <--a
\includegraphics[width=.99\columnwidth]{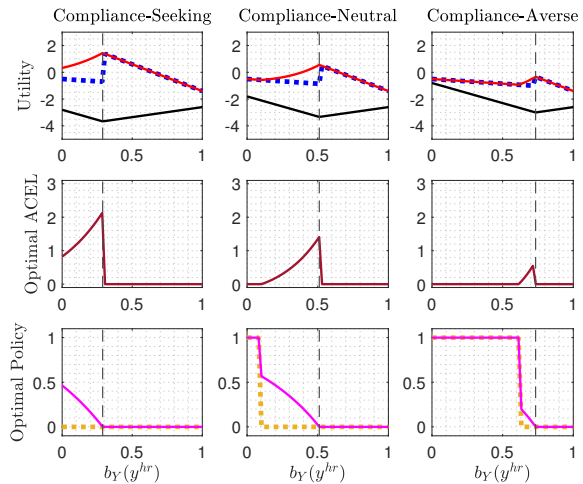} 
\label{fig:3by3_noCPT}
%\vspace{-4mm}
\caption{
Utilities, the optimal ACEl, and the optimal recommendation policies in the first, second, and third rows, respectively, versus prior statistic $b_Y(y^{hr})\in [0,1]$ concerning insiders with three compliance attitudes under $\gamma_s=\gamma_d=1$. 
The defender's ISeL $J_D(\pi_z,b_X,\bar{v}_D,\bar{v}_U)$, her optimal ASeL $J_D(\pi^*,b_X,\bar{v}_D,\bar{v}_U)$, and an insider's optimal ASaL $J_U(\pi^*,b_X,\bar{v}_U)$ are in blue, red, and black, respectively. 
Two elements of the optimal recommendation policy, $\pi^*(s^{ic}|x^{sa})$ and $\pi^*(s^{ic}|x^{ta})$, are illustrated in orange and pink, respectively. 
The vertical dashed black lines represent the belief threshold $t_{bt}\in [0,1]$. % in Section \ref{sec:casestudy_initial compliance}.
\label{fig:3by3}
}
\end{figure}
Following Section \ref{sec:casestudy_initial compliance}, the belief threshold $t_{bt}\in [0,1]$, represented by the vertical dashed black lines, divides the entire prior belief region into the compliant region  $b_Y(y^{hr})\in [t_{bt},1]$  on the right and non-compliant region  $b_Y(y^{hr})\in [0,t_{bt})$ on the left, where an insider takes $a^{co}$ and $a^{ic}$, respectively. 
Under the compliant regions, an insider tends to take compliant actions, resulting in zero ACEL and zero-information recommendation policy $\pi^*(s^{ic}|x^{sa})=\pi^*(s^{ic}|x^{ta})=0$. 
Under the non-compliant regions where an insider tends not to comply, the optimal recommendation policy induces positive ACEL. 
The defender's ISeL in compliant regions is larger than the one in non-compliant regions as shown by the blue lines in the two regions. 
%In Fig. \ref{fig:3by3_noCPT}, 
As an insider changes from being compliance-averse to compliance-seeking, his ASaL in black decreases in the non-compliant region, the belief threshold reduces (also illustrated in Fig. \ref{fig:initial compliance}), and the peak of the optimal ACEL increases. %maximum value 
The orange and pink lines illustrate that a large ACEL results from a more distinguished recommendation policy, i.e., a larger difference between $\pi^*(s^{ic}|x^{sa})$ and $\pi^*(s^{ic}|x^{ta})$.  
Moreover, the defender can recommend compliant actions, i.e., $s^{co}$, with a high probability as an insider changes from compliance-averse to compliance-seeking. 
%it shows the following recommendation design principles: under tolerant audit $x^{ta}$, the defender 
%In Fig. \ref{fig:3by3_CPT}, as $\gamma_s$ increases from $1$ to $2.55$ (i.e., an insider becomes more risk-averse) with $\gamma_d=1$, the belief threshold decreases, and the peak of the optimal ACEL increases. As $\gamma_d$ decrease from $1$ to $0.88$ (i.e., an increasing effect of diminishing sensitivity), the belief threshold 
Despite the linearity of the defender's ISeL in blue, her optimal ASeL in red and the optimal ACEL in brown are nonlinear in $b_Y$, as shown in Remark \ref{remark:nonlinear}. 
In Fig. \ref{fig:3by3}, we further observe that an insider’s ISaL coincides with his optimal ASaL, both represented by the black solid lines for all $b_Y(y^{hr})\in [0,1]$, which corroborates Proposition \ref{proposition:win-win}.

\section{Conclusion}
\label{sec:conclusion}
This work has developed ZETAR as a proactive framework to improve compliance of insiders with different incentives by zero-trust audits and recommendations. 
By a strategic and customized information disclosure of the audit scheme, the defender manages to influence an insider's incentives in favor of the corporate security objectives. 
We have formulated primal and dual convex problems with different levels of recommendation customization to provide a unified computational framework for ZETAR. 
We have shown its strong duality and degeneration to linear programs with fully customized recommendation policies. 
The dual problem has offered an interpretation of ZETAR from an insider's perspective; i.e., each insider aims to minimize his effort to satisfy the security objective of the corporate network. 

%Considering fully customized recommendation policies, 
We have characterized the structure of trustworthy recommendation policies and compliance status under malicious, self-interested, and amenable insiders. 
%Amendable insiders are naturally willing to follow security rules and have a strong sense of responsibility to enhance security. 
%Malicious insiders, e.g., the insider whose computer has been hacked, use their inside access privilege for fraud, sabotage, espionage, revenge, and blackmail. 
%Others are self-interested and prioritize their convenience or benefits over corporate security. 
The characterizations have led to fundamental principles 
%(e.g., EPU non-compliance under policy permutation, preservation of compliance under linear incentive transformations, opposite trustworthiness under centrosymmetry, and compliance equivalency under permutation), conditions (e.g., prior indifference and prior non-dominance conditions), 
and information disclosure guidelines to  insiders. 
Leveraging zero-trust design principles, we have assumed no trust and knowledge of the insider's incentives and developed efficient feedback algorithms to learn the insider's incentive based on the audit result of his behaviors. 
After identifying the policy separability principle and characterizing the Completely Trustworthy (CT) policy sets determined by the insider's incentive as convex polytopes, we have adopted a binary search algorithm to learn the vertices of the polytope, which is guaranteed to achieve an accuracy of $\epsilon>0$ within $2^{n-1}n\log_2(1/\epsilon)$ steps. 

We have used a case study to corroborate that ZETAR enhances compliance for insiders with different extrinsic and intrinsic incentives.  
The results have shown that ZETAR can well adapt to insiders with  different risk and compliance attitudes and structurally improve the defender's average security level when interacting with compliance-averse insiders. 
%achieve a relative improvement of $188$ percentage points in the satisfaction of risk-averse insiders. 
Under binary sets of actions and SP states, we have shown that insiders adopt a threshold policy with belief threshold $t_{bt}\in [0,1]$ that divides the belief region into compliance region  $b_Y(y^{hr})\in [t_{bt},1]$ and non-compliant one  $b_Y(y^{hr})\in [0,t_{bt})$. 
Finally, CT recommendation policies have been shown to improve corporate network security without decreasing insiders' satisfaction level. 
%have observed that ZETAR introduces non-negative payoffs to both the defender and insiders for binary action sets. 

% % use section* for acknowledgment
% \section*{Acknowledgment}

% This work was supported in part by the National
% Science Foundation (NSF) under Grant ECCS-1847056, Grant BCS-2122060; in part by Shuimu Tsinghua Scholar Program 2022SM046; in part by International Postdoctoral Exchange Fellowship Program(Talent-Introduction Program) YJ20220128; and in part by National Natural Science Foundation of China No.62341109, No.62341106, Shanghai Municipal Science and Technology Major Project, and Tsinghua University Initiative Scientific Research Program.

% all $x\in \mathcal{X}, \theta\in\Theta, a\in\mathcal{A}$


\ifCLASSOPTIONcaptionsoff
  \newpage
\fi



%\bibliographystyle{IEEEtran}

\bibliographystyle{IEEEtran}
\bibliography{ZETAR}

\begin{IEEEbiography}[{\includegraphics[width=1in,height=1.25in,clip,keepaspectratio]{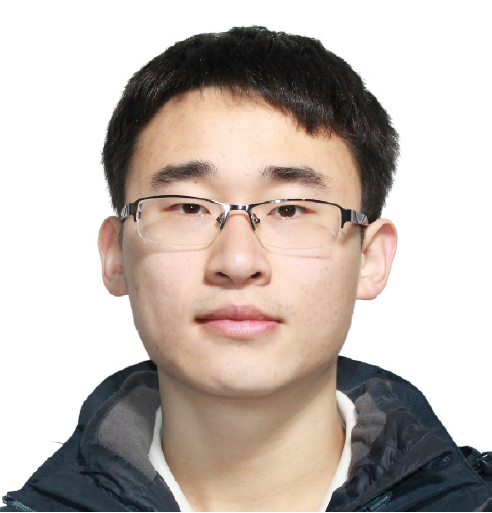}}]{Linan Huang} 
received the B.Eng. degree (Hons.) in Electrical Engineering from Beijing Institute of Technology, China, in 2016 and the Ph.D. degree in electrical engineering from New York University (NYU), Brooklyn, NY, USA, in 2022. 
He is currently an assistant researcher at Tsinghua University. 
%He is a recipient of many awards, including the 2022 Dante Youla Award for graduate research excellence, the best paper award at the 2022 GameSec conference, and the 2022 INFORMS MAS Koopman Prize. 
His research interests include dynamic decision-making in the multi-agent system, mechanism design, artificial intelligence, cybersecurity, and satellite networks. 
\end{IEEEbiography}

\begin{IEEEbiography}[{\includegraphics[width=1in,height=1.25in,clip,keepaspectratio]{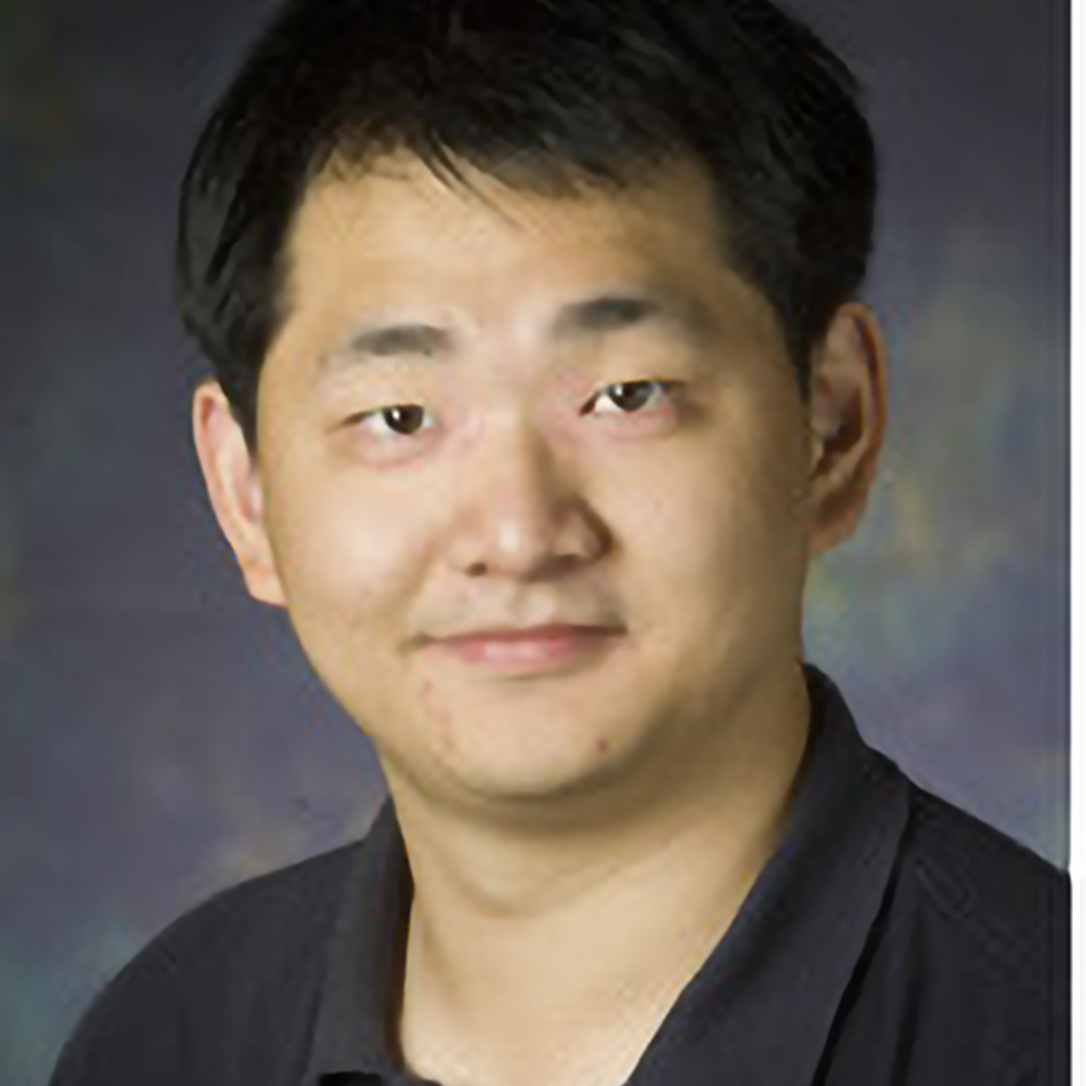}}]{Quanyan Zhu}
(SM’02-M’14) %IEEE student member 2002 and member 2014
received B. Eng. in Honors Electrical Engineering from McGill University in 2006, M. A. Sc. from the University of Toronto in 2008, and Ph.D. from the University of Illinois at Urbana-Champaign (UIUC) in 2013. 
After stints at Princeton University, he is currently an associate professor at the Department of Electrical and Computer Engineering, New York University (NYU). He is an affiliated faculty member of the Center for Urban Science and Progress (CUSP) and Center for Cyber Security (CCS) at NYU. His current research interests include game theory, machine learning, cyber deception, and cyber-physical systems.
\end{IEEEbiography}



%\vfill


\end{document}